\documentclass[a4paper,12pt]{article}
\usepackage[utf8]{inputenc}
\usepackage[top=1in, bottom=1in, left=1in, right=1in]{geometry}
\usepackage{amsmath}
\usepackage{graphicx}
\usepackage{enumerate}
\usepackage{natbib}

\usepackage{color}
\usepackage{amssymb}
\usepackage{amsthm}
\usepackage{graphics}
\usepackage{epstopdf}
\usepackage{epsfig}
\usepackage{comment}
\usepackage{multirow} 
\usepackage{float}
\usepackage{stix}
\usepackage{enumitem}
\usepackage{xr}
\usepackage{hyperref}

\newtheorem{theorem}{Theorem}
\newtheorem{proposition}{Proposition}
\newtheorem{lemma}{Lemma}
\newtheorem{corollary}{Corollary}

\newtheorem{assumption}{Assumption}

\newcommand{\bZ}{\mathbf{Z}} 
\newcommand{\bz}{\mathbf{z}} 
\newcommand{\bX}{\mathbf{X}}
\newcommand{\bx}{\mathbf{x}} 
\newcommand{\bB}{\mathbf{B}} 
\newcommand{\bb}{\mathbf{b}} 
\newcommand{\bw}{\mathbf{w}} 
\newcommand{\bg}{\mathbf{g}} 
\newcommand{\bH}{\mathbf{H}} 
\newcommand{\bq}{\mathbf{q}}

\newcommand{\bgamma}{\boldsymbol{\gamma}} 
\newcommand{\blambda}{\boldsymbol{\lambda}} 
\newcommand{\btheta}{\boldsymbol{\theta}} 
\newcommand{\balpha}{\boldsymbol{\alpha}} 
\newcommand{\EE}{\mathbb{E}} 
\newcommand{\PP}{\mathbb{P}} 
 
\newcommand{\cL}{\mathcal{L}} 
\newcommand{\cZ}{\mathcal{Z}} 
\newcommand{\bKK}{[K]_{K'}} 
\newcommand{\cK}{\mathcal{K}} 
\newcommand{\cN}{\mathcal{N}} 

\renewcommand{\star}{{\usefont{OML}{cmm}{m}{it}\symbol{63}}}
\def\ind{\begin{picture}(9,8)
		\put(0,0){\line(1,0){9}}
		\put(3,0){\line(0,1){8}}
		\put(6,0){\line(0,1){8}}
	\end{picture}
}

\begin{document}

	\def\spacingset#1{\renewcommand{\baselinestretch}%
		{#1}\small\normalsize} \spacingset{1}


		\title{\bf Balancing Weights for Causal Inference in Observational Factorial Studies}
		\author{Ruoqi Yu\\
			Department of Statistics, University of Illinois Urbana-Champaign\\
			and \\
			Peng Ding\\
			Department of Statistics, University of California, Berkeley}
		\maketitle

	\bigskip
	\begin{abstract}
		Many scientific questions in biomedical, environmental, and psychological research involve understanding the effects of multiple factors on outcomes. While factorial experiments are ideal for this purpose, 
		randomized controlled treatment assignment is generally infeasible in many empirical studies. Therefore, investigators must rely on observational data, where drawing reliable causal inferences for multiple factors remains challenging. As the number of treatment combinations grows exponentially with the number of factors, some treatment combinations can be rare or missing by chance in observed data, further complicating factorial effects estimation.
		To address these challenges, we propose a novel weighting method tailored to observational studies with multiple factors. Our approach uses weighted observational data to emulate a randomized factorial experiment, enabling simultaneous estimation of the effects of multiple factors and their interactions. Our investigations reveal a crucial nuance: achieving balance among covariates, as in single-factor scenarios, is necessary but insufficient for unbiasedly estimating factorial effects; balancing the factors is also essential in multi-factor settings. Moreover, we extend our weighting method to handle missing treatment combinations in observed data. Finally, we study the asymptotic behavior of the new weighting estimators and propose a consistent variance estimator, providing reliable inferences on factorial effects in observational studies.
		
	\end{abstract}
	
	\noindent%
	\vfill
	
	\newpage

\spacingset{1.5} 

	\section{Introduction}
	Assessing the effects of multiple factors is crucial in various scientific fields, such as biomedical research, environmental sciences, and psychology, as it helps shape decision-making and policy development. For instance, \cite{berenson1998association} used the Bogalusa Heart Study to investigate the effect of multiple cardiovascular risk factors on atherosclerosis in children and young adults, and \cite{rillig2019role} studied how ten global change factors affect soil properties, soil processes, and microbial biodiversity. 
	A simple and common approach to evaluating the effects of multiple factors is to consider one factor at a time. However, this method requires considerable time and resources to estimate all effects of interest with the same precision, as investigators need to repeat the procedure multiple times. Moreover, evaluating one factor at a time limits our understanding of how the effect of one factor depends on other factors, 
	which is critical in complex systems. For example, effective treatments of diseases such as sepsis, dementia, and stroke require combining treatments to target multiple components in cellular pathways \citep{berry2015platform}. 
	
	To understand both the main effects and interactions of multiple factors, a powerful tool is the randomized factorial experiment involving random assignments of all possible treatment combinations to units 
	\citep{wu2011experiments}. 
	Extensive research has been conducted on randomized factorial experiments 
	\citep{box1978statistics,dasgupta2015causal,dong2015using,branson2016improving,lu2016covariate,mukerjee2018using,zhao2018randomization,egami2018causal,kuhn2019simulation,li2020rerandomization,zhao2022regression,pashley2022causal,zhao2023covariate}.
	Although randomized experiments are considered the gold standard for estimating treatment effects, they are not always feasible due to ethical or practical considerations, forcing investigators to rely on observational data. However, the non-randomized treatment assignment mechanism in observational studies poses challenges for drawing reliable causal inferences, necessitating the removal of confounding due to observed covariates. 
	Therefore, it is essential to explore how we can design observational studies (without examining any outcomes of interest) to emulate a randomized factorial experiment to estimate the effects of multiple factors and their interactions simultaneously, which is key to ensuring the objectivity of causal conclusions \citep{rubin2008objective}.
	
	Weighting is a popular approach for designing observational studies to emulate randomized experiments by removing confounding. 
	Weighting aims to assign new weights to the individuals so that the weighted treatment groups are comparable in observed covariate distributions. 
	Much of the literature on weighting has focused on one binary treatment. Widely used methods include inverse-probability weighting 
	\citep{robins2000marginal,hirano2001estimation,hirano2003efficient} and balancing weighting 
	\citep{hainmueller2012entropy,zubizarreta2015stable,chan2016globally,li2018balancing,wong2018kernel,zhao2019covariate,wang2020minimal,bruns2022outcome,cohn2023balancing,fan2023optimal,yu2020treatment}.
	A straightforward approach to generalize the weighting methods for one factor to multiple factors is to view the treatment combination as a single treatment with multiple levels, where researchers have made progress in recent years 
	\citep{mccaffrey2013tutorial,fong2014covariate,li2019propensity,de2020direct,shu2023robust}. 
	For a review of the relevant literature, see \cite{lopez2017estimation}. By creating similar treatment combination groups, we can estimate all factorial effects of these factors, including all high-order interactions. 
	This idea works well when the number of factors is relatively small and all treatment combinations are not rare. However, challenges arise with many factors or some rare treatment combinations. Firstly, this approach can be computationally challenging due to the exponential growth of the number of treatment combinations as the number of factors increases.
	Moreover, when the number of factors is large, some treatment combinations are likely to be rare or empty, making this approach infeasible. 
	When estimating all factorial effects is impossible, can we still draw inferences for the relatively important ones?
	
	To answer this question, it is crucial to first determine which factorial effects are essential to estimate. Generally, low-order effects are considered more important than high-order effects, because the high-order effects are often believed to be smaller.  
	This is known as the \textit{effect hierarchy principle}. 
	Additionally, factorial designs typically assume the \textit{effect sparsity principle}, meaning that only a small number of relatively important effects are present \citep{box1986analysis}. Thus, it is generally reasonable for investigators to assume that high-order interactions are negligible (i.e., essentially zero) and focus on estimating main effects and low-order interactions. 
	For more discussions on these principles, see \cite{wu2011experiments}. 
	Although these effect principles originated from agricultural and industrial experiments, they are applicable to many other research fields. Still, investigators should validate their relevance using domain-specific knowledge when applying them to observational factorial studies. 
	With these guiding principles, the next question is: how can we design an observational study to estimate the main effects and low-order interactions?

	Towards this goal, we introduce a novel weighting approach designed explicitly for observational studies with multiple factors, as presented in \S\ref{sec:method}. By employing a weighted sample to emulate a randomized factorial experiment, we can estimate the effects of multiple factors and their interactions simultaneously using the same set of weights without the need to design observational studies repeatedly for different factorial effects. Similar to the traditional balance weighting methods for a single binary treatment, balancing the observed covariates is necessary to obtain unbiased estimates of factorial effects. However, our investigations suggest that researchers must also balance the factors in the presence of multiple factors. We also study the asymptotic behavior of the new weighting estimators and propose a consistent variance estimator, which enables downstream statistical inference.
	The proposed weighting method easily generalizes to design incomplete observational factorial studies with incompletely observed treatment combinations, as discussed in \S\ref{sec:extension}. This extension is important when some treatment combination groups are empty, making it possible to deal with a large number of factors. Finally, we evaluate the performance of the proposed method via simulation in \S\ref{sec:sim} and apply it 
	to estimate the factorial effects of four volatile organic compounds on tachycardia using an observational study in \S\ref{sec:data}.
	

	\section{Problem Setup and Weighting Framework} \label{sec:basic}
	\subsection{Problem Setting and Notation}\label{sec:notation}
	Suppose we have $K\geq2$ binary factors denoted by $z_k\in\{-1,+1\}$ for $k=1,...,K$. Then there are $Q=2^K$ possible treatment combinations $\mathbf{z}=(z_1,\dots,z_K)$. 
	Let $\mathcal{Z}=\{-1,+1\}^K$ denote the set of possible treatment combinations. 
	Our goal is to estimate the effects of multiple factors, including their interactions, under a full factorial design. 
	To define the causal estimands, we follow the potential outcome framework for factorial designs proposed by \cite{dasgupta2015causal}. 
	We denote the potential outcome if assigned to treatment combination $\mathbf{z}\in\mathcal{Z}$ as $Y(\mathbf{z})$. 
	Each individual has $Q$ potential outcomes; we denote the column vector of potential outcomes as $\mathbf{Y}=\left(Y(\mathbf{z})\right)_{\mathbf{z}\in\mathcal{Z}}$. Also, let $\EE[\mathbf{Y}]=\left(\EE[Y(\mathbf{z})]\right)_{\mathbf{z}\in\mathcal{Z}}$ denote the corresponding column vector of expectations.

	We define the causal estimands as in a randomized factorial experiment with a balanced design, where the treatment assignment is not related to any individual characteristics and each unit has an equal probability of receiving any treatment combination. Specifically, the main effect of each treatment $z_k$, $\tau_k$, is defined as the comparison between the average potential outcomes of receiving treatment $z_k$  and the average potential outcomes of not receiving it, uniformly averaged over all treatment combinations of other factors. More precisely, we can write $\tau_k$ as a contrast of the expected potential outcomes $\EE[\mathbf{Y}]$:
	$$\tau_k=\frac{1}{2^{K-1}}\mathbf{g}_k^\textup{T}\EE[\mathbf{Y}],$$
	where the contrast vector $\mathbf{g}_k=(g_{k\mathbf{z}})_{\mathbf{z}\in\mathcal{Z}}$ is a $2^K$-dimensional vector with half $+1$s and half $-1$s, indicating whether the treatment combination $\mathbf{z}$ has $z_k=+1$ or $z_k=-1$. 
	Additionally, the interaction effect of two factors $z_k$ and $z_{k'}$, denoted as $\tau_{k,k'}$, 
	measures the difference between the average effect of receiving $z_k$ 
	versus not receiving it 
	when $z_{k'}$ is received ($z_{k'}=+1$)
	or not ($z_{k'}=-1$), uniformly averaged over all treatment combinations of the other factors. 
	Following \cite{dasgupta2015causal},
	we define the corresponding contrast vector as $\mathbf{g}_{k,k'}=\mathbf{g}_k\circ \mathbf{g}_{k'},$ 
	where $\circ$ denotes the component-wise product. Then, 
	this interaction effect can be written as $$\tau_{k,k'}=\frac{1}{2^{K-1}}\mathbf{g}_{k,k'}^\textup{T}\EE[\mathbf{Y}].$$ 
	We can define higher-order interactions in a similar way, for which the contrast vector equals to the component-wise product of the corresponding contrast vectors for the main effects.
	In general, with $[K] = \{1,\dots,K\}$ denoting the set of indices for all factors, let $\tau_{\cK}$ denote the factorial effect for treatments in ${\cK}\subseteq[K]$ and let $\mathbf{g}_{\cK}=(g_{{\cK}\mathbf{z}})_{\mathbf{z}\in\mathcal{Z}}$ denote the corresponding contrast vector.
	
	To estimate the factorial effects, suppose there are $N$ subjects in the observational study, with $N_\mathbf{z}$ individuals with treatment combination $\mathbf{z}$. 
	We assume $(\mathbf{Z}_i, \mathbf{X}_i, \mathbf{Y}_i)$, $i=1,\dots,N$, are independent and identically distributed, where $\mathbf{Z}_i\in\mathcal{Z}$ denotes the $K$-dimensional received factors, $\mathbf{X}_i=(X_{i1},\dots,X_{iD})$ denotes the $D$-dimensional covariate with density $f(\mathbf{X})$, and $\mathbf{Y}_i$ denotes the $Q$-dimensional column vector of potential outcomes for unit $i$. 
	Additionally, we assume the \textit{ignorability} of the treatment assignment conditional on the observed covariates $\mathbf{X}$: $0<\PP( \mathbf{Z}_i =\bz\mid\bX_i)<1$ and $Y_i(\mathbf{z})\ \ind \ \bZ_i \mid\mathbf{X}_i$ for all $\mathbf{z}\in\mathcal{Z}$ \citep{rosenbaum1983central}.
	Since the covariate distribution may differ across treatment groups in an observational study, let $f_\mathbf{z}(\mathbf{X})=f(\mathbf{X}\mid \mathbf{Z}=\mathbf{z})$ denote the covariate distribution in group $\mathbf{z}$. 
	We can write the observed outcome for unit $i$ as $Y_i^\textup{obs}=Y_i(\bZ_i)=\sum_{\mathbf{z}\in\mathcal{Z}}I(\mathbf{Z}_i=\mathbf{z})Y_i(\mathbf{z})$, 
	where $I(\mathbf{Z}_i=\mathbf{z})$ is an indicator for whether unit $i$ received treatment combination $\mathbf{z}$. 
	
	Although each individual has $Q$ potential outcomes, we cannot observe the potential outcome $Y_i(\mathbf{z})$ if $\mathbf{Z}_i\neq\mathbf{z}$. Moreover, the treatment assignment is not randomized, and may depend on observed covariates $\mathbf{X}_i$. With these challenges, how can we estimate all factorial effects of interest? 
	
	\subsection{Weighting for Observational Factorial Studies}

	To estimate the treatment effects of multiple factors and their interactions under a factorial structure, we use the weighting framework, which is a popular tool to deal with one binary factor ($K=1$ in our setting). As discussed in Section~\ref{sec:notation}, each factorial effect $\tau_{\cK}$ can be written as a contrast of all expected potential outcomes, with the corresponding contrast vector $\mathbf{g}_{\cK}=(g_{{\cK}\mathbf{z}})_{\mathbf{z}\in\mathcal{Z}}$. 
	Following \cite{han2021contrast}, we write $g_{{\cK}\mathbf{z}}$, the contrast coefficient of the expected potential outcome under treatment combination $\mathbf{z}$, as $g_{{\cK}\mathbf{z}}=g_{{\cK}\mathbf{z}}^+-g_{{\cK}\mathbf{z}}^-$, where $g_{{\cK}\mathbf{z}}^+=\max(g_{{\cK}\mathbf{z}},0)$ and $g_{{\cK}\mathbf{z}}^-=\max(-g_{{\cK}\mathbf{z}},0)$. 
	We decompose the factorial effect $\tau_{\cK}$ into two parts $\tau_{\cK}=\tau^+_{\cK}-\tau^-_{\cK}$, where
	$$\tau^+_{\cK}=\frac{1}{2^{K-1}}\sum_{\mathbf{z}\in\mathcal{Z}}g_{{\cK}\mathbf{z}}^+\mathbb{E}[Y(\mathbf{z})]\quad \textrm{and} \quad \tau^-_{\cK}=\frac{1}{2^{K-1}}\sum_{\mathbf{z}\in\mathcal{Z}}g_{{\cK}\mathbf{z}}^-\mathbb{E}[Y(\mathbf{z})]$$
	correspond to the average expected potential outcomes of the positive and negative parts, respectively.
	Estimating $\tau^+_{\cK}$ and $\tau^-_{\cK}$ follows the same logic, so by symmetry, we mainly focus on estimating $\tau^+_{\cK}$ in the following discussion.
	
	Under ignorability, we have 
	\begin{align*}
		\tau^+_{\cK}=\frac{1}{2^{K-1}}\sum_{\mathbf{z}\in\mathcal{Z}}g_{{\cK}\mathbf{z}}^+\int\EE[Y(\mathbf{z})\mid\mathbf{X}]f(\mathbf{X})\textup{d}\mathbf{X}
		=\frac{1}{2^{K-1}}\sum_{\mathbf{z}\in\mathcal{Z}}g_{{\cK}\mathbf{z}}^+\int\EE[Y\mid\bZ=\bz,\mathbf{X}]f(\mathbf{X})\textup{d}\mathbf{X}.
	\end{align*}
	Estimating the target parameter $\tau^+_{\cK}$	is challenging because we do not observe $Y(\mathbf{z})$ for individuals from the population covariate distribution $f(\mathbf{X})$. With $N_\mathbf{z}$ individuals with outcome $Y(\mathbf{z})$ and covariate distribution $f_\mathbf{z}(\mathbf{X})=f(\mathbf{X}\mid\mathbf{Z}=\mathbf{z})$,  we introduce an additional term -- the weighting function $w_\mathbf{z}(\mathbf{X})=f(\bX)/f_{\bz}(\bX)$ -- so that 
	$$\tau^+_{\cK}
	=\frac{1}{2^{K-1}}\sum_{\mathbf{z}\in\mathcal{Z}}g_{{\cK}\mathbf{z}}^+\int\EE[Y\mid\mathbf{Z}=\mathbf{z},\mathbf{X}]w_\mathbf{z}(\mathbf{X})f_\mathbf{z}(\mathbf{X})\textup{d}\mathbf{X}.$$
	Consider the weighting estimator 
	\begin{equation}\label{eq:w}
		\widehat\tau^+_{\cK}=\frac{1}{2^{K-1}}\sum_{\mathbf{z}\in\mathcal{Z}}g_{{\cK}\mathbf{z}}^+\left[\frac{1}{N_\mathbf{z}}\sum_{i:\mathbf{Z}_i=\mathbf{z}}Y_i^\textup{obs}w_\mathbf{z}(\mathbf{X}_i)\right].
	\end{equation}
	Given the weights, the weighted average $\widehat\tau^+_{\cK}$ is the sample analogue of $\tau^+_{\cK}$.
	
	From equation~\eqref{eq:w}, we need to assign a weight of $w_{\mathbf{Z}_i}(\mathbf{X}_i)$ to the observed outcome of individual $i$, determined by the specific treatment combination they received. 
	Let $w_i=Nw_{\mathbf{Z}_i}(\mathbf{X}_i)/(2^{K-1}N_{\mathbf{Z}_i})$ denote the normalized weights, and let $A_{i{\cK}}^+=\sum_{\mathbf{z}\in\mathcal{Z}}g_{{\cK}\mathbf{z}}^+I(\mathbf{Z}_i=\mathbf{z})$ and $A_{i{\cK}}^-=\sum_{\mathbf{z}\in\mathcal{Z}}g_{{\cK}\mathbf{z}}^-I(\mathbf{Z}_i=\mathbf{z})=1-A_{i{\cK}}^+$ denote whether individual $i$ belongs to the positive or negative part of contrast $\mathbf{g}_{\cK}$, respectively.  We estimate the factorial effect $\tau_{\cK}$ with the weighting estimator 
	\begin{equation}\label{eq:weighting}
		\widehat\tau_{\cK}=\widehat\tau^+_{\cK}-\widehat\tau^-_{\cK}, \quad \textrm{ where }\widehat\tau_{\cK}^\Omega=\frac{1}{N}\sum_{i=1}^N w_i A_{i{\cK}}^\Omega Y_i^\textup{obs}, \textrm{ for } \Omega=+,-.
	\end{equation}
	How can we estimate the weights $w_i$ so that the same weights can be used to estimate multiple factorial effects simultaneously?
	
	\subsection{A Naive Weighting Method for Estimating All Factorial Effects}
	One straightforward approach to estimating the weights is to transform the problem with multiple factors into a problem with one multi-valued treatment by directly focusing on the treatment combinations. In doing so, the goal is to estimate $\EE[Y_i(\bz)]$ for all $\bz\in\cZ$. 
	
	A common choice of weights is based on the generalized propensity score 
	$\pi_{\bz}(\bX)=\PP(\bZ=\bz\mid\bX)$, where $\bz\in\cZ$ and $\sum_{\bz\in\cZ} \pi_{\bz}(\bX) =1$ \citep{imbens2000role,imai2004causal}. Specifically, the weights are chosen as the inverse probability of receiving the treatment combination, i.e., $w_i\propto 1/\pi_{\bZ_i}(\bX_i)$. 
	As the generalized propensity score is typically unknown, it needs to be estimated in practice, for example, using multinomial logistic regression. 
	However, this estimation process can lead to significant biases in factorial effects estimation  due to the misspecification of the generalized propensity score model \citep{kang2007demystifying}.
	In addition, a large number of factors and extreme estimated probabilities can result in unstable estimates of the treatment effects. 
	
	In recent years, balancing weighting methods have been proposed to improve the robustness of weighting methods \citep{fong2014covariate,li2019propensity,de2020direct,ben2023multilevel,shu2023robust}. Rather than estimating the generalized propensity score $\pi_{\bz}(\bX)$, these methods estimate weights 
	$w_i$ directly by comparing 
	some basis functions of the covariates \citep{fong2014covariate}, $h_s(\bX), s=1,\dots,S$, which is usually implemented as  $$\left|{1\over N}\sum_{i=1}^N w_iI(\bZ_i=\bz)h_s(\bX_i)-{1\over N}\sum_{i=1}^N h_s(\bX_i)\right|\le \Delta_s, \textrm{ for }\bz\in\cZ \textrm{ and } s=1,\ldots,S,$$
	where the $\Delta_s$'s are the tolerances of covariate imbalance.
	Ideally, to remove all biases in factorial effects estimation, we would set the tolerance levels $\Delta_s$ to zero, requiring exact balance among treatment combination groups. However, it may result in an infeasible optimization problem when the sample size is not large enough. Specifically,
	treatment combination sizes must satisfy the condition $N_{\bz}\geq S\textrm{ for all } \bz\in\cZ$ to ensure the existence of such weights.
	As the number of factors $K$ increases, it becomes harder to satisfy this condition, necessitating a larger sample size $N$ to ensure that all treatment combination groups have enough samples. 
	In cases where exact balance is unachievable due to the limited sample size, an alternative is to allow approximated balance by setting $\Delta_s>0$. 
	In such scenarios, choosing the tolerances $\Delta_s$ optimally is crucial in reducing the biases in factorial effects estimation while challenging in making the optimization problem feasible. 
	
	This idea of viewing the treatment combinations as a multi-valued treatment implicitly assumes we want to estimate all factorial effects. 
	When unbiased estimation of all factorial effects is impossible, we take a different perspective and consider the following relaxed research question: Can we still draw reliable inferences for the relatively important ones?
	
	In practice, researchers are often primarily interested in the main effects and low-order interactions up to order $K'$ and assume high-order interactions are negligible. For instance, a canonical choice is $K' = 2$. 
	Let $[K]_{k}=\{J\subseteq[K]: 0<|J|\leq k\}$ denote the set of all $k$th-order and lower-order treatment interaction index subsets. Therefore, we assume that the outcome depends on treatment combination $\mathbf{z}$ only through treatment interactions in $\bKK$.
	
	\begin{assumption}\label{as1}
		The interaction effects with orders greater than $K'$ are negligible, where $K'\leq K$.
	\end{assumption}
	
	Given the challenges discussed previously, how can we estimate the factorial effects of interest $\tau_{\cK}=\mathbf{g}_{\cK}^\textup{T}\EE[\mathbf{Y}]/2^{K-1}$ for all $\cK\in\bKK$?

	\section{A New Weighting Approach for Estimating Factorial Effects}\label{sec:method}
	In this section, we introduce a new weighting approach to estimating the treatment effects of multiple factors and their non-negligible interactions under a factorial structure. 
	As described in Section~\ref{sec:notation}, each factorial effect can be expressed as a contrast of all expected potential outcomes. To illustrate the main idea of our new weighting method, we first consider in Section~\ref{sec:single} the estimation of one factorial effect $\tau_{\cK}$ for treatments in ${\cK}\subseteq[K]$, with the corresponding contrast vector $\mathbf{g}_{\cK}=(g_{{\cK}\mathbf{z}})_{\mathbf{z}\in\mathcal{Z}}$. We discuss the simultaneous estimation of multiple factorial effects in Section~\ref{sec:full}. Besides factorial effect estimation, we also introduce a variance estimator for the effect estimation in Section~\ref{subsec:var}, so that we can conduct inferences for these factorial effects.
	
	\subsection{Weighting for Estimating a Single Factorial Effect}\label{sec:single}

	\subsubsection{General Additive Outcome Model}\label{sec:model:additive}

	To describe the new weighting method, we start with a general additive model for the expected potential outcomes, which is formally described in Assumption~\ref{asm1}. Specifically, we assume the expected potential outcomes are additive in the confounders $\mathbf{X}$ and treatment interactions up to order $K'\leq K$. 
	
	\begin{assumption}\label{asm1}
		The conditional expectation of potential outcomes follows a general additive outcome model: 
		$\mathbb{E}[Y_i(\mathbf{z})\mid\mathbf{X}_i]=\mu(\mathbf{X}_i)+\nu(\mathbf{z}),$ where $h_s$'s are prespecified basis functions, $\mu(\mathbf{X}_i)=\sum_{s=1}^S\alpha_s h_s(\mathbf{X}_i)$ for some unknown $\alpha_s$'s and $\nu(\mathbf{z})=\sum_{J\in\bKK}\beta_{J}\prod_{j\in J}z_j$ for some unknown $\beta_J$'s.
	\end{assumption}
	
	With this general additive outcome model in Assumption~\ref{asm1}, we can decompose the conditional bias of estimating $\tau_{\cK}$ into two components, 
	$\mathbb{E}[\widehat\tau^\Omega_{\cK}-\tau^\Omega_{\cK}\mid\{\mathbf{X}_i\}_{i=1}^N,\{\mathbf{Z}_i\}_{i=1}^N]$ for $\Omega=+,-$, where
	\begin{equation}\label{eq:add}
		\begin{split}
			\mathbb{E} \left[\widehat\tau^\Omega_{\cK}-\tau^\Omega_{\cK}\mid\{\mathbf{X}_i\}_{i=1}^N,\{\mathbf{Z}_i\}_{i=1}^N\right]
			&=\sum_{s=1}^S\alpha_s\left(\frac{1}{N}\sum_{i=1}^N w_iA_{i{\cK}}^\Omega h_s(\mathbf{X}_i)-\mathbb{E}[h_s(\mathbf{X}_i)]\right)\\
			&+\sum_{J\in\bKK}\beta_J\left(\frac{1}{N}\sum_{i=1}^N w_iA_{i{\cK}}^\Omega\prod_{j\in J}Z_{ij}-\frac{1}{2^{K-1}} \sum_{\mathbf{z}\in\mathcal{Z}}g_{{\cK}\mathbf{z}}^\Omega\prod_{j\in J}z_{j}\right).
		\end{split}
	\end{equation}
	For detailed derivations of \eqref{eq:add}, see \S S3.1 in the Supplementary Materials.
	The above analysis of bias decomposition offers two key insights. First, to control the bias, we need to balance the basis functions of the covariates in the positive and negative parts of a contrast. 
	Although we may not know $\mathbb{E}[h_s(\mathbf{X}_i)]$ in practice, it is natural to replace it with its sample analogue $\sum_{i=1}^N h_s(\mathbf{X}_i)/N$. This leads to the covariate balance constraints 
	$$\sum_{i=1}^N w_iA_{i{\cK}}^\Omega h_s(\mathbf{X}_i)=\sum_{i=1}^N h_s(\mathbf{X}_i), \quad \textrm{ for } \Omega=+,-  \textrm{ and } s=1,\dots,S.$$  This requirement is similar to that for traditional balancing weighting methods.
	However, it is not enough to obtain an unbiased estimator of $\tau_{\cK}$ with multiple factors. In addition to the  different covariate distributions, the treatment distributions may also differ for each contrast. Since the factors can affect the outcome through other channels besides the factorial effect under consideration, we also need to control for the confounding from other factors. This intuition aligns with 
	our bias decomposition above and suggests that we need to include the factors themselves and the non-negligible interactions as additional ``covariates'' for balancing such that 
	$$\frac{1}{N}\sum_{i=1}^N w_iA_{i{\cK}}^\Omega\prod_{j\in J}Z_{ij}=\frac{1}{2^{K-1}}\sum_{\mathbf{z}\in\mathcal{Z}} g_{{\cK}\mathbf{z}}^\Omega\prod_{j\in J}z_j, \quad \textrm{ for } \Omega=+,-  \textrm{ and }  J\in\bKK,$$ 
	to emulate the treatment assignment mechanism in a balanced factorial design. These additional balance constraints can be simplified as $$\sum_{i=1}^N w_iA_{i{\mathcal{K}}}^\Omega\prod_{j\in J}Z_{ij}=0 \textrm{, for }\Omega=+,- \textrm{ and } J\in[K]_{K'}, J\neq\mathcal{K}.$$
	A similar observation was mentioned in \cite{chattopadhyay2022implied} in the context of regression with multi-valued treatments. Notably, when there is only one factor (i.e., $K=1$), this set of balancing constraints on factors is equivalent to the balancing constraints for a constant covariate, and, therefore, does not need to be considered separately. This observation is consistent with the traditional balancing weighting methods for one binary factor. However, with $K\geq 2$, we have more balancing constraints to consider. In addition, the number of balancing constraints is an increasing function of the highest order of non-negligible interactions $K'$. The smaller $K'$ is, the easier the problem becomes.

	\subsubsection{Outcome Model with Treatment Effect Heterogeneity}\label{sec:model:interaction}
	The general additive outcome model in Assumption~\ref{asm1} 
	covers a broad class of functions by allowing a flexible form of basis functions. However, this model separates the effects of treatments and covariates and does not account for treatment effect heterogeneity, which is common in many applications. 
	To address this limitation, we consider a more general class of outcome models with treatment effect heterogeneity in Assumption~\ref{asm2}, by allowing the coefficient of the basis functions to depend on the treatment combination $\mathbf{z}$. This is a subtle but critical relaxation that can yield more robust and reliable causal conclusions. 
	\begin{assumption}\label{asm2}
		The conditional expectation of potential outcomes follows an outcome model with treatment effect heterogeneity: $\mathbb{E}[Y_i(\mathbf{z})\mid\mathbf{X}_i]=\sum_{J\in\bKK}\mu(\mathbf{X}_i; \mathbf{z}_J),$ 
		where $\mu(\mathbf{X}_i; \mathbf{z}_J)$ belongs to the span of prespecified basis functions $h_s(\mathbf{X}_i)$, $s=1,\dots,S$, i.e., $\mu(\mathbf{X}_i; \mathbf{z}_J)=\sum_{s=1}^S \alpha_{sJ}\left(\prod_{j\in J}z_j\right) h_{s}(\mathbf{X}_i)$.
		This characterization leads to 
		$\mathbb{E}[Y_i(\mathbf{z})\mid\mathbf{X}_i]=\sum_{s=1}^S \alpha_{s\mathbf{z}}h_{s}(\mathbf{X}_i)$, 
		where $\alpha_{s\mathbf{z}} =\sum_{J\in\bKK}\alpha_{sJ}\prod_{j\in J}z_j$ are unknown coefficients that depend on $\bz$.
	\end{assumption}
	Equivalently, we can write the heterogeneous treatment effect outcome model in Assumption~\ref{asm2} using new basis functions $q_{sJ}(\bX_i,\bz)=h_s(\mathbf{X}_i)\prod_{j\in J}z_{j}$ such that 
	\begin{equation}
		\mathbb{E}[Y_i(\mathbf{z})\mid\mathbf{X}_i]=\sum_{s=1}^S\sum_{J\in\bKK} \alpha_{sJ}q_{sJ}(\bX_i,\bz)=\balpha^\textup{T}\bq(\bX_i,\bz),
	\end{equation}
	where $\balpha=\{\alpha_{sJ}\}_{s=1,\dots,S,\ J\in\bKK}$ and $\bq(\bX_i,\bz)=\{q_{sJ}(\bX_i,\bz)\}_{s=1,\dots,S,\ J\in\bKK}$. Similar to working with the general additive outcome model in Assumption~\ref{asm1}, we analyze the bias decomposition to obtain the required balance constraints. 
	The decompositions
	\begin{align*}
		\mathbb{E}[\widehat\tau^\Omega_{\cK}-\tau^\Omega_{\cK}\mid\{\mathbf{X}_i\}_{i=1}^N,\{\mathbf{Z}_i\}_{i=1}^N]
		=&\balpha^\textup{T}\left[\frac{1}{N}\sum_{i=1}^N w_iA_{i{\cK}}^\Omega\bq(\bX_i,\bZ_i)
		-\frac{1}{2^{K-1}}\sum_{\mathbf{z}\in\mathcal{Z}}g_{{\cK}\mathbf{z}}^\Omega\mathbb{E}[\bq(\bX_i,\bz)]\right],  \textrm{ for } \Omega=+,-,
	\end{align*}
	suggest that to achieve unbiased estimates of the factorial effect $\tau_{\cK}$ under the heterogeneous treatment effect outcome model, we need to balance the new basis functions $q_{sJ}(\bX_i,\bZ_i)=h_s(\mathbf{X}_i)\prod_{j\in J}Z_{ij}$ for all $s=1,\dots, S,\ J\in\bKK$. 
	By replacing $\mathbb{E}[\bq(\bX_i,\bz)]$ with its sample analogue, we require 
	$$\sum_{i=1}^N w_iA_{i{\cK}}^\Omega\bq(\bX_i,\bZ_i)=\frac{1}{2^{K-1}}\sum_{\mathbf{z}\in\mathcal{Z}}g_{{\cK}\mathbf{z}}^\Omega\sum_{i=1}^N \bq(\bX_i,\bz), \quad \textrm{ for } \Omega=+,-.$$
	As in \S\ref{sec:model:additive}, the number of constraints $2S\left[{K\choose1}+{K\choose2}\dots+{K\choose K'}\right]$ grows as the highest interaction order $K'$ increases. The key difference due to the flexible form of the outcome model is that we must now balance the interactions between the treatment combinations and covariates, which results in more balance constraints. To highlight the difference between the general additive outcome model (Assumption~\ref{asm1}) and the heterogeneous treatment effect outcome model (Assumption~\ref{asm2}), we examine a simple scenario in which only the main effects are non-zero (i.e., $K'=1$) in \S S3.2 of the Supplementary Materials.
	
	\subsection{Weighting for Estimating Multiple Factorial Effects Simultaneously}\label{sec:full}
	Suppose the investigator believes interactions with order greater than $K'\leq K$ are negligible (Assumption~\ref{as1}). The goal is to estimate $Q_+={K\choose1}+{K\choose2}+\cdots+{K\choose K'}$ non-negligible factorial effects simultaneously. Specifically, we seek to estimate the factorial effect $\tau_{\cK}$ with contrast vector $\mathbf{g}_{\cK}$ for all ${\cK}\in\bKK=\{J\subseteq[K]: 0<|J|\leq K'\}$. 
	Applying the derivations from \S\ref{sec:model:interaction} to each of the $Q_+$ factorial effects of interest leads to the following balance constraints that are necessary to obtain unbiased factorial effect estimates under the outcome model with treatment effect heterogeneity:
	for all ${\cK}\in\bKK$ and $\Omega=+,-$, 
	$$\sum_{i=1}^N w_iA_{i{\cK}}^\Omega\bq(\bX_i,\bZ_i)=\frac{1}{2^{K-1}}\sum_{\mathbf{z}\in\mathcal{Z}}g_{{\cK}\mathbf{z}}^\Omega\sum_{i=1}^N \bq(\bX_i,\bz),$$ 
	where $\bq(\bX_i,\bZ_i)=\{q_{sJ}(\bX_i,\bZ_i)\}_{s=1,\dots,S,\ J\in\bKK}$ and $q_{sJ}(\bX_i,\bZ_i)=h_s(\mathbf{X}_i)\prod_{j\in J}Z_{ij}$. 
	
	Any weights that satisfy the balance constraints described above can provide unbiased estimates of all non-negligible factorial effects. Unlike most existing balancing weighting methods, which often allow both positive and negative weights, 
	we focus on non-negative weights since the ideal weighting function $w_\mathbf{z}(\mathbf{X})$ is the ratio of two density functions. Non-negative weights also avoid extrapolation and therefore make the results more interpretable. If there are several choices of non-negative weights $w_i\geq0$ that satisfy the above balance constraints, which weight should we use? One standard solution is to minimize a convex measure of the weights, $m(w_i)$, to select the optimal set of weights. Investigators can choose $m(\cdot)$ based on their needs and preferences. For instance, entropy balancing weights proposed by \cite{hainmueller2012entropy} use $m(x)=x\log x$; calibration weights in \cite{chan2016globally} use $m(x)=R(x,1)$ to avoid extreme weights, where $R(x,x_0)$ is a continuously differentiable, non-negative, and strictly convex function for any fixed $x_0\in\mathbb{R}$; \cite{de2020direct} use $m(x)=x^2$ to minimize the variance of the treatment effects estimation.

	In summary, to obtain the optimal factorial weights for estimating multiple factorial effects $\tau_\cK, \cK\in\bKK$, we solve the following optimization problem:
	\begin{equation}\label{prob: prime}
		\begin{aligned}
			\min_{w_i\geq0: i=1,\dots,N} \quad &\sum_{i=1}^N m(w_i)\\ 
			\text{subject to} \quad &\sum_{i=1}^N w_iA_{i{\cK}}^\Omega\bq(\bX_i,\bZ_i)=\frac{1}{2^{K-1}}\sum_{\mathbf{z}\in\mathcal{Z}}g_{{\cK}\mathbf{z}}^\Omega\sum_{i=1}^N \bq(\bX_i,\bz), \textrm{ for } \Omega=+,- \textrm{ and } \cK\in\bKK.
		\end{aligned}
	\end{equation}
	In general, problem~(\ref{prob: prime}) is feasible with high probability. See Proposition~S1 in the Supplementary Materials \S S4.1 for a formal statement and proof. 
	
	To efficiently solve the constrained convex optimization problem \eqref{prob: prime}, we consider its dual form as an unconstrained concave maximization problem. We first write the linear constraints of problem~(\ref{prob: prime}) as $\mathbf{Bw}=\bb$, where the matrix $\bB$ has $i$th column $\bB_{i}$ representing all balance criteria evaluated at individual $i$ and $\bb=\sum_{i=1}^N \bb_i$. 
	This leads to the dual problem 
	\begin{equation}\label{prob:dual}
		\max_{\boldsymbol{\gamma}\geq0,\boldsymbol{\lambda}} \cL^*(\boldsymbol{\lambda},\boldsymbol{\gamma}), 
	\end{equation}
	where 
	$\cL^*(\boldsymbol{\lambda},\boldsymbol{\gamma})=\sum_{i=1}^N\{\rho(\gamma_i-\boldsymbol{\lambda}^\textup{T}\mathbf{B}_{i})-\boldsymbol{\lambda}^\textup{T}\mathbf{b}_i\}$, with $\rho(v)=m((m')^{-1}(v))-v(m')^{-1}(v)$. 
	For the derivations of the form of the dual problem, see \S S3.3 in the Supplementary Materials. Since the dual problem~(\ref{prob:dual}) has no other constraints except for the non-negativity constraints of the Lagrangian multipliers $\boldsymbol{\gamma}$, it is easier to solve compared with the original problem~(\ref{prob: prime}). 
	
	Suppose the dual problem \eqref{prob:dual} has optimal solution $(\widehat{\blambda},\widehat{\bgamma})$. Then we can obtain the optimal factorial weights using the closed-form expression 
	$$\widehat w_i=(m')^{-1}(\widehat\gamma_i-\widehat{\blambda}^\textup{T}\mathbf{B}_{i}).$$ 
	Details of this derivation are provided in \S S3.3 in the Supplementary Materials.
	
	To control the variance of the factorial effects estimates, we follow \cite{zubizarreta2015stable} and  consider a special choice $m(w_i)=w_i^2$ in the objective function of the original problem~(\ref{prob: prime}), which gives the optimal factorial weights 
	$\widehat w_i=(\widehat\gamma_i-\widehat{\blambda}^\textup{T}\mathbf{B}_{i})/2.$ 
	When $m(x)=x^2$, the dual problem becomes
	\begin{equation}\label{prob:dual2}
		\max_{\boldsymbol{\gamma}\geq0,\boldsymbol{\lambda}} \cL^*(\boldsymbol{\lambda},\boldsymbol{\gamma})=\sum_{i=1}^N\left\{-\frac{1}{4}(\gamma_i-\boldsymbol{\lambda}^\textup{T}\mathbf{B}_{i})^2-\boldsymbol{\lambda}^\textup{T}\mathbf{b}_i\right\}.
	\end{equation}
	Finding the optimal solution $(\widehat{\blambda},\widehat{\bgamma})$ of the dual problem can be computationally expensive, since the number of decision variables in problem~(\ref{prob:dual2}) increases linearly with the sample size $N$. To improve the computational efficiency, we can simplify the dual problem~(\ref{prob:dual2}) further by exploiting the non-negativity conditions of $\widehat{\bgamma}$. Specifically, the Karush-Kuhn-Tucker conditions 
	allow us to express $\widehat{\bgamma}$ explicitly as $\widehat\gamma_i=\widehat{\blambda}^\textup{T}\mathbf{B}_{i}I\left(\widehat{\blambda}^\textup{T}\mathbf{B}_{i}\geq0\right).$
	As a result, we can reformulate the dual problem~(\ref{prob:dual2}) as the following non-constrained optimization problem:
	\begin{equation}\label{prob:dual2simple}
		\max_{\boldsymbol{\lambda}} \cL^{**}(\boldsymbol{\lambda})
		=\sum_{i=1}^N\left\{-\frac{1}{4}(\boldsymbol{\lambda}^\textup{T}\mathbf{B}_{i})^2I(\boldsymbol{\lambda}^\textup{T}\mathbf{B}_{i}<0)-\boldsymbol{\lambda}^\textup{T}\mathbf{b}_i\right\}.
	\end{equation}
	The simplified problem in \eqref{prob:dual2simple} can be efficiently solved using numerical optimization methods such as Newton's method and gradient methods. 
	
	Weights obtained by solving original problem~\eqref{prob: prime} or its dual problem~ \eqref{prob:dual2simple} yield $\sqrt{N}$-consistent estimates for the factorial effects under mild regularity conditions. 
	The formal results are presented in Theorem~\ref{thm:rootn} below; see \S S4.3 in the Supplementary Materials for the proof.
	
	\begin{theorem}\label{thm:rootn}
		Under Assumptions~\ref{as1}, \ref{asm2} and S1, the weighting estimator $\widehat\tau_{\cK}$ with weights $\widehat{\mathbf{w}}$ solving problem~(\ref{prob: prime}) with $m(x)=x^2$ is asymptotically normal for any $\cK\in\bKK$, i.e., $$\sqrt{N}\left(\widehat\tau_{\cK}-\tau_{\cK}\right)\rightarrow \cN\left(0,{\sigma}^{*2}_{\cK}\right),\  {\rm with }\  {\sigma}^{*2}_{\cK}=\sigma^2_{\cK}+\sigma^2_\epsilon,$$
		where
		$$
		\sigma^2_{\cK}=\frac{1}{2^{K}}{\rm Var}\left(\sum_{\mathbf{z}\in\mathcal{Z}}g_{{\cK}\mathbf{z}}\mathbb{E}[Y_i(\mathbf{z})\mid\mathbf{X}_i]\right)\qquad {\rm and}\qquad \sigma^2_{\epsilon }=\frac{\bar\sigma^2}{4}\blambda^{* \textup{T}}\EE\left[\bB_i\bB_i^\textup{T} I(\blambda^{*\textup{T}}\bB_i<0)\right]\blambda^*
		$$
		with
		$
		\blambda^*={\rm argmax}_{\blambda}\EE[-\frac{1}{4}(\blambda^\textup{T}\mathbf{B}_{i})^2I(\blambda^\textup{T}\mathbf{B}_{i}<0)-\boldsymbol{\lambda}^\textup{T}\mathbf{b}_i]
		$ and $\bar\sigma^2$ defined in Assumption~S1.
	\end{theorem}
	Based on Theorem~\ref{thm:rootn}, our weighting estimator $\widehat{\tau}_{\cK}$ is $\sqrt{N}$-consistent for the factorial effect $\tau_{\cK}$. Moreover, the variance $\sigma^{*2}_{\cK}$ depends on two components -- the variance $\sigma^2_{\cK}$ from the projection of potential outcomes to the covariate space and the variance $\sigma^2_{\epsilon}$ from the residuals of the projection. Similar variance decompositions occur in semiparametric efficiency theory \citep{robins1994estimation,hahn1998role,chan2016globally}. Notably, the same residual variance $\sigma^2_{\epsilon}$ is relevant for estimating all factorial effects, while the projection variance $\sigma^2_{\cK}$ varies according to the contrast vector $\bg_\cK$. 
	
	In addition, the weighting estimators for all non-negligible factorial effects follow a multivariate normal distribution asymptotically. This result, presented in Corollary~\ref{cor}, supports simultaneous inference of multiple factorial effects. The proof is straightforward and is hence omitted.
	\begin{corollary}\label{cor}
		Under Assumptions~\ref{as1}, \ref{asm2} and S1, the weighting estimators $\{\widehat\tau_{\cK}\}_{\cK\in\bKK}$ with weights $\widehat{\mathbf{w}}$ solving problem~(\ref{prob: prime}) with $m(x)=x^2$ satisfy $$\sqrt{N}\left(\{\widehat\tau_{\cK}\}_{\cK\in\bKK}-\{\tau_{\cK}\}_{\cK\in\bKK}\right)\rightarrow \cN\left(\mathbf{0},\Sigma^*\right),\  {\rm with }\  \Sigma^*=\Sigma+\Sigma_\epsilon,$$
		where the $(s,t)$th entry of \ $\Sigma$ equals
		$
		\Sigma_{st}=\frac{1}{2^{K}}{\rm Cov}\left(\sum_{\mathbf{z}\in\mathcal{Z}}g_{{\cK}_s\mathbf{z}}\mathbb{E}[Y_i(\mathbf{z})\mid\mathbf{X}_i],\sum_{\mathbf{z}\in\mathcal{Z}}g_{{\cK}_t\mathbf{z}}\mathbb{E}[Y_i(\mathbf{z})\mid\mathbf{X}_i]\right)$ and all entries of \ $\Sigma_{\epsilon }$ 
		equals $\sigma^2_{\epsilon }$.
	\end{corollary}
	
	While the theoretical results presented here employ a specific choice of $m(x)=x^2$, similar properties remain valid for other selections of $m(\cdot)$, provided that $m(\cdot)$ is strictly convex and continuously differentiable. For analogous prerequisites concerning the objective function, see \cite{chan2016globally}. Regardless of the chosen $m(\cdot)$ function that meets these criteria, we can establish asymptotic normality under certain regularity conditions. The primary distinction arises in the asymptotic variance, where the residual variance $\sigma^2_\epsilon$ adapts according to the choice of $m(\cdot)$.
	
	A common technique in the weighting literature to further improve the performance of a weighting method is to combine it with outcome regression adjustments \citep{robins1994estimation,rosenbaum2002covariance}. This type of augmented estimator is also applicable to our balancing weighting method in a factorial design. For more discussions, see \S S3.4 in the Supplementary Materials.
	
	Throughout this paper, we focus on the standard factorial effects and the idea applies to general contrasts, as defined in \cite{zhao2022regression}. 
	
	\subsection{Variance estimation}\label{subsec:var}
	To estimate the asymptotic variance of $\sqrt{N}(\widehat{\tau}_{\cK}-\tau_{\cK})$, we adopt a similar approach to that in \cite{chan2016globally}. Define a combined parameter  $\btheta=(\blambda^\textup{T},t)^\textup{T}$ and 
	$\eta_i(\btheta)=(\psi_i'(\blambda), \   w_i(A_{i{\cK}}^+-A_{i{\cK}}^-)Y_i^\textup{obs}-t)^\textup{T},$ 
	where $\psi_i'(\blambda)=-\frac{1}{2}\mathbf{B}_{i}\mathbf{B}_{i}^\textup{T}\blambda I(\boldsymbol{\lambda}^\textup{T}\mathbf{B}_{i}<0)-\mathbf{b}_i$ is the first derivative of $\psi_i(\blambda)=-\frac{1}{4}(\boldsymbol{\lambda}^\textup{T}\mathbf{B}_{i})^2I(\boldsymbol{\lambda}^\textup{T}\mathbf{B}_{i}<0)-\boldsymbol{\lambda}^\textup{T}\mathbf{b}_i$ with respect to $\blambda$
	and the weights $w_i=(\gamma_i-\lambda^\textup{T}\bB_i)/2=-\blambda^\textup{T}\bB_i I(\blambda^\textup{T}\bB_i<0)/2$ is also a function of $\blambda$.
	Recall that $\widehat{\blambda}={\rm argmax}_{\blambda} \frac{1}{N}\sum_{i=1}^N\psi_i(\blambda) \textrm{ and }  \widehat{\tau}_{\cK}=\frac{1}{N}\sum_{i=1}^N\widehat w_i A_{i{\cK}}^+Y_i^\textup{obs}-\frac{1}{N}\sum_{i=1}^N\widehat w_i A_{i{\cK}}^-Y_i^\textup{obs}.$
	Then  $\widehat{\btheta}=(\widehat{\blambda}^\textup{T},\widehat\tau_{\cK})^\textup{T}$ is the $Z$-estimator that satisfies the following estimating equation 
	$N^{-1}\sum_{i=1}^N \eta_i(\widehat{\btheta})=0.$
	In addition,  ${\btheta}^*=({\blambda}^{*\textup{T}},\tau_{\cK})^\textup{T}$ satisfies $ \EE[\eta_i(\btheta^*)]=0,$
	where ${\blambda}^*={\rm argmax}_{\blambda} \EE[\psi_i(\blambda)], \  w_i^*=-\blambda^{*\textup{T}}\bB_i I(\blambda^{*\textup{T}}\bB_i<0)/2, \textrm{ and } {\tau}_{\cK}=\EE[w_i^* (A_{i{\cK}}^+-A_{i{\cK}}^-)Y_i^\textup{obs}]=\frac{1}{2^{K-1}}\sum_{\mathbf{z}\in\mathcal{Z}}g_{{\cK}\mathbf{z}}\EE[Y_i(\mathbf{z})].$
	Therefore, the asymptotic variance of $\sqrt{N}(\widehat{\btheta}-\btheta^*)$ can be written as  $$\bH^{*-1}\EE\left[\eta_i({\btheta}^*)\eta_i({\btheta}^{*})^\textup{T}\right]\bH^{*-1}, \ {\rm with }\ \bH^*=\EE\left[\eta_i'({\btheta}^*)\right],$$
	where $\eta_i'({\btheta}^*)$ is the first derivative of $\eta_i({\btheta})$ evaluated at $\btheta=\btheta^*$. For details, see the proof of Theorem~\ref{thm:var} in \S S4.4 in the Supplementary Materials.
	
	Since we are only concerned with the asymptotic variance of $\sqrt{N}(\widehat{\tau}_{\cK}-\tau_{\cK})$, we focus on estimating the last element of $\bH^{*-1}\EE[\eta_i({\btheta}^*)\eta_i({\btheta}^{*})^\textup{T}]\bH^{*-1}$. Therefore, we only need to consider the last row of $\bH^{*-1}$. 
	Then, we have an estimator for  $\sigma^{*2}_{\cK}$: $$\widehat\sigma^{*2}_{\cK}=L^\textup{T}\left[\frac{1}{N}\sum_{i=1}^N\eta_i(\widehat{\btheta})\eta_i(\widehat{\btheta})^\textup{T}\right]L=\frac{1}{N}\sum_{i=1}^N\left(\eta_i(\widehat{\theta})^\textup{T}L\right)^2,$$
	where $$L=\left[\left(-\frac{1}{N}\sum_{i=1}^N\bB_i^\textup{T}(A_{i{\cK}}^+-A_{i{\cK}}^-)Y_i^\textup{obs}I\left(\widehat{\blambda}^\textup{T}\bB_i<0\right)/2\right)\left(-\frac{1}{N}\sum_{i=1}^N\bB_i\bB_i^\textup{T} I\left(\widehat{\blambda}^\textup{T}\bB_i<0\right)/2\right)^{-1}, -1\right]^\textup{T}$$ is a column vector with the same dimension as $\eta_i(\theta)$.
	\begin{theorem}\label{thm:var}
		Under Assumptions~\ref{as1}, \ref{asm2} and S1, we have 
		$\widehat\sigma^{*2}_{\cK}\overset{\PP}{\rightarrow}\sigma^{*2}_{\cK}.$
	\end{theorem} 
	Theorem~\ref{thm:var} suggests that $\widehat\sigma^{*2}_{\cK}$ is a consistent estimator of the asymptotic variance of $\sqrt{N}(\widehat{\tau}_{\cK}-\tau_{\cK})$. See \S S4.4 in the Supplementary Materials for the proof. Notably, the same variance estimation technique can be applied to provide a consistent estimator for the covariance matrix of multiple factorial effects. This can be achieved by considering a revised combined parameter  $\btheta=(\blambda^\textup{T},(t_\cK)_{\bKK}^\textup{T})^\textup{T}$ and working with the $Z$-estimator for $\eta_i(\btheta)=(
	\psi_i'(\blambda)^\textup{T}, (w_i(A_{i{\cK}}^+-A_{i{\cK}}^-)Y_i^\textup{obs}-t_\cK)_{\bKK}^\textup{T})^\textup{T}.$
	The detailed discussion is omitted due to the limited space.

	\section{Observational Factorial Studies with Incomplete Treatment Combinations}\label{sec:extension}
	In the previous discussions, we defined factorial effects as contrasts of all expected potential outcomes based on the ideal scenario of a full factorial design, where all possible combinations of factors are randomized. We then estimated the factorial effects by emulating a randomized full factorial experiment.
	However, in some practical scenarios, it may not be possible to emulate a full factorial design, especially when the number of factors $K$ is large, or some treatment combinations are rare or missing in the observed data. As the number of treatment combinations grows exponentially with the number of factors, it is likely that some treatment combinations are not observed due to limitations in units or cost. In such cases, how can we estimate the factorial effects of interest? 
	
	In this section, we extend the proposed balancing weighting method to emulate an incomplete factorial design \citep{byar1993incomplete,pashley2022causal}, which can refer to any subset of a full factorial design. 
	Suppose there are $Q_\textup{u}$ unobserved treatment combinations. Let $\mathcal{Z}_{\textup{u}}$ denote this set of $Q_\textup{u}$ unobserved treatment combinations, and let $\mathcal{Z}_\textup{o}$ denote the set of $Q_\textup{o}=Q-Q_\textup{u}$ observed treatment combinations. When the interaction effects with order greater than $K'$ are negligible (Assumption~\ref{as1}), we only need to estimate $Q_+={K\choose 1}+{K\choose 2}+\cdots+{K\choose K'}$ factorial effects; all the other $Q_-=Q-Q_+-1$ factorial effects are zero.
	How can we use the observed $Q_\textup{o}$ treatment combinations in $\mathcal{Z}_\textup{o}$ to infer the $Q_+$ factorial effects?
	
	To answer this question, we begin by exploring a fundamental question: Is it possible to identify the $Q_+$ factorial effects using only $Q_\textup{o}$ observed treatment combinations? Let $\mathbf{G}=[\mathbf{g}_0,\mathbf{g}_1,\dots,\mathbf{g}_{K}, \mathbf{g}_{1,2},\dots,\mathbf{g}_{K-1,K},\dots,\mathbf{g}_{1,...,K}]$ denote the design matrix for 
	the average outcome across treatments and all $Q-1$ factorial effects, where the first column $\mathbf{g}_0=(+1,\dots,+1)^\textup{T}$. 
	The corresponding vector of average outcome and factorial effects equals $\boldsymbol{\tau} = (2\tau_0, \tau_1, ..., \tau_K, \tau_{1,2}, ....,  \tau_{1,...,K})^\textup{T}=\frac{1}{2^{K-1}}\mathbf{G}^\textup{T}\mathbb{E[\mathbf{Y}]}.$ 
	Under Assumption~\ref{as1}, the high-order interactions are negligible, so the last $Q_-$ entries of $\boldsymbol{\tau}$ are zero.
	For convenience, we rearrange the entries of $\mathbb{E}(\mathbf{Y})$ and $\mathbf{G}$ such that the unobserved treatment combinations occur after the observed ones and the negligible contrasts occur after the non-negligible ones. Using tilde to denote a rearranged vector and matrix, we partition them as $$\mathbb{E}(\mathbf{\tilde{Y}})=\begin{bmatrix}
		(\mathbb{E}[Y(\mathbf{z})])_{\mathbf{z}\in\mathcal{Z}_\textup{o}}\\
		(\mathbb{E}[Y(\mathbf{z})])_{\mathbf{z}\in\mathcal{Z}_\textup{u}}
	\end{bmatrix}
	\quad \textrm{ and } \quad 
	\mathbf{\tilde{G}}=\begin{bmatrix}
		\mathbf{G}_{\textup{o}+} & \mathbf{G}_{\textup{o}-}\\
		\mathbf{G}_{\textup{u}+} & \mathbf{G}_{\textup{u}-}
	\end{bmatrix},$$ 
	where $\mathbf{G}_{\textup{o}+}$ is the $Q_\textup{o}\times (Q_++1)$ submatrix for the observed treatment combinations and  non-negligible contrasts, $\mathbf{G}_{\textup{u}+}$ is the $Q_\textup{u}\times (Q_++1)$ submatrix for the unobserved treatment combinations and non-negligible contrasts, $\mathbf{G}_{\textup{o}-}$ is the $Q_\textup{o}\times Q_-$ submatrix for the observed treatment combinations and negligible contrasts, $\mathbf{G}_{\textup{u}-}$ is the $Q_\textup{u}\times Q_-$ submatrix for the unobserved treatment combinations and negligible contrasts.
	Then  $\boldsymbol{\tau}=\frac{1}{2^{K-1}}\mathbf{\widetilde{G}}^\textup{T}\mathbb{E}(\mathbf{\widetilde{Y}}).$ 
	Due to the orthogonality of $\mathbf{\widetilde{G}}$ that $\mathbf{\widetilde{G}}^\textup{T}\mathbf{\widetilde{G}} = 2^K \mathbf{I}_{2^K}$, 
	we can express the expected potential outcomes in terms of the factorial effects using the relationship $\mathbb{E}[\mathbf{\widetilde{Y}}] = 2^{-1}\mathbf{\widetilde{G}}\boldsymbol{\tau}.$ 
	This one-to-one correspondence between the factorial effects $\boldsymbol{\tau}$ and the expected potential outcomes $\mathbb{E}[\mathbf{\widetilde{Y}}]$ allows us to identify all factorial effects in a full factorial design when all treatment combinations are observed. However, when some treatment combinations are unobserved, it is impossible to estimate the corresponding expected potential outcome without further assumptions. Consequently, we cannot estimate the factorial effects using the above one-to-one correspondence directly. Still, this relationship between $\boldsymbol{\tau}$ and $\mathbb{E}[\mathbf{\widetilde{Y}}]$ is useful, as we can use it to infer the expected potential outcomes for the unobserved treatment combinations $(\mathbb{E}[Y(\mathbf{z})])_{\mathbf{z}\in\mathcal{Z}_\textup{u}}$ by utilizing the negligibility of the $Q_-$ highest-order interactions, given by
	$$\mathbf{0}=\mathbf{G}_{\textup{o}-}^\textup{T}(\mathbb{E}[Y(\mathbf{z})])_{\mathbf{z}\in\mathcal{Z}_\textup{o}}+\mathbf{G}_{\textup{u}-}^\textup{T}(\mathbb{E}[Y(\mathbf{z})])_{\mathbf{z}\in\mathcal{Z}_\textup{u}}.$$ Specifically, if $\mathbf{G}_{\textup{u}-}$ has full row rank, we can infer the expected potential outcomes for unobserved treatment combinations $(\mathbb{E}[Y(\mathbf{z})])_{\mathbf{z}\in\mathcal{Z}_\textup{u}}$ using the relationship 
	$$(\mathbb{E}[Y(\mathbf{z})])_{\mathbf{z}\in\mathcal{Z}_\textup{u}}=-\mathbf{G}_{\textup{u}-}(\mathbf{G}_{\textup{u}-}^\textup{T}\mathbf{G}_{\textup{u}-})^{-1}\mathbf{G}_{\textup{o}-}^\textup{T}(\mathbb{E}[Y(\mathbf{z})])_{\mathbf{z}\in\mathcal{Z}_\textup{o}}.$$
	It is worth noting that having a full row rank $\mathbf{G}_{\textup{u}-}$ is a necessary and sufficient condition for identifying the non-negligible factorial effects. This condition is only possible to hold if 
	the number of non-negligible factorial effects $Q_+$ is no more than the number of observed treatment combinations $Q_\textup{o}$.

	Let $\boldsymbol{\tau}_+$ denote the vector of average outcome and non-negligible factorial effects. With the identified expected potential outcomes for the unobserved treatment combinations, we can identify $\boldsymbol{\tau}_+$ under an incomplete factorial design as in Proposition~\ref{prop:rel} below. 
	Proposition~\ref{prop:rel} suggests that the factorial effects, originally defined based on the mean of all potential outcomes (the treatment combination groups contribute equally to each contrast), are equivalent to linear combinations of the expected potential outcomes for the observed treatment combinations.
	Therefore, different treatment combination groups can contribute to the factorial effects with different weights.
	To illustrate the idea, we consider a toy example with $K=3$ binary treatments 
	in \S S3.5 in the Supplementary Materials. Notably, we can also relax the ignorability assumption in the sense that (i) the probabilistic assumption only needs to hold for the observed treatment combinations, and (ii) the unconfounded assumption only requires $\mathbf{Z}_i \ind (Y_i(\mathbf{z}))_{\mathbf{z}\in\mathcal{Z}_\textup{o}}\mid\mathbf{X}_i$.
	\begin{proposition}\label{prop:rel}
		If the matrix $\mathbf{G}_{\textup{u}-}$ has full row rank, we have
		$$\boldsymbol{\tau}_+=\frac{1}{2^{K-1}}\mathbf{G}_\textup{o}^\textup{T}(\mathbb{E}[Y(\mathbf{z})])_{\mathbf{z}\in\mathcal{Z}_\textup{o}}, \quad	\textrm{ where } \mathbf{G}_\textup{o}^\textup{T}	=\mathbf{G}_{\textup{o}+}^\textup{T}-\mathbf{G}_{\textup{u}+}^\textup{T}\mathbf{G}_{\textup{u}-}(\mathbf{G}_{\textup{u}-}^\textup{T}\mathbf{G}_{\textup{u}-})^{-1}\mathbf{G}_{\textup{o}-}^\textup{T}.$$ 
	\end{proposition}

	With the formula for the non-negligible factorial effects in Proposition~\ref{prop:rel}, we can formulate our optimization problem in a similar way to Section~\ref{sec:full} to find the optimal weights that mimic an incomplete factorial design and hence can be used to estimate the non-negligible factorial effects. Although the problem formulation is similar as before, it is worth noting that the indicators of whether unit $i$ contribute to the positive or negative parts of factorial effect ${\tau}_\cK$ in a full factorial design, $A_{i\cK}^+$ and $A_{i\cK}^-$, have new meanings in an incomplete factorial design, representing the weights that unit $i$ contributes to the positive or negative parts of factorial effect ${\tau}_\cK$. Importantly, the theoretical properties and the form of the variance estimator remain unchanged, and we omit the proofs since they are the same as before.
	
	\section{Simulation}\label{sec:sim}
	\subsection{Performance of the Weighting Estimator}
	
	In the first set of simulations, we evaluate the performance of the proposed weighting estimator for estimating factorial effects. Suppose we observe three factors $\bZ_i=(Z_{i1},Z_{i2},Z_{i3})$ and five covariates $\bX_i=(X_{i1},\dots,X_{i5})$ for each individual $i$. The covariates $\bX_i$ follows a multivariate normal distribution $\cN(\mathbf{\mu},\Sigma)$, where $\mathbf{\mu}=(0,0,0,0,0)^\textup{T}$ and $\Sigma$ has diagonal elements 1 and off-diagonal elements $\rho$. The treatment assignment mechanism for $Z_{ik}$ is independent across $k$'s and  satisfies a logistic regression that $\PP(Z_{ik}=1)=1/(1+\exp(-\boldsymbol{\beta}_k^\textup{T}\bX_i))$, where $\boldsymbol{\beta}_1=(1/4,2/4,0,3/4,1), \boldsymbol{\beta}_2=(3/4,1/4,1,0,2/4), \boldsymbol{\beta}_3=(1,0,3/4,2/4,1/4)$. Here, we assume the conditional independence of factors given covariates for convenience, but it is not needed in the theory throughout the paper. This treatment assignment mechanism ensures all $2^3=8$ treatment combination groups are non-empty and are observed so that the proposed weighting estimators in Section~\ref{sec:method} are applicable. Suppose only the main effects of the three treatments are non-negligible. We consider three outcome models: an additive outcome $Y_{i1}=2\sum_{j=1}^5X_{ij}+\sum_{k=1}^3Z_{ik}+\epsilon_{i1}$, a heterogeneous treatment effect outcome $Y_{i2}=2\sum_{j=1}^5X_{ij}+\sum_{k=1}^3Z_{ik}+\sum_{j=1}^5X_{ij}\sum_{k=1}^3Z_{ik}+\epsilon_{i2}$, and a misspecified outcome $Y_{i3}=4\sin(X_{i1})+\exp(0.4X_{i2}^2)+\left(\min(1,X_{i1})+X_{i2}\right)Z_{i1}+X_{i1}Z_{i2}+\sum_{j=1}^5X_{ij}Z_{i3}+\epsilon_{i3}$, where $\epsilon_{ij}$'s are independent errors following a standard normal distribution. The true main effects for $Y_1$ are $\tau_1=2,\tau_2=2,\tau_3=2$. The true main effects for $Y_2$ are $\tau_1=2,\tau_2=2,\tau_3=2$. The true main effects for $Y_3$ are $\tau_1=2\mathbb{E}[\min(X_1,1)],\tau_2=0,\tau_3=0$. We consider four  estimators for each main effect $\tau_k, k=1,2,3$, under each outcome model: 
	(i) the additive regression estimator, which is twice the coefficient of $Z_{k}$ when regressing $Y$ on the intercept, centralized $X_{j}$ (i.e., $X_j-\bar X_j$) and $Z_{k}$ for $j=1,\dots,5, k=1,2,3$, (ii) the interaction regression estimator, which is twice the coefficient of $Z_{k}$ when regressing $Y$ on the intercept, $X_j-\bar X_j$, $Z_k$ and $(X_j-\bar X_j)Z_k$ for $j=1,\dots,5, k=1,2,3$ \citep{zhao2023covariate}, (iii) the proposed weighting estimator using balance constraints under the general additive model assumption (additive balance constraints) and covariate basis functions $h_s(\bX)=X_{s}, s=1,\dots,5$, and (iv) the proposed weighting estimator using balance constraints under the outcome model assumption with treatment effect heterogeneity (interaction balance constraints) and the same set of basis functions as (iii). We vary the sample size as $N=500,1000,2000$ and vary the covariance as $\rho=0.2,0.4,0.6$ for each scenario. In each simulation setting, we compare the absolute bias and root mean squared error (RMSE) using 1000 repetitions. We summarize the results for $N=1000$ and $\rho=0.4$ in Figure~\ref{fig:est}. We present the results for other choices of $N$ and $\rho$ and evaluate the effects of sample size $N$ and covariance $\rho$ in the Supplementary Materials  \S S1.
	
	\begin{figure}[h!]
		\centering
		\includegraphics[scale=0.4]{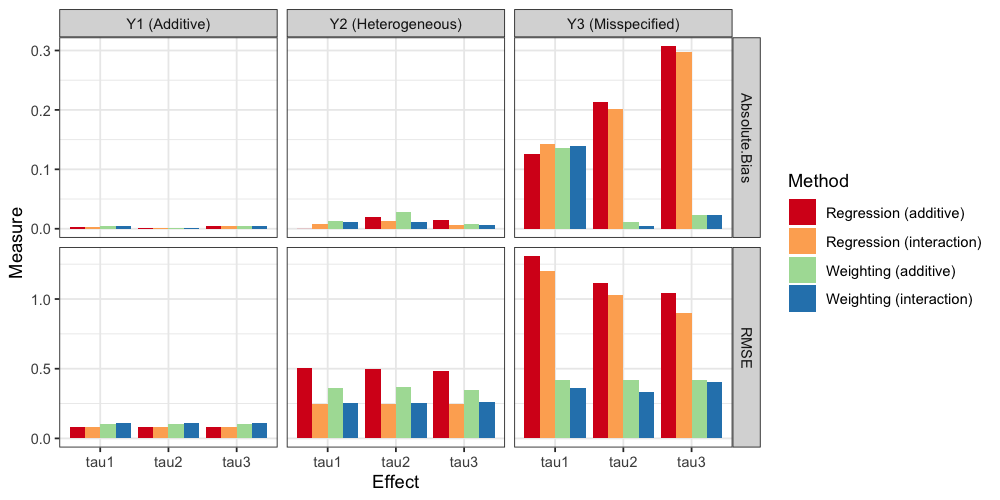}
		\caption{\small Absolute bias and root mean squared error (RMSE) for estimating three main effects over 1000 repetitions using four estimators when $N=1000$ and $\rho=0.4$.}
		\label{fig:est}
	\end{figure}
	
	Figure~\ref{fig:est}  suggests that when the outcome is additive, the four estimators have similar performances and the regression estimators achieve slightly smaller RMSEs than the two weighting estimators. When the treatment effect is heterogeneous, the two methods allowing treatment effect heterogeneity (interaction regression estimator and weighting with interaction balance constraints) have smaller biases and RMSEs than the other two  methods. When the outcome is misspecified, the weighting estimator with interaction balance constraints outperforms the other methods, achieving the smallest RMSEs. This observation confirms the robustness of the proposed weighting method. 
	
	In the Supplementary Materials \S S1, we conduct additional simulations with a larger number of factors and under heteroskedasticity. We also consider another set of treatment assignment mechanism where not all treatment combination groups are observed, hence the proposed method in \S\ref{sec:method} is no longer applicable and the performance of the generalized method in \S\ref{sec:extension} is evaluated.
	
	\subsection{Performance of the Variance Estimator}
	In the next set of simulations, we evaluate the performance of the proposed variance estimator in Theorem~\ref{thm:var}. We consider the same setting as the first set of simulations, with a sample size of $N=500,1000,2000$ and $\rho=0.4$. We estimate $\sqrt{N}(\widehat\tau_k-\tau_k)$ in two ways: we obtain the benchmark estimator by adjusting the variance of estimated factorial effects by the sample size and also consider the average proposed consistent variance estimator across 1000 repetitions. From the results in Figure~\ref{fig:var}, we can see that  the simulated variance estimator and the consistent variance estimator are similar, with a ratio getting closer to 1 as the sample size increases. Additionally, we calculate the empirical coverage probability of 95\% confidence intervals obtained using normal distribution (Theorem~\ref{thm:rootn}) and the consistent variance estimator (Theorem~\ref{thm:var}). Results in Figure~\ref{fig:var} suggest that the confidence intervals have coverage probabilities close to 0.95. 
	
	\begin{figure}[h!]
		\centering
		\includegraphics[scale=0.35]{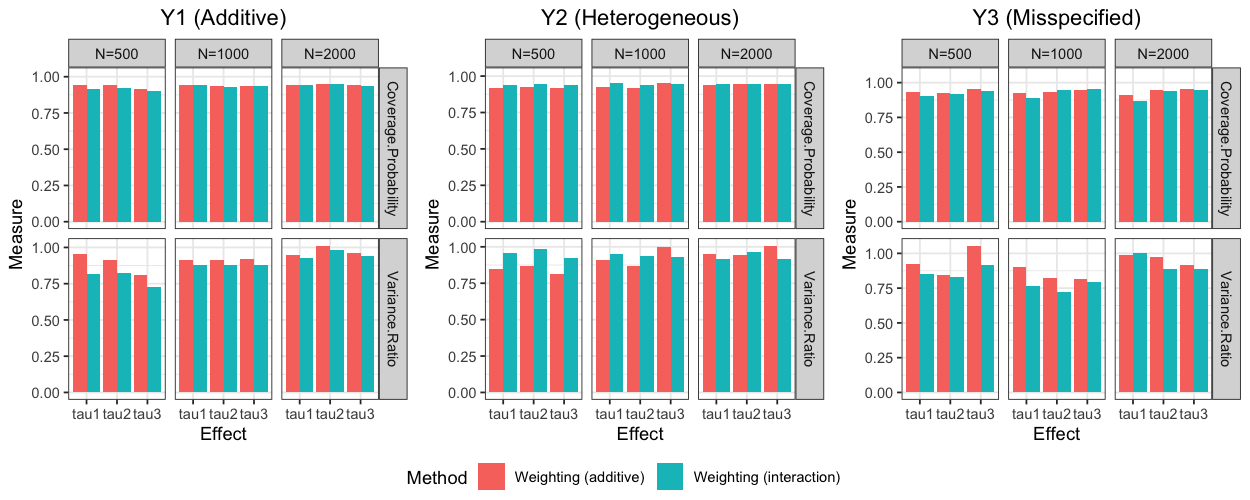}
		\caption{\small Simulated variance estimator and consistent variance estimator for three main effects ($\sqrt{N}(\widehat\tau_k-\tau_k), k=1,2,3$). Variance ratio is the ratio of the consistent variance estimator to the simulated variance estimator. Empirical coverage probabilities of 95\% confidence intervals are calculated using the normal distribution and the proposed consistent variance estimator.}
		\label{fig:var}
	\end{figure}

	\vspace{-1cm}
	
	\section{Application: VOCs Exposure and Tachycardia}\label{sec:data}
	Exposure to volatile organic compounds (VOCs) is pervasive in our daily life due to common indoor and outdoor sources such as building materials, home cleaning products, gasoline, and industrial emissions \citep{batterman2014personal}. These exposures can increase the risk of adverse health consequences, including lung and pulmonary function issues, liver and kidney dysfunction, birth defects (e,g., neurologic impairment), asthma, and respiratory complications \citep{arif2007association,batterman2014personal,johnson1999review,montero2018volatile,woodruff2011environmental}. While many studies have focused on the effects of individual environmental pollutants, people are often exposed to multi-pollutants and pollution mixtures in real life \citep{park2014environmental}. To address this gap, recent studies have begun to investigate the effects of multi-pollutants and pollution mixtures \citep{batterman2014personal,park2014environmental,patel2013systematic}. 
	
	We analyze the biomonitoring data from the  National Health and Nutritional Examination Survey (NHANES) 2003-2012 to evaluate the effects of individual VOC exposures and their two-way interactions on tachycardia (i.e., rapid heart rate) - a risk factor to many cardiovascular diseases - among women of reproductive age. Specifically, we consider four common VOC exposures -- Benzene, 1,4-Dichlorobenzene, Ethylbenzene, and methyl tert-butyl ether (MTBE) -- for a nationally representative sample of 2003 women aged 15-49 years old in the U.S. population. 
	For more discussions about sources of these VOC exposures, see \S S2 in the Supplementary Materials.
	We are interested in estimating the effects of four binary treatments on whether the VOCs' blood concentration exceeds their detection limits.
	Our outcome of interest is heart rate. We also have access to people's demographic information (age, race, family income to poverty ratio), hypertension, alcohol consumption, and smoking status. As shown in Table~S3 in the Supplementary Materials, the sample sizes and the covariate distributions vary across treatment combination groups.

	To remove confounding due to the observed covariates and draw reliable causal conclusions, we employ the proposed weighting estimator with interaction balance constraints and the covariates themselves as basis functions to emulate a $2^4$ full factorial design. As shown in Figure~S6 in the Supplementary Materials \S S2.3, our weighting method greatly improved covariate balances (in terms of absolute standardized mean differences) in all contrasts for main effects and two-way interactions. Successfully controlling for confounding due to age, race, family income to poverty ratio, hypertension, alcohol consumption, and smoking status allows us to estimate the factorial effects of VOC exposures on tachycardia more reliably.
	
	Next, we estimate the factorial effects using our weighting estimator, estimate the variance using the proposed consistent variance estimator, and construct the corresponding 95\% confidence intervals, assuming there are no other unmeasured confounders. The results are summarized in Table~\ref{tb:eff}. Our analysis suggests that exposure to 1,4-Dichlorobenzene and MTBE can significantly increase heart rates, but there is no significant evidence for the effects of Benzene and Ethylbenzene. Furthermore, we found no significant evidence of the effects of one VOC on heart rate depending on the other. Our findings highlight the importance of being cautious while selecting home improvement products and choosing living areas with safe air and water quality. Additionally, our results strengthen the reliability of previous studies that have focused on individual chemicals since the interaction effects are insignificant \citep{arif2007association,montero2018volatile,weichenthal2011traffic}. 
	In comparison, using regression with treatment-covariate interactions to adjust for the confounding effects (as in the simulation studies) would fail to detect the significant factorial effects for the four VOC exposures we examined, whereas the weighting method yields more robust and accurate estimates of these effects.
	\begin{table}[h!]
		\centering
		\caption{Factorial effects (four main effects and six second-order interaction effects) for the NHANES study: Estimated factorial effects, consistent estimated variance, corresponding 95\% confidence interval (CI). $^*$ indicates significant effects at level $\alpha=0.05$.}
		\resizebox{\textwidth}{!}{
			\begin{tabular}{llrrrrrrrr}  
				\hline \hline
				&Treatment& \multicolumn{3}{c}{Weighting}&& \multicolumn{3}{c}{Regression} \\ \cline{3-5}\cline{7-9}
				&& Estimate & Variance &95\%  CI&& Estimate & Variance &95\%  CI \\   \hline
				\multirow{4}{*}{Main Effect}&Benzene & -0.288 & 0.643 & (-1.860, 1.284) &&-0.485  & 0.917  & (-2.362, 1.392)\\   
				&1,4-Dichlorobenzene & 1.715 & 0.655 & (0.129, 3.300)$^*$ && 0.320 &  0.665  &  (-1.278, 1.918)\\   
				&Ethylbenzene & -0.292 & 0.645 & (-1.866, 1.282) && 0.016  & 0.979  &  (-1.923, 1.955)\\   
				&MTBE & 2.056 & 0.658 & (0.466, 3.646)$^*$ && 1.261 &  0.776  &  (-0.466, 2.989)\\   \hline
				\multirow{6}{*}{Two-way Interaction}&Benzene: 1,4-Dichlorobenzene & 0.335 & 0.644 & (-1.238, 1.908) && 0.510  & 0.864  &  (-1.311,      2.331)\\   
				&Benzene: Ethylbenzene & 0.992 & 0.649 & (-0.586, 2.571) && 1.322 &  0.898  &  (-0.535, 3.179)\\   
				&Benzene: MTBE& -0.657 & 0.647 & (-2.234, 0.919) && -0.699  & 0.848  &  (-2.504, 1.107)\\   
				&1,4-Dichlorobenzene: Ethylbenzene & -1.561 & 0.647 & (-3.138, 0.015) && -1.102  & 0.701   & (-2.742, 0.539)\\   
				&1,4-Dichlorobenzene: MTBE & -0.183 & 0.642 & (-1.754, 1.387) && -1.127 &  0.354   & (-2.293, 0.040)\\   
				&Ethylbenzene: MTBE & -1.556 & 0.647 & (-3.132, 0.020) && -1.029   &0.734   & (-2.708, 0.650)\\  \hline \hline
		\end{tabular}}
		\label{tb:eff}
	\end{table}

	\section{Discussion}
	Although this paper focuses on the factorial effects of binary treatments, the proposed methods can be readily extended to multi-level treatments, albeit with additional notational complexity (see Appendix A of \cite{zhao2022regression}). Specifically, for $l$-level treatments, we redefine $A_i^\Omega$ over the extended range $\Omega=1,\dots,l$ instead of the binary choice $\Omega=+,-$. The balance constraints derived in this paper can then be applied to yield unbiased estimates of factorial effects $\tau_\cK$. The procedure for estimating weights and conducting inferences for the factorial effects remains the same as discussed in the main text, so we omit the details for brevity.
	
	The weighting framework introduced in this paper relies on the unconfoundedness assumption, which, while commonly adopted in observational studies, may not always hold in practice. A valuable direction for future research would be to investigate how unmeasured confounders might affect the estimation of factorial effects. One possible approach is to conduct sensitivity analyses. While some sensitivity analysis methods have been proposed for the weighting methods \citep{zhao2019sensitivity,soriano2023interpretable}, extending these techniques to accommodate multiple treatments in factorial studies would be a beneficial advance.

	\section*{Supplementary Material}
	The supplement file contains additional simulation experiments, further details of the NHANES study, extended discussions and illustrative examples of key concepts from the main paper, and the technical details and proofs of the theoretical results. Code for implementing the proposed method is available at the GitHub repository \url{https://github.com/ruoqiyu/FactorialOS}.

	\section*{Data Availability Statement}

	Data are publicly available at  \url{https://wwwn.cdc.gov/nchs/nhanes/Default.aspx}.
	
	\section*{Acknowledgements}
	The authors would like to thank Avi Feller, Luke Keele, Qingyuan Zhao, and José Zubizarreta for insightful discussions.

	\clearpage
	\begin{center}
		{
			\Large Supplementary Materials
		}
	\end{center}

		\setcounter{section}{0}
	\setcounter{equation}{0}
	\setcounter{figure}{0}
	\setcounter{table}{0}
	\setcounter{proposition}{0}
	\setcounter{assumption}{0}
	\def\theequation{S\arabic{section}.\arabic{equation}}
	\def\thesection{S\arabic{section}}
	\def\thefigure{S\arabic{figure}}
	\def\thetable{S\arabic{table}}
	\def\theproposition{S\arabic{proposition}}
	\def\theassumption{S\arabic{assumption}}


	\section{Additional Simulation}\label{app:sim}
	\subsection{Additional Simulation With Varying Sample Sizes} \label{sup:samplesize}
	In this set of simulations, we evaluate the performance of the proposed weighting estimator for estimating factorial effects with various choices of sample sizes $N=500,1000,2000$. We consider the same setting as in \S 5 of the main paper with three factors $\bZ_i=(Z_{i1},Z_{i2},Z_{i3})$ and five covariates $\bX_i=(X_{i1},\dots,X_{i5})$ for each individual $i$. Recall that the covariates $\bX_i$ follows a multivariate normal distribution $\cN(\mathbf{\mu},\Sigma)$, where $\mathbf{\mu}=(0,0,0,0,0)^\textup{T}$ and 
	$\Sigma$ has diagonal elements 1 and off-diagonal elements $\rho=0.4$. The treatment assignment mechanism for $Z_{ik}$ is independent across $k$'s and  satisfies a logistic regression that $\PP(Z_{ik}=1)=1/(1+\exp(-\boldsymbol{\beta}_k^\textup{T}\bX_i))$, where $\boldsymbol{\beta}_1=(1/4,2/4,0,3/4,1), \boldsymbol{\beta}_2=(3/4,1/4,1,0,2/4), \boldsymbol{\beta}_3=(1,0,3/4,2/4,1/4)$. Suppose only the main effects of the three treatments are non-negligible. We consider two outcome models: a well-specified outcome $Y_{i1}=2\sum_{j=1}^5X_{ij}+\sum_{k=1}^3Z_{ik}+\sum_{j=1}^5X_{ij}\sum_{k=1}^3Z_{ik}+\epsilon_{i1}$, and a misspecified outcome $Y_{i2}=4\sin(X_{i1})+\exp(0.4X_{i2}^2)+\left(\min(1,X_{i1})+X_{i2}\right)Z_{i1}+X_{i1}Z_{i2}+\sum_{j=1}^5X_{ij}Z_{i3}+\epsilon_{i2}$, where $\epsilon_{ij}$'s are independent errors following a standard normal distribution. The true main effects for $Y_1$ are $\tau_1=2,\tau_2=2,\tau_3=2$. The true main effects for $Y_2$ are $\tau_1=2\mathbb{E}[\min(X_1,1)],\tau_2=0,\tau_3=0$.  We consider four  estimators for each main effect $\tau_k, k=1,2,3$, under each outcome model: 
	(i) the additive regression estimator, which is twice the coefficient of $Z_{k}$ when regressing $Y$ on the intercept, centralized $X_{j}$ (i.e., $X_j-\bar X_j$) and $Z_{k}$ for $j=1,\dots,5, k=1,2,3$, (ii) the interaction regression estimator, which is twice the coefficient of $Z_{k}$ when regressing $Y$ on the intercept, $X_j-\bar X_j$, $Z_k$ and $(X_j-\bar X_j)Z_k$ for $j=1,\dots,5, k=1,2,3$ \citep{zhao2023covariate}, (iii) the proposed weighting estimator using balance constraints under the general additive model assumption (additive balance constraints) and covariate basis functions $h_s(\bX)=X_{s}, s=1,\dots,5$, and (iv) the proposed weighting estimator using balance constraints under the outcome model assumption with treatment effect heterogeneity (interaction balance constraints) and the same set of basis functions as (iii). In each simulation setting, we compare the absolute bias and root mean squared error (RMSE) using 1000 repetitions. We summarize the results in Figure~\ref{fig:estn}. The performances of all four estimators get better as the sample size increases. The proposed weighting estimator with interaction balance constraints remains the best estimator as we vary the sample size. 
	
	\begin{figure}[h!]
		\centering
		\includegraphics[scale=0.45]{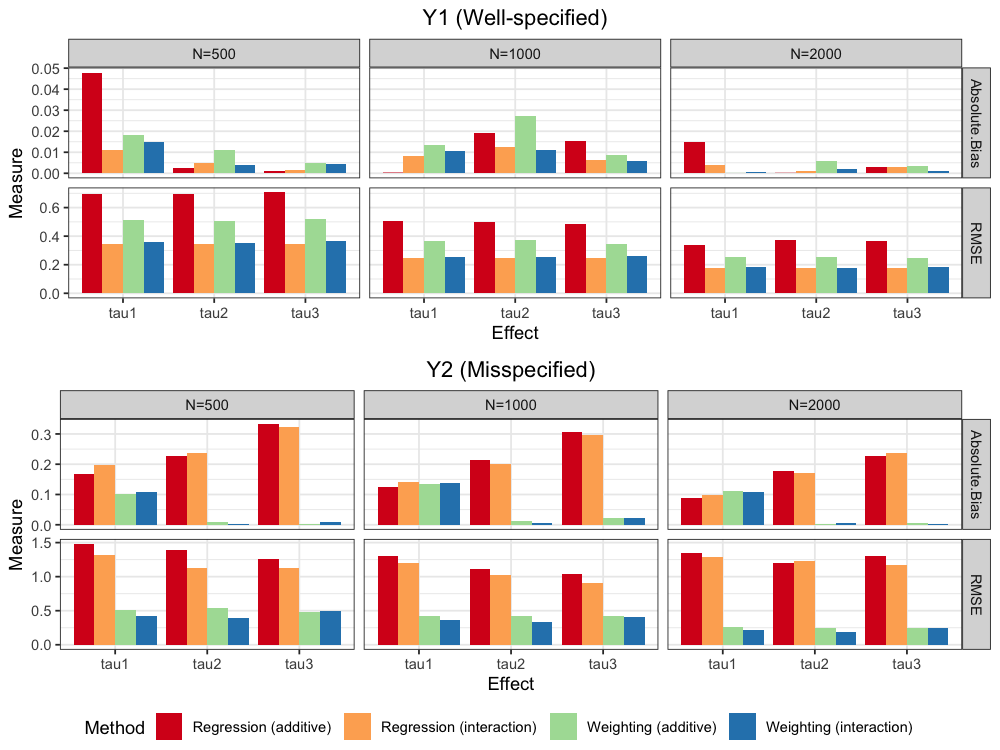}
		\caption{Bias and root mean squared error (RMSE) for estimating three main effects over 1000 repetitions using four estimators when $N=500,1000,2000$ and $\rho=0.4$.}
		\label{fig:estn}
	\end{figure}
	
	\subsection{Additional Simulation With Varying Covariance Matrices}
	In the second set of simulations, we consider the same setting as in \S 5 of the main paper with different choices of covariance for each pair of covariates $\rho=0.2,0.4,0.6$. For the detailed description of the setup, see \S\ref{sup:samplesize} of the Supplementary Materials. In each simulation setting, we calculate the absolute bias and root mean squared error (RMSE) of the four estimators using 1000 repetitions. We summarize the results in Figure~\ref{fig:estr}. We can observe that the proposed weighting estimator maintains small bias and RMSE as $\rho$ varies, suggesting its robustness to dependence among covariates.
	
	\begin{figure}[h!]
		\centering
		\includegraphics[scale=0.45]{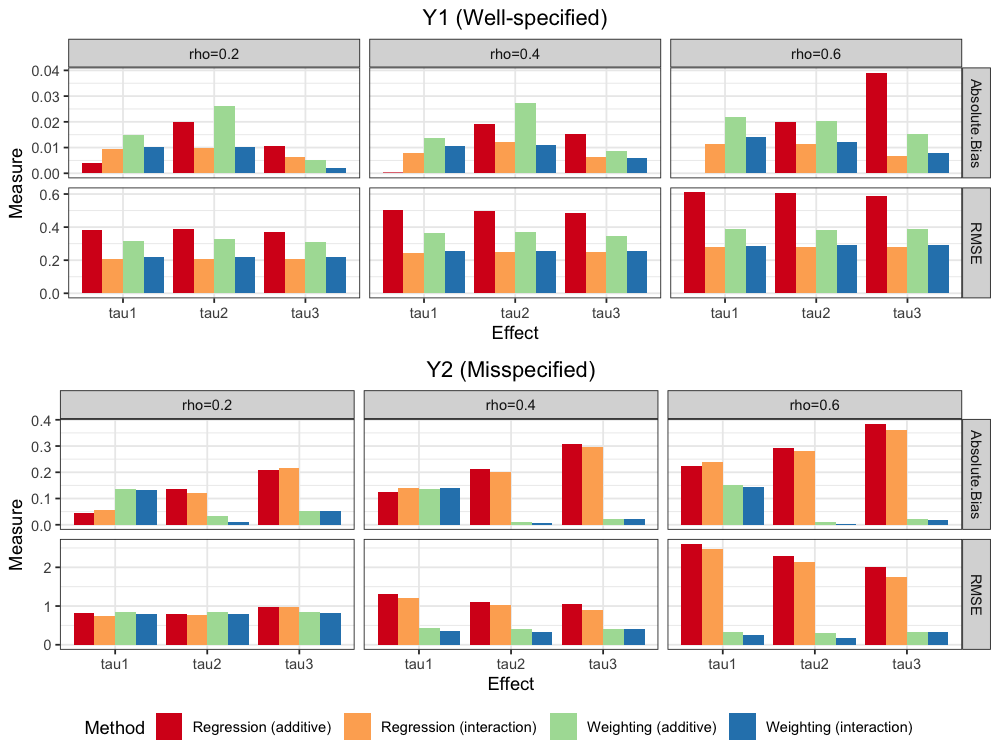}
		\caption{Bias and root mean squared error (RMSE) for estimating three main effects over 1000 repetitions using four estimators when $N=1000$ and $\rho=0.2,0.4,0.6$.}
		\label{fig:estr}
	\end{figure}
	
	\subsection{Performance Under Heteroskedasticity}
	
	In the following set of simulations, we evaluate the performance of the proposed weighting estimator with interaction balance constraints under heteroskedasticity. We consider the setting in the first set of simulation in \S 5 of the main paper with a sample size of $N=1000$ and covariance $\rho=0.4$, but allow the independent errors $\epsilon_{ij}$ to follow a normal distribution $\cN(0, v_{ij})$, where the variance $v_{ij}$ is randomly drawn from a uniform distribution on the range $[0,C]$. Here, the upper bound $C$ controls the scale of heteroskedasticity, and we consider $C=1$ and $C=10$. We calculate the bias, RMSE, 
	and 95\% confidence interval coverage probabilities over 1000 repetitions. Results in Figure~\ref{fig:hetero} indicate that the proposed weighting estimator performs well in both cases, achieving small biases, RMSEs, and reasonable confidence interval coverages.
	
	\begin{figure}[h!]
		\centering
		\includegraphics[scale=0.45]{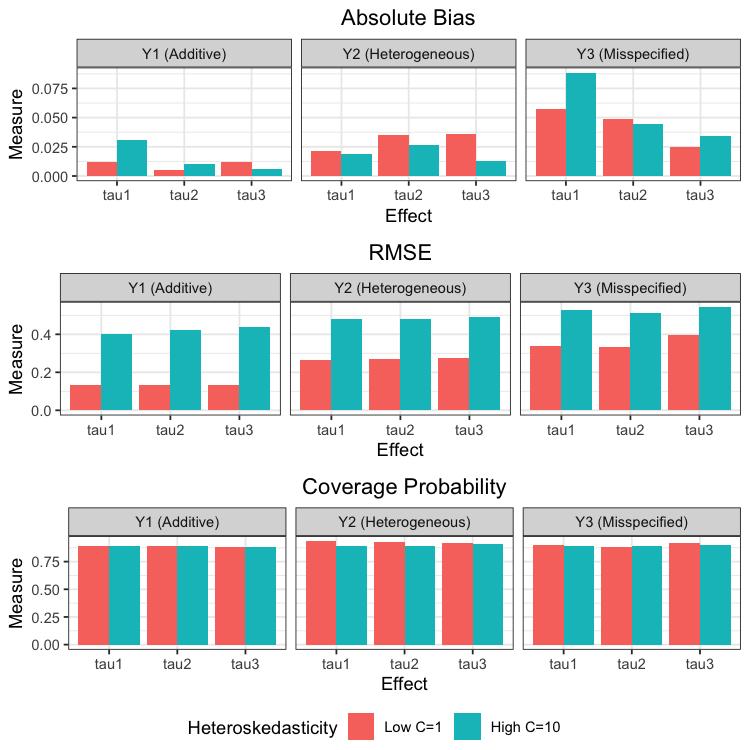}
		\caption{Performance under heteroskedasticity: Bias, root mean squared error (RMSE), and coverage probability of 95\% confidence intervals for estimating three main effects using the proposed weighting estimator under an outcome model assumption with treatment effect heterogeneity and covariate basis functions $h_s(\bX)=X_{s}, s=1,\dots,5$, over 1000 repetitions. }
		\label{fig:hetero}
	\end{figure}

	\subsection{Simulation with Unobserved Treatment Combination Groups}
	In this set of simulations, we evaluate the performance of the proposed weighting estimators when some treatment combination groups are not observed. We consider the setting in the first set of simulations in \S 5 of the main paper with a sample size of $N=1000$ and covariance $\rho=0.4$ but a different treatment assignment mechanism. Specifically, let $\boldsymbol{\beta}_1=(1/4,2/4,0,3/4,1), \boldsymbol{\beta}_2=(3/4,1/4,1,0,2/4), \boldsymbol{\beta}_3=(1,0,3/4,2/4,1/4)$.  We assume $\PP(Z_{ik}=1)=1/(1+\exp(-\boldsymbol{\beta}_k^\textup{T}\bX_i))$ for $k=1,2$ and $\PP(Z_{i3}=1)=1$ if $Z_{i1}=Z_{i2}=1$, $\PP(Z_{i3}=-1)=1$ if $Z_{i1}=Z_{i2}=-1$, and $\PP(Z_{i3}=1)=1/(1+\exp(-\boldsymbol{\beta}_3^\textup{T}\bX_i))$ otherwise. In doing so, two of the $2^3=8$ treatment combination groups are not observed, so the weighting estimators in \S4 in the main paper are applicable. We first compare the performance of the two regression estimators and the two weighting estimators using biases and RMSEs over 1000 repetitions. Based on the results in Figure \ref{fig:estu}, we can observe that the weighting estimator with interaction balance constraints can consistently achieve small biases and RMSEs in all cases. In addition, we evaluate the performance of our variance estimators by calculating the 95\% confidence interval coverage probabilities over 1000 repetitions. Results in Table~\ref{tb:unobs} indicate that the proposed variance estimator is close to the simulated variance, and the confidence interval coverage probabilities are close to 0.95.
	
	\begin{figure}[h!]
		\centering
		\includegraphics[scale=0.45]{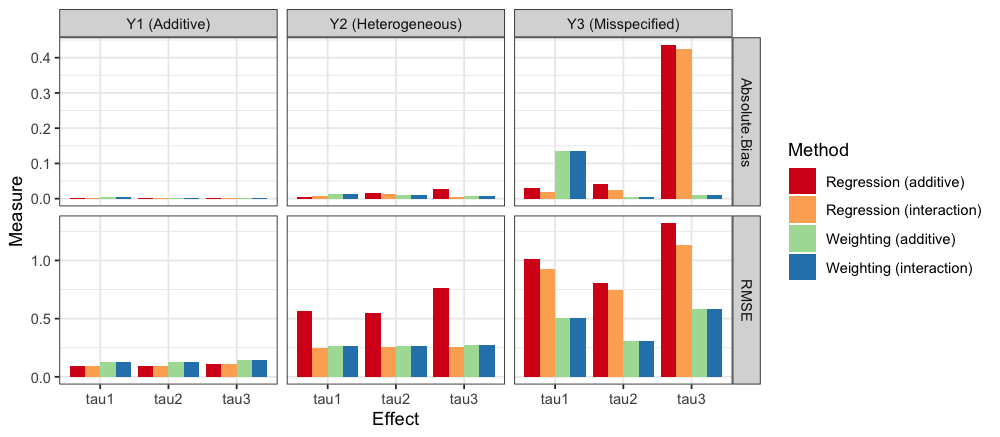}
		\caption{Performance with unobserved treatment combinations: Bias and root mean squared error (RMSE) for estimating three main effects over 1000 repetitions using four estimators when $N=1000$ and $\rho=0$.}
		\label{fig:estu}
	\end{figure}

	\begin{table}[ht]
		\centering
		\caption{Performance with unobserved treatment combinations: Bias, root mean squared error (RMSE), simulated variance estimator, consistent variance estimator, ratio of the consistent variance estimator to the simulated variance estimator, and coverage probability of 95\% confidence intervals for estimating three main effects using the proposed weighting estimator under an outcome model assumption with treatment effect heterogeneity and covariate basis functions $h_s(\bX)=X_{s}, s=1,\dots,5$, over 1000 repetitions. }
		\resizebox*{0.8\textwidth}{!}{
			\begin{tabular}{llrrrrrr}
				\hline	\hline
				Outcome & Effect & Bias & RMSE & Simulated Var & Consistent Var & Variance Ratio & CI Coverage \\   \hline
				\multirow{3}{*}{$Y_1$ (Additive)}&$\tau_1$ & -0.004 & 0.127 & 16.123 & 14.427 & 0.895 & 0.933 \\ 
				&$\tau_2$ & 0.001 & 0.123 & 15.232 & 13.783 & 0.905 & 0.932 \\ 
				&$\tau_3$ & -0.001 & 0.143 & 20.499 & 18.818 & 0.918 & 0.943 \\ \hline
				\multirow{3}{*}{$Y_2$ (Heterogeneous)}&$\tau_1$ & 0.012 & 0.264 & 69.521 & 66.146 & 0.951 & 0.952 \\
				&$\tau_2$ &0.011 & 0.265 & 69.913 & 65.563 & 0.938 & 0.936 \\ 
				&$\tau_3$ &0.006 & 0.274 & 74.914 & 70.772 & 0.945 & 0.940 \\  \hline
				\multirow{3}{*}{$Y_3$ (Misspecified)}&$\tau_1$ & 0.135 & 0.503 & 235.185 & 157.349 & 0.669 & 0.913\\
				&$\tau_2$ &-0.005 & 0.309 & 95.803 & 79.909 & 0.834 & 0.954 \\
				&$\tau_3$ & 0.010 & 0.586 & 343.184 & 242.920 & 0.708 & 0.940 \\ 
				\hline
		\end{tabular}}
		\label{tb:unobs}
	\end{table}

	\subsection{Additional Simulation With More Factors}
	In this additional set of simulations, we consider a more complicated scenario with five factors $\bZ_i=(Z_{i1},\dots,Z_{i5})$ and five covariates $\bX_i=(X_{i1},\dots,X_{i5})$. As in the simulation in the main paper, the covariates follow a multivariate normal distribution $\cN(\mathbf{\mu},\Sigma)$, where $\mathbf{\mu}=(0,0,0,0,0)^\textup{T}$ and $\Sigma_{ij}=1$ for $i=j$
	and $\Sigma_{ij}=\rho=0.4$ otherwise. The treatment assignment mechanism for $Z_{ik}$  is independent across $k$'s and satisfies a logistic regression that $\PP(Z_{ik}=1)=1/(1+\exp(-\boldsymbol{\beta}_k^\textup{T}\bX_i))$, where $\boldsymbol{\beta}_1=(1/4,2/4,0,3/4,1), \boldsymbol{\beta}_2=(3/4,1/4,1,0,2/4), \boldsymbol{\beta}_3=(1,0,3/4,2/4,1/4), \boldsymbol{\beta}_4=(1/4,-1/4,1,3/4,2/4), 
	\boldsymbol{\beta}_5=(0,3/4,-2/4,2/4,1/4)$. Suppose all main effects and second-order interaction effects are non-negligible. We consider a sample size of $N=2000$ and two outcome models: a well-specified outcome $Y_{i1}=3\sum_{j=1}^5X_{ij}+2\sum_{k=1}^3Z_{ik}+2\sum_{j=1}^5X_{ij}\sum_{k=1}^2Z_{ik}+\sum_{k_1, k_2=1,...,5: k_1\neq k_2}Z_{k_1}Z_{k_2}+\sum_{j=1}^5X_{ij}(Z_{i1}Z_{i2}+Z_{i1}Z_{i3})+\epsilon_{i1}$, and a misspecified outcome $Y_{i2}=\exp(0.4X_{i5}^2)+2\sum_{k=1}^5Z_{ik}+\max(X_{i5},1)Z_{i1}+2\sum_{j=1}^5X_{ij}Z_{i2}+\sum_{k_1, k_2=1,...,5: k_1\neq k_2}Z_{k_1}Z_{k_2}+\sin(X_{i5})Z_{i1}Z_{i2}+\sum_{j=1}^5X_{ij}Z_{i1}Z_{i3}+\epsilon_{i2}$, where $\epsilon_{ij}$'s are independent errors following a standard normal distribution. The true main effects for $Y_1$ are $\tau_1=\cdots=\tau_5=4$ and $\tau_{12}=\cdots=\tau_{45}=2$. The true main effects for $Y_2$ are $\tau_1=4+2\mathbb{E}[\max(X_5,1)]$, $\tau_2=\cdots=\tau_5=4$ and $\tau_{12}=\cdots=\tau_{45}=2$. We compare four estimators: (i) the additive regression estimator by regressing $Y$ on the intercept, $X_{1}-\bar X_{1},\dots,X_{5}-\bar X_{5}$, $Z_{1},\dots,Z_{5}$ and $Z_{k_1}Z_{k_2}$ for $k_1=1,\dots,4, k_2=k_1+1,\dots,5$, (ii) the interaction regression estimator by regressing $Y$ on the intercept, $X_{j}-\bar X_{j},Z_{k},Z_{k_1}Z_{k_2}, (X_j-\bar X_{j})Z_k, (X_j-\bar X_{j})Z_{k_1}Z_{k_2}$ for $j=1,...,5; k=1,...,5; k_1=1,\dots,4, k_2=k_1+1,\dots,5$, (iii) the proposed weighting estimator using additive balance constraints and covariate basis functions $h_s(\bX)=X_{s}, s=1,\dots,5$, and (iv) the proposed weighting estimator using interaction balance constraints and covariate basis functions $h_s(\bX)=X_{s}, s=1,\dots,5$,. We compare the bias and RMSE with 1000 repetitions and summarize the results in Figure~\ref{fig:est5}. We observe similar patterns as before: (i) when the outcome is correctly specified, the regression estimator and the weighting estimator have similar performance; (ii) when the regression model is misspecified, the weighting estimators with interaction balance constraints can achieve smaller RMSEs. 
	
	\begin{figure}[h!]
		\centering
		\includegraphics[scale=0.45]{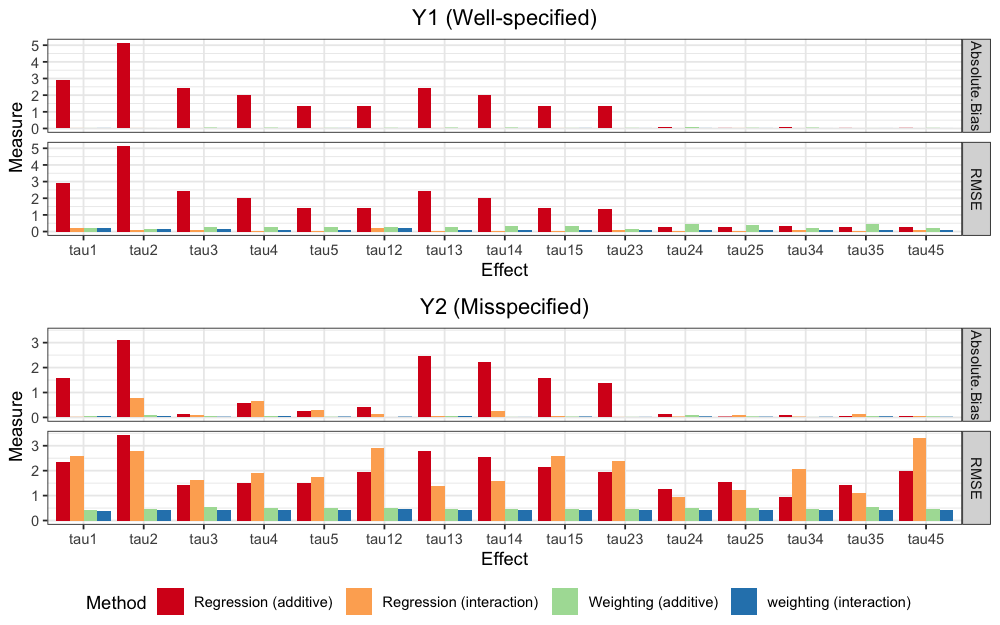}
		\caption{Bias and root mean squared error (RMSE) for estimating five main effects and ten two-way interactions over 1000 repetitions using four estimators when $N=1000$ and $\rho=0.4$.}
		\label{fig:est5}
	\end{figure}
	
	\subsection{Factorial Observational Studies with Rare Treatment Combination Groups}
	In this set of simulations, consider a setting where all treatment combination groups are observed but some of them are rare. We consider three factors $\bZ_i=(Z_{i1},Z_{i2},Z_{i3})$ and five covariates $\bX_i=(X_{i1},\dots,X_{i5})$ for each individual $i$. The covariates $\bX_i$ follows a multivariate normal distribution $\cN(\mathbf{\mu},\Sigma)$, where $\mathbf{\mu}=(0,0,0,0,0)^\textup{T}$ and $\Sigma$ has diagonal elements 1 and off-diagonal elements $\rho=0.4$. The treatment assignment mechanism for $Z_{ik}$ is independent across $k$'s and  satisfies a logistic regression that $\PP(Z_{ik}=1)=1/(1+\exp(-\alpha_k-\boldsymbol{\beta}_k^\textup{T}\bX_i))$, where $\alpha_1=-2,\alpha_2=0,\alpha_3=2$ and $\boldsymbol{\beta}_1=(1/4,2/4,0,3/4,1), \boldsymbol{\beta}_2=(3/4,1/4,1,0,2/4), \boldsymbol{\beta}_3=(1,0,3/4,2/4,1/4)$. Suppose only the main effects of the three treatments are non-negligible. We consider three outcome models: an additive outcome $Y_{i1}=2\sum_{j=1}^5X_{ij}+\sum_{k=1}^3Z_{ik}+\epsilon_{i1}$, a heterogeneous treatment effect outcome $Y_{i2}=2\sum_{j=1}^5X_{ij}+\sum_{k=1}^3Z_{ik}+\sum_{j=1}^5X_{ij}\sum_{k=1}^3Z_{ik}+\epsilon_{i2}$, and a misspecified outcome $Y_{i3}=4\sin(X_{i1})+\exp(0.4X_{i2}^2)+\left(\min(1,X_{i1})+X_{i2}\right)Z_{i1}+X_{i1}Z_{i2}+\sum_{j=1}^5X_{ij}Z_{i3}+\epsilon_{i3}$, where $\epsilon_{ij}$'s are independent errors following a standard normal distribution. We consider two candidate ways of using the proposed weighting methods (with interaction balance constraints) to estimate main effects $\tau_k, k=1,2,3$, under each outcome model: 
	(i) use the entire observed data to emulate a full factorial design; (ii) discard the two rare treatment combination groups ($\mathbf{z}=(1,-1,-1)$ and $(1,1,-1)$) and use the remaining data to emulate an incomplete factorial design. In each simulation setting, we compare the bias and root mean squared error (RMSE) using 1000 repetitions. We summarize the results in Table~\ref{tb:rare}. We can see that it is generally safer to use the whole observed data (option 1), which can give us smaller bias and RMSE. 
	
	\begin{table}[ht]
		\centering
		\caption{Performance with rare treatment combination groups: Bias and root mean squared error (RMSE)}
		\resizebox*{0.7\textwidth}{!}{
			\begin{tabular}{llrrrrr}
				\hline	\hline
				Outcome & Effect & \multicolumn{2}{c}{Bias} && \multicolumn{2}{c}{RMSE}  \\  \cline{3-4}\cline{6-7}
				&&Full &Incomplete && Full & Incomplete\\ \hline
				\multirow{3}{*}{$Y_1$ (Additive)}&$\tau_1$ & 0.062 & -0.370 && 0.317 & 3.917 \\ 
				&$\tau_2$ &-0.008 & -0.250 && 0.309 & 3.484 \\ 
				&$\tau_3$ & 0.101 & -0.200 && 0.325 & 4.855 \\  \hline
				\multirow{3}{*}{$Y_2$ (Heterogeneous)}&$\tau_1$ & 0.082 & -0.352 && 0.453 & 4.825 \\ 
				&$\tau_2$ &-0.028 & -0.257 && 0.430 & 3.406 \\ 
				&$\tau_3$ &-0.005 & -0.221 && 0.412 & 3.842 \\  \hline
				\multirow{3}{*}{$Y_3$ (Misspecified)}&$\tau_1$ & -0.047 & -0.117 && 0.878 & 3.593 \\
				&$\tau_2$ &-0.126 & -0.014 && 0.891 & 3.004 \\ 
				&$\tau_3$ &0.060 & 0.071 && 0.902 & 3.919 \\ 
				\hline
		\end{tabular}}
		\label{tb:rare}
	\end{table}
	
	\section{NHANES Study: VOCs Exposure and Tachycardia}\label{app:nhanes}
	\subsection{Sources of VOCs Exposure}
	According to the United States Environmental Protection Agency, Benzene and Ethylbenzene exposures can come from both outdoor sources (tobacco smoke, gas stations, motor vehicle exhaust, and industrial emissions) and indoor sources (glues, paints, furniture wax, and detergents). The primary source of 1,4-Dichlorobenzene is from breathing vapors from 1,4-Dichlorobenzene products used in the home, such as mothballs and toilet deodorizer blocks. People can be exposed to MTBE by breathing contaminated air and drinking contaminated water from gasoline and fuel facilities. 
	
	\subsection{Descriptive Statistics for the NHANES Study}
	In the NHANSE Study, the four factors give us $2^4=16$ treatment combination groups. Table~\ref{tb:descriptive} describes the sample size and covariate means in each group.
	We can observe that the sample sizes and the covariate distributions vary across treatment combination groups. 
	\begin{table}[h!]
		\centering
		\caption{Descriptive statistics for $2^4=16$ treatment combination groups in the NHANES study, with their sample size and covariate means.}
		\resizebox{0.9\textwidth}{!}{
			\begin{tabular}{rrrrrrrrrrr}  \hline\hline
				\multicolumn{4}{c}{Treatment Combination}&Sample&\multicolumn{6}{c}{Covariate Mean}\\\cline{1-4}\cline{6-11}
				Benzene & 1,4-Dichlorobenzene & Ethylbenzene & MTBE & Size & Age & Black & Povertyr & Hypertension & Alcohol & Smoking \\   \hline
				-1 & -1 & -1 & -1 & 322 & 34.42 & 0.11 & 3.08 & 0.19 & 0.71 & 0.02 \\   
				-1& -1 & -1 & +1& 139 & 33.13 & 0.08 & 2.98 & 0.20 & 0.64 & 0.02 \\   
				-1 & -1 & +1 & -1 &  86 & 35.76 & 0.08 & 3.06 & 0.27 & 0.74 & 0.02 \\   
				-1 & -1 & +1 & +1 &  66 & 34.48 & 0.15 & 3.13 & 0.18 & 0.65 & 0.02 \\   
				-1 & +1 & -1 & -1 & 386 & 33.72 & 0.29 & 2.43 & 0.20 & 0.57 & 0.02 \\   
				-1 & +1 & -1 & +1 & 126 & 32.98 & 0.41 & 2.30 & 0.25 & 0.52 & 0.01 \\   
				-1 & +1 & +1 & -1 &  74 & 29.69 & 0.18 & 2.09 & 0.09 & 0.59 & 0.04 \\   
				-1 & +1 & +1 & +1 &  75 & 34.20 & 0.35 & 2.22 & 0.13 & 0.60 & 0.05 \\   
				+1 & -1 & -1 & -1 &  21 & 32.33 & 0.10 & 2.96 & 0.24 & 0.81 & 0.24 \\   
				+1 & -1 & -1 & +1 &  27 & 32.70 & 0.11 & 2.25 & 0.22 & 0.63 & 0.11 \\   
				+1& -1& +1 & -1 & 158 & 34.74 & 0.13 & 2.08 & 0.28 & 0.78 & 0.75 \\   
				+1 & -1 & +1 & +1& 147 & 33.80 & 0.11 & 2.74 & 0.20 & 0.76 & 0.45 \\   
				+1 & +1 & -1 & -1 &  34 & 31.44 & 0.21 & 2.17 & 0.35 & 0.50 & 0.21 \\   
				+1 & +1& -1 & +1 &  15 & 32.73 & 0.13 & 2.33 & 0.13 & 0.60 & 0.07 \\   
				+1 & +1& +1 & -1 & 203 & 34.94 & 0.32 & 1.82 & 0.26 & 0.79 & 0.72 \\   
				+1 & +1 & +1 & +1 & 124 & 34.05 & 0.33 & 1.98 & 0.27 & 0.68 & 0.49 \\    
				\hline\hline
		\end{tabular}}
		\label{tb:descriptive}
	\end{table}
	
	\subsection{Covariate Balance Before and After Weighting}\label{app:covbal}
	To remove confounding due to the observed covariates and draw reliable causal conclusions, the proposed weighting estimator (with interaction balance constraints and first-order basis functions) is applied to reduce the biases from the observed covariates. We evaluate the covariate balance in terms of absolute standardized mean differences and summarize the results in Figure~\ref{fig:balance}. We can observe that our weighting method greatly improved covariate balances in all contrasts for main effects and two-way interactions. 
	\begin{figure}[h!]
		\centering
		\includegraphics[scale=0.5]{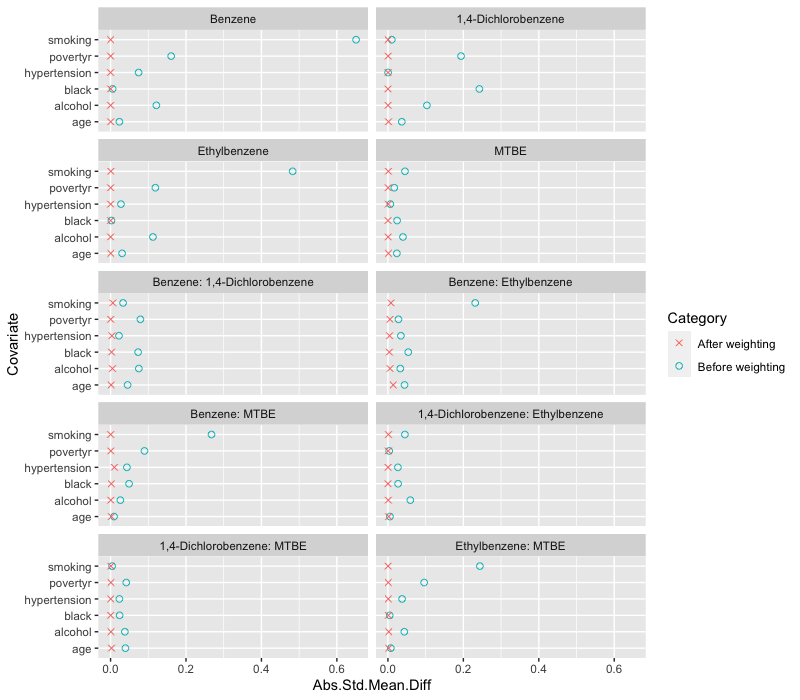}
		\caption{Absolute standardized mean differences before and after weighting.}
		\label{fig:balance}
	\end{figure}

	\section{Discussions and Examples}\label{app:discussion}
	
	\subsection{Bias Decomposition Under the General Additive Outcome Model}\label{app:bias}
	In this subsection, we give more details for equation (3) in \S 3.1.1. 
	With the general additive outcome model in Assumption~2, we can decompose the 
	bias of estimating $\tau_{\cK}$ into two components:
	
	\begin{align*}
		\mathbb{E}\left[\widehat\tau^+_{\cK}-\tau^+_{\cK}\mid\{\mathbf{X}\}_{i=1}^N,\{\mathbf{Z}\}_{i=1}^N\right]
		=&\mathbb{E} \left[\frac{1}{N}\sum_{i=1}^N w_iA_{i{\cK}}^+Y_i^\textup{obs} -\frac{1}{2^{K-1}}\sum_{\mathbf{z}\in\mathcal{Z}}g_{{\cK}\mathbf{z}}^+\mathbb{E}[Y(\mathbf{z})]\ \middle|\ \{\mathbf{X}\}_{i=1}^N,\{\mathbf{Z}\}_{i=1}^N\right]\\
		=&\frac{1}{N}\sum_{i=1}^N w_iA_{i{\cK}}^+\mathbb{E}[Y_i(\mathbf{Z}_i)\mid\mathbf{X}_i] -\frac{1}{2^{K-1}} \sum_{\mathbf{z}\in\mathcal{Z}}g_{{\cK}\mathbf{z}}^+\mathbb{E}[Y(\mathbf{z})]\\
		=&\frac{1}{N}\sum_{i=1}^N w_iA_{i{\cK}}^+\left(\sum_{s=1}^S\alpha_s h_s(\mathbf{X}_i)+\sum_{J\in\bKK}\beta_{J}\prod_{j\in J}Z_{ij}\right)\\
		&-\frac{1}{2^{K-1}} \sum_{\mathbf{z}\in\mathcal{Z}}g_{{\cK}\mathbf{z}}^+\left(\sum_{s=1}^S\alpha_s \mathbb{E}[h_s(\mathbf{X}_i)]+\sum_{J\in\bKK}\beta_{J}\prod_{j\in J}z_{j}\right)\\
		=&\sum_{s=1}^S\alpha_s\left(\frac{1}{N}\sum_{i=1}^N w_iA_{i{\cK}}^+h_s(\mathbf{X}_i)-\mathbb{E}[h_s(\mathbf{X}_i)]\right)\\
		&+\sum_{J\in\bKK}\beta_J\left(\frac{1}{N}\sum_{i=1}^N w_iA_{i{\cK}}^+\prod_{j\in J}Z_{ij}-\frac{1}{2^{K-1}} \sum_{\mathbf{z}\in\mathcal{Z}}g_{{\cK}\mathbf{z}}^+\prod_{j\in J}z_{j}\right),
	\end{align*}
	and similarly,
	\begin{align*}
		\mathbb{E} \left[\widehat\tau^-_{\cK}-\tau^-_{\cK}\mid \{\mathbf{X}\}_{i=1}^N,\{\mathbf{Z}\}_{i=1}^N\right]
		=&\sum_{s=1}^S\alpha_s\left(\frac{1}{N}\sum_{i=1}^N w_iA_{i{\cK}}^-h_s(\mathbf{X}_i)-\mathbb{E}[h_s(\mathbf{X}_i)]\right)\\
		&+\sum_{J\in\bKK}\beta_J\left(\frac{1}{N}\sum_{i=1}^N w_iA_{i{\cK}}^-\prod_{j\in J}Z_{ij}-\frac{1}{2^{K-1}} \sum_{\mathbf{z}\in\mathcal{Z}}g_{{\cK}\mathbf{z}}^-\prod_{j\in J}z_{j}\right).
	\end{align*}
	
	\subsection{A Toy Example of Estimating Main Effects ($K'=1$)}\label{app:diff}
	To highlight the difference between the general additive outcome model in 
	Assumption~2 and the heterogeneous treatment effect outcome model in 
	Assumption~3  in the main paper, we examine a simple scenario in which only the main effects are non-zero (i.e., $K'=1$). That is, we are interested in estimating the main effects $\tau_\cK, {\cK}\in{[K]}_1$ (i.e., ${\cK}=1,\dots,K$). 
	
	\begin{itemize}
		\item 
		The general additive outcome model becomes
		$$\mathbb{E}[Y_i(\mathbf{z})\mid\mathbf{X}_i]=\sum_{s=1}^S\alpha_s h_s(\mathbf{X}_i)+\sum_{k=1}^K\beta_{k}z_k.$$ The analysis in Section~3.1.1 suggests that we need the following balance constraints to achieve an unbiased estimate of the main effect $\tau_{\cK}$: $\textrm{for all } s=1,\dots, S$, $k=1,\dots,K$, 
		\begin{align*}
			\sum_{i=1}^N w_iA_{i{\cK}}^+h_s(\mathbf{X}_i)=\sum_{i=1}^N h_s(\mathbf{X}_i), &\quad \frac{1}{N}\sum_{i=1}^N w_iA_{i{\cK}}^+Z_{ik}=\frac{1}{2^{K-1}}\sum_{\mathbf{z}\in\mathcal{Z}} g_{{\cK}\mathbf{z}}^+z_{k},\\
			\sum_{i=1}^N w_iA_{i{\cK}}^-h_s(\mathbf{X}_i)= \sum_{i=1}^N h_s(\mathbf{X}_i), &\quad \frac{1}{N}\sum_{i=1}^N w_iA_{i{\cK}}^-Z_{ik}=\frac{1}{2^{K-1}}\sum_{\mathbf{z}\in\mathcal{Z}} g_{{\cK}\mathbf{z}}^-z_k.
		\end{align*}
		
		\item 
		The outcome model with treatment effect heterogeneity becomes $$\mathbb{E}[Y_i(\mathbf{z})\mid\mathbf{X}_i]=\sum_{s=1}^S h_{s}(\mathbf{X}_i) \sum_{k=1}^K\alpha_{sk}z_k.$$ To obtain unbiased estimates under this relaxed model, we require a different set of balance constraints as discussed in Section~3.1.2: $\textrm{ for all }s=1,\dots, S$, $k=1,\dots,K$, 
		\begin{align*}
			\sum_{i=1}^Nw_iA_{i{\cK}}^+h_s(\mathbf{X}_i)Z_{ik}&=\frac{\sum_{\mathbf{z}\in\mathcal{Z}}g_{{\cK}\mathbf{z}}^+z_k}{2^{K-1}}\sum_{i=1}^N h_s(\mathbf{X}_i), \\
			\sum_{i=1}^N w_iA_{i{\cK}}^-h_s(\mathbf{X}_i)Z_{ik}&= \frac{\sum_{\mathbf{z}\in\mathcal{Z}}g_{{\cK}\mathbf{z}}^-z_{k}}{2^{K-1}}\sum_{i=1}^N h_s(\mathbf{X}_i).
		\end{align*}
	\end{itemize}
	The comparison above indicates that we need to balance $2S+2K$ constraints under the general additive outcome model and $2SK$ constraints under the heterogeneous treatment effect outcome model, where $S$ denotes the number of basis functions for covariates and $K$ represents the number of treatments. With multiple factors ($K\geq 2$), the heterogeneous treatment effect outcome model requires more balance constraints when there are more than two basis functions ($S>2$), making the problem more challenging but also enhancing the robustness of the conclusion. On the other hand, as discussed in Section~3.1.1, when there is only one treatment (i.e., $K=1$), the balancing constraints for the treatments are equivalent to those for a constant covariate. Thus, we need to balance $2S$ balancing constraints under both models.
	
	\subsection{Derivation of the Dual Problem of Problem (5)}\label{app:dual}
	
	Recall that we write the linear constraints of problem~(5) as $\mathbf{Bw}=\bb$, where the matrix $\bB$ has $i$th column $\bB_{i}$ representing all balance criteria evaluated at individual $i$ and $\bb=\sum_{i=1}^N \bb_i$. Specifically, 
	\begin{align*}
		\bB_i&=\left\{A_{i{\cK}}^+q_{sJ}(\bX_i,\bZ_i), A_{i{\cK}}^-q_{sJ}(\bX_i,\bZ_i)\right\}_{{\cK}\in\bKK,\ s=1,\dots,S,\ J\in\bKK},\\ \bb_i&=\left\{\frac{1}{2^{K-1}}\sum_{\bz\in\mathcal{Z}} g_{{\cK}\bz}^+q_{sJ}(\bX_i,\bz),\frac{1}{2^{K-1}}\sum_{\bz\in\mathcal{Z}} g_{{\cK}\bz}^-q_{sJ}(\bX_i,\bz)\right\}_{{\cK}\in\bKK,\ s=1,\dots,S,\ J\in\bKK}.
	\end{align*}
	Since $(\bZ_i, \bX_i)$'s are independent and identically distributed, it follows that $(\bB_{i}, \bb_{i})$, $i=1,\dots,N$, are independent and identically distributed. We can rewrite the balance constraints $\mathbf{Bw}=\bb$ as $\sum_{i=1}^N \bB_{i} w_i=\sum_{i=1}^N \bb_i$. Then, the corresponding Lagrangian is given by
	\begin{equation}\label{prob: lagrangian}
		\cL(\mathbf{w};\blambda,\bgamma)=\sum_{i=1}^N m(w_i)+\blambda^\textup{T}(\mathbf{Bw}-\mathbf{b})-\bgamma^\textup{T}\mathbf{w},
	\end{equation}
	where $\blambda$ and $\bgamma\ge 0$ are the Lagrangian multipliers for the balance constraints and non-negativity requirements, respectively. The Lagrangian relaxation of problem~(5) minimizes the Lagrangian in equation~(\ref{prob: lagrangian}). For any set of Lagrangian multipliers $(\blambda, \bgamma)$, the corresponding optimal solution of weights can be obtained by setting 
	$
	\frac{\partial}{\partial w_i} \cL(\mathbf{w};\blambda,\bgamma) =m'(w_i)+\blambda^\textup{T}\bB_{i}-\gamma_i
	=0.
	$
	We can then obtain the closed-form optimal solution 
	$$\widehat w_i(\blambda,\bgamma)=(m')^{-1}(\gamma_i-\blambda^\textup{T}\bB_{i}),$$
	with optimal objective value 
	\begin{align*}
		\cL^*(\boldsymbol{\lambda},\boldsymbol{\gamma})
		&=\min_\mathbf{w} \cL(\mathbf{w};\boldsymbol{\lambda},\boldsymbol{\gamma})\\
		&=\sum_{i=1}^N\left\{m(\widehat w_i(\boldsymbol{\lambda},\boldsymbol{\gamma}))+\boldsymbol{\lambda}^\textup{T}\left(\mathbf{B}_{i}\widehat w_i(\boldsymbol{\lambda},\boldsymbol{\gamma})-\mathbf{b}_i\right)-\gamma_i \widehat w_i(\boldsymbol{\lambda},\boldsymbol{\gamma})\right\} \\
		&=\sum_{i=1}^N\left\{m\left((m')^{-1}\left(\gamma_i-\boldsymbol{\lambda}^\textup{T}\mathbf{B}_{i}\right)\right)+\boldsymbol{\lambda}^\textup{T}\left(\mathbf{B}_{i}(m')^{-1}(\gamma_i-\boldsymbol{\lambda}^\textup{T}\mathbf{B}_{i})-\mathbf{b}_i\right)-\gamma_i(m')^{-1}\left(\gamma_i-\boldsymbol{\lambda}^\textup{T}\mathbf{B}_{i}\right)\right\} \\
		&=\sum_{i=1}^N\{\rho(\gamma_i-\boldsymbol{\lambda}^\textup{T}\mathbf{B}_{i})-\boldsymbol{\lambda}^\textup{T}\mathbf{b}_i\},
	\end{align*}
	where  $\rho(v)=m((m')^{-1}(v))-v(m')^{-1}(v)$.
	This leads to the dual problem~(6).

	\subsection{Weighting Combined With Outcome Regression}\label{app:aug}
	A common technique in the weighting literature to further improve the performance is to combine the weighting method with outcome regression adjustments \citep{robins1994estimation,rosenbaum2002covariance,athey2018approximate,cohn2023balancing}. This type of augmented estimator is also applicable to our balancing weighting method in a factorial design. Specifically, in the first step of our method, we estimate the outcome regression model using basis functions $\bq(\bX,\bz)$ such that $\EE[Y(\bz)\mid\bX]=\balpha^\textup{T}\bq(\bX_i,\bz)$. The coefficients $\balpha$ can be estimated using ordinary least squares (OLS) or regularization methods. We denote the estimated coefficients as $\widehat{\balpha}_\textup{reg}$. In the second step, we weight the residuals using weights that solve the optimization problem~(5) in the main paper, balancing the basis functions $\bq(\bX,\bz)$ exactly by requiring $\textrm{ for all }\ \cK\in\bKK$,$$\sum_{i=1}^N w_iA_{i{\cK}}^+\bq(\bX_i,\bZ_i)=\frac{1}{2^{K-1}}\sum_{\mathbf{z}\in\mathcal{Z}}g_{{\cK}\mathbf{z}}^+\sum_{i=1}^N \bq(\bX_i,\bz), $$ and
	$$\sum_{i=1}^N w_iA_{i{\cK}}^-\bq(\bX_i,\bZ_i)=\frac{1}{2^{K-1}}\sum_{\mathbf{z}\in\mathcal{Z}}g_{{\cK}\mathbf{z}}^-\sum_{i=1}^N \bq(\bX_i,\bz).$$ Finally, we use the following augmented weighting estimator to estimate the factorial effects $\tau_\cK$, $\cK\in\bKK$:
	$$\widehat\tau_{\cK,\textup{Aug}}=\widehat\tau^+_{\cK,\textup{Aug}}-\widehat\tau^-_{\cK,\textup{Aug}}, $$
	where $$\widehat\tau^+_{\cK,\textup{Aug}}=\frac{1}{N}\sum_{i=1}^N w_i A_{i\cK}^+ (Y_i^\textup{obs}-\widehat{\balpha}^\textup{T}_{\textup{reg}}\bq(\bX_i,\bZ_i))+\frac{1}{2^{K-1}N}\sum_{\bz\in\cZ}g_{\cK\bz}^+\widehat{\balpha}^\textup{T}_{\textup{reg}}\sum_{i=1}^N\bq(\bX_i,\bz)$$
	and $$\widehat\tau^-_{\cK,\textup{Aug}}=\frac{1}{N}\sum_{i=1}^N w_i A_{i\cK}^- (Y_i^\textup{obs}-\widehat{\balpha}^\textup{T}_{\textup{reg}}\bq(\bX_i,\bZ_i))+\frac{1}{2^{K-1}N}\sum_{\bz\in\cZ}g_{\cK\bz}^-\widehat{\balpha}^\textup{T}_{\textup{reg}}\sum_{i=1}^N\bq(\bX_i,\bz).$$
	With exact balance constraints in the construction of weights, we have 
	$$\widehat\tau^+_{\cK,\textup{Aug}}=\widehat\tau^+_{\cK}=\frac{1}{N}\sum_{i=1}^N w_i A_{i\cK}^+ Y_i^\textup{obs}, \quad
	\text{ and } \quad \widehat\tau^-_{\cK,\textup{Aug}}=\widehat\tau^-_{\cK}=\frac{1}{N}\sum_{i=1}^N w_i A_{i\cK}^- Y_i^\textup{obs}.$$
	Therefore, the augmented estimator is equivalent to the original weighting estimator. This observation is consistent with the findings in \cite{bruns2025augmented}, which studies the case with one binary factor. Without exact balance, the augmented estimator can further remove covariance imbalance by outcome regression.

	\subsection{An Illustrating Example with Unobserved Treatment Combinations}\label{app:unobserved}
	To illustrate the idea in \S 4 in the main paper, consider a toy example with $K=3$ binary treatments, i.e., $2^3=8$ treatment combinations. Then we have 
	$$\boldsymbol{\tau}=\begin{bmatrix}
		2\tau_0\\
		\tau_1\\
		\tau_2\\
		\tau_3\\
		\tau_{1,2}\\
		\tau_{1,3}\\
		\tau_{2,3}\\
		\tau_{1,2,3}
	\end{bmatrix},
	\quad
	\mathbf{G}=\begin{bmatrix}
		+1&-1 &-1&-1 &+1&+1&+1&-1\\
		+1&-1 &-1&+1 &+1&-1&-1&+1\\
		+1&-1 &+1&-1 &-1&+1&-1&+1\\
		+1&-1 &+1&+1&-1&-1&+1&-1\\
		+1&+1 &-1&-1&-1&-1&+1&+1\\
		+1&+1 &-1&+1&-1&+1&-1&-1\\
		+1&+1 &+1&-1&+1&-1&-1&-1\\
		+1&+1 &+1&+1&+1&+1&+1&+1
	\end{bmatrix},
	\quad
	\mathbb{E}[\mathbf{Y}]=\begin{bmatrix}
		\mathbb{E}[Y(-1,-1,-1)]\\
		\mathbb{E}[Y(-1,-1,+1)]\\
		\mathbb{E}[Y(-1,+1,-1)]\\
		\mathbb{E}[Y(-1,+1,+1)]\\
		\mathbb{E}[Y(+1,-1,-1)]\\
		\mathbb{E}[Y(+1,-1,+1)]\\
		\mathbb{E}[Y(+1,+1,-1)]\\
		\mathbb{E}[Y(+1,+1,+1)]
	\end{bmatrix}.
	$$
	Suppose $Q_\textup{u}=1$ treatment combination $(+1,+1,+1)$ is not observed and the third-order interaction is negligible, i.e., $K'=2$, $Q_-=1$ and $\tau_{1,2,3}=0$. In order to identify a unique relationship between $\mathbb{E}[Y(+1,+1,+1)]$ and the other expected potential outcomes. We have 
	$$\mathbf{G}_{\textup{o}-}^\textup{T}=\begin{bmatrix}
		-1&+1&+1&-1&+1&-1&-1
	\end{bmatrix},
	\quad \mathbf{G}_{\textup{u}-}^\textup{T}=+1.
	$$
	Since $\mathbf{G}_{\textup{u}-}^\textup{T}$ is invertible, we have 
	\begin{align*}
		\mathbb{E}[Y(+1,+1,+1)]=&\mathbb{E}[Y(-1,-1,-1)]-\mathbb{E}[Y(-1,-1,+1)]-\mathbb{E}[Y(-1,+1,-1)]+\mathbb{E}[Y(-1,+1,+1)]\\
		&\quad\quad\quad\quad\quad\quad\quad-\mathbb{E}[Y(+1,-1,-1)]+\mathbb{E}[Y(+1,-1,+1)]+\mathbb{E}[Y(+1,+1,-1)].
	\end{align*}

	Then we can  identify the factorial effects in the following way:
	\begin{align*}
		\begin{bmatrix}
			2\tau_0\\
			\tau_1\\
			\tau_2\\
			\tau_3\\
			\tau_{1,2}\\
			\tau_{1,3}\\
			\tau_{2,3}
		\end{bmatrix}
		&=\frac{1}{4}\begin{bmatrix}
			+2  &  0  &  0   & +2  &   0   &  +2    & +2\\
			0   & -2   & -2   &  0  &   0   &  +2   &  +2\\
			0   & -2  &   0  &   +2 &   -2  &   0  &   +2\\
			0    & 0   & -2   &  +2  &  -2   &  +2   &  0\\
			+2   &  0  &  -2   &  0   & -2  &   0   &  +2\\   
			+ 2    &-2  &   0   &  0  &  -2   & + 2  &   0\\    
			+ 2   & -2  &  -2 &    +2  &   0  &   0 &    0
		\end{bmatrix}
		\begin{bmatrix}
			\mathbb{E}[Y(-1,-1,-1)]\\
			\mathbb{E}[Y(-1,-1,+1)]\\
			\mathbb{E}[Y(-1,+1,-1)]\\
			\mathbb{E}[Y(-1,+1,+1)]\\
			\mathbb{E}[Y(+1,-1,-1)]\\
			\mathbb{E}[Y(+1,-1,+1)]\\
			\mathbb{E}[Y(+1,+1,-1)]
		\end{bmatrix}.
	\end{align*}
	With this identification of factorial effects, we can formulate the optimization problem analogously to Section~3.2 to find the optimal weights for estimating the non-negligible factorial effects. Here, however,  
	$A_{i\cK}^+$ and $A_{i\cK}^-$ represent the weights that unit $i$ contributes to the positive or negative parts of factorial effect ${\tau}_\cK$. Specifically, let $\mathring{g}_{\cK\mathbf{z}}$ denote the new contrast vector for ${\tau}_\cK$, then we define 
	$A_{i\cK}^+=\sum_{\mathbf{z}\in\mathcal{Z}}\mathring{g}_{\cK\mathbf{z}}^+I(\mathbf{Z}_i=\mathbf{z})$ and $A_{i\cK}^-=\sum_{\mathbf{z}\in\mathcal{Z}}\mathring{g}_{\cK\mathbf{z}}^-I(\mathbf{Z}_i=\mathbf{z})$. In this toy example, $A_{i\cK}^+$ and $A_{i\cK}^-$ can take values of 0 or 2, whereas in the full factorial design they take values of 0 or 1.

	On the other hand, if we have one more unobserved treatment combination, say $Q_\textup{u}=2$, and treatment combinations $(+1,+1,-1)$ and $(+1,+1,+1)$ are not observed, problems arise if only the third-order interaction is negligible, i.e., $Q_-=1$ and $\tau_{1,2,3}=0$. Specifically, we have 
	$$\mathbf{G}_{\textup{o}-}^\textup{T}=\begin{bmatrix}
		-1&+1&+1&-1&+1&-1
	\end{bmatrix},
	\quad \mathbf{G}_{\textup{u}-}^\textup{T}=\begin{bmatrix}
		+1&+1\\
	\end{bmatrix}.
	$$
	We can verify that $\mathbf{G}_{\textup{u}-}$ does not have full row rank. As a result, with two unknowns in a single equation, there are infinitely many feasible pairs for $\mathbb{E}[Y(+1,+1,-1)]$ and $\mathbb{E}[Y(+1,+1,+1)]$, which gives infinitely many contrast forms for the factorial effects of interest. Note that everything would work well if we further assumed that second-order interactions are also negligible. Therefore, having a full row rank $\mathbf{G}_{\textup{u}-}$ is necessary and sufficient to identify the non-negligible factorial effects.

\section{Technical Results and Proofs}\label{app:proof}
\subsection{Existence of Balance Weights: Proposition~\ref{thm:feasible} and Proof}\label{app:prop1}
\begin{proposition}\label{thm:feasible}
Suppose $h_s(\bx), s=1,\dots,S$ are bounded orthonormal basis functions of the covariate space with covariate distribution $f(\mathbf{x})$, i.e., 
$\int h_s(\bx)h_t(\bx)f(\bx)\textup{d}\bx=1$ if $s=t$ and $\int h_s(\bx)h_t(\bx)f(\bx)\textup{d}\bx=0$ otherwise.
Then problem~(5) is feasible with probability greater than $1-1/N$ if $(S^2\log S) /N\to0$ as $N \to \infty$.
\end{proposition}

We first present and prove Lemma~\ref{thm:orthonormal}, which is useful for proving Proposition~\ref{thm:feasible}.
\begin{lemma}\label{thm:orthonormal}
For any set of bounded orthonormal basis functions $h_s(\bx),s=1,\dots,S$ of the covariate space with covariate distribution $f(\bx)$, the $S$ length-$N$ vectors $(h_s(\bX_i))_{i=1,\dots,N}$ are linearly independent with probability greater than $1-1/N$ if $(S^2\log S) /N\to0$.
\end{lemma}
\begin{proof}[Proof of Lemma~\ref{thm:orthonormal}]
To prove Lemma~\ref{thm:orthonormal}, we need to show that 
$$\sum_{s=1}^S a_{s}h_s(\bX_i)=0,\textrm{ for all } i=1,\dots,N \quad \text{ if and only if } \quad a_{s}=0,\textrm{ for all } s=1,\dots,S.$$
When $a_{s}=0,\textrm{ for all } s=1,\dots,S$, it is clear that $\sum_{s=1}^S a_{s}h_s(\bX_i)=0,\textrm{ for all } i=1,\dots,N$. We only need to prove the other direction. Suppose $\sum_{s=1}^S a_{s}h_s(\bX_i)=0,\textrm{ for all } i=1,\dots,N$. Without loss of generality, we assume $|a_s|\leq1,\textrm{ for all } s=1,\dots,S$. Then for any $t=1,\dots,S$, we have $$\sum_{s=1}^S a_{s}h_s(\bX_i)h_t(\bX_i)=0,\textrm{ for all } i=1,\dots,N.$$ Hence we have $$\sum_{s=1}^S a_{s}\left[\frac{1}{N}\sum_{i=1}^N h_s(\bX_i)h_t(\bX_i)\right]=0.$$

On the other hand, since $h_s(\bx),s=1,\dots,S$ are orthonormal, we have $$\mu_{st}=\int h_{s}(\bx)h_{t}(\bx)f(\bx)d\bx=\begin{cases}
	0&\text{ if } s\neq t\\
	1&\text{ if } s= t\\
\end{cases}.$$
Suppose the bounded basis functions satisfy $|h_s(x)|\leq u, \textrm{ for all } s=1,\dots,S$. By Hoeffding's inequality, we have $$\PP\left(\left|\frac{1}{N}\sum_{i=1}^N h_s(\bX_i)h_t(\bX_i)-\mu_{st}\right|>\epsilon\right)\leq2\exp\left(-\frac{N\epsilon^2}{2u^2}\right).$$

By applying the union bound, we have 
\begin{align*}
	&\PP\left(\left|\frac{1}{N}\sum_{i=1}^N h_s(\bX_i)h_t(\bX_i)-\mu_{st}\right|\leq\epsilon, \ \textrm{ for all } s,t=1,\dots,S\right)\\
	\leq&1-\sum_{s,t=1}^S \PP\left(\left|\frac{1}{N}\sum_{i=1}^N h_s(\bX_i)h_t(\bX_i)-\mu_{st}\right|>\epsilon\right)\\
	\leq&1-2S^2\exp\left(-\frac{N\epsilon^2}{2u^2}\right).
\end{align*}
By choosing $$\epsilon=\sqrt{\frac{2\log 2 u^2+2u^2\log N+4u^2\log S}{N}}=O\left(\sqrt{\frac{\log N}{N}}\right),$$ we have 
$$\PP\left(\left|\frac{1}{N}\sum_{i=1}^N h_s(\bX_i)h_t(\bX_i)-\mu_{st}\right|\leq\epsilon, \ \textrm{ for all } s,t=1,\dots,S\right)\geq1-1/N. $$
Conditional on $\left|\frac{1}{N}\sum_{i=1}^N h_s(\bX_i)h_t(\bX_i)-\mu_{st}\right|\leq\epsilon, \ \textrm{ for all } s,t=1,\dots,S$, we have $$\left|\sum_{s=1}^S a_{s}\left[\frac{1}{N}\sum_{i=1}^N h_s(\bX_i)h_t(\bX_i)\right]- a_t\right|\leq\sum_{s=1}^S |a_s|\epsilon, \textrm{ for all } t=1,\dots,S.$$
Suppose not all coefficients $a_s$ are zero. We let $a_t$ denote the one with the largest absolute value. 
When $(S^2\log S) /N\to0$, we have $\epsilon\leq1/S<|a_t|/(\sum_s|a_s|)$ and hence $$\sum_{s=1}^S a_{s}\left[\frac{1}{N}\sum_{i=1}^N h_s(\bX_i)h_t(\bX_i)\right]\neq 0, $$ which leads to a contradiction.
Therefore, the linear independence holds with probability at least $1-1/N$.
\end{proof}

We now prove Proposition~\ref{thm:feasible}.
\begin{proof}[Proof of Proposition~\ref{thm:feasible}]
The constraints in problem~(5) have a clean and symmetric form, corresponding to the positive and negative parts of each contrast $\bg_{\cK}$, where ${\cK}\in\bKK$. However, some of these constraints might be redundant. For instance, 
\begin{itemize}
	\item For all ${\cK}\in\bKK, s=1,\dots, S, J\in\bKK$, we have a pair of constraints, corresponding to the positive and negative parts of contrast $\bg_{\cK}$. For any fixed $s=1,\dots, S, J\in\bKK$, the sum of these two constraints is the same across $\cK\in\bKK$, which has the following form: 
	\begin{equation}\label{sumcons}
		\sum_{i=1}^N w_iq_{sJ}(\bX_i,\bZ_i)=\frac{1}{2^{K-1}}\sum_{\mathbf{z}\in\mathcal{Z}}\sum_{i=1}^N q_{sJ}(\bX_i,\bz). \tag{$*$}
	\end{equation}
	We call the constraint in \eqref{sumcons} a summary constraint. Therefore, we only need to consider the summary constraint and all the balancing constraints corresponding to the positive part of each contrast. This summary constraint can be written in the same form as other balancing constraints in the following way: for ${\cK}=\emptyset$, we denote the corresponding contrast vector as $\bg_{\cK}=(+1,\dots,+1)$ so $A_{i{\cK}}^+=1$ for all $i=1,\dots,N$. For notation convenience, we define $\bKK':=\bKK\cup\emptyset$.
	
	\item For each ${\cK}\in\bKK'$, several choices of  $J\in\bKK$ lead to the same balance constraint, so only one of them needs to stay in our optimization problem.
	Specifically, when $J\supset {\cK}$, the corresponding balancing constraints are the same as those created by choosing $J$ as the difference between these two sets $J-{\cK}$. 
\end{itemize}

Now we remove the redundant constraints and focus on the optimization problem~(5) with refined constraints. For each ${\cK}\in\bKK'$, let $\bKK({\cK})$ denote the set of $J\in\bKK$ such that $q_{sJ}(\bX_i,\bZ_i)$ remains as refined constraints. 
In addition, we write the revised matrix version of the balance constraints as $\widetilde{\bB}\bw=\widetilde{\bb}$. Next, we show that problem~(5) is feasible with a high probability if the sample size is large enough with respect to the number of balancing constraints.

For notational convenience, we let  $r_J(\bz)=\prod_{j\in J}z_j, J\in\bKK$ denote the basis functions for the treatment space. Then we have $q_{sJ}(\bX_i,\bz)=h_s(\bX_i)r_J(\bz)$. 
By \cite{dasgupta2015causal}, we know that $r_J(\bz), J\subseteq[K]$ are orthogonal to each other. 

To prove Proposition~\ref{thm:feasible}, it suffices to focus on the linear independence of  rows of matrix $\widetilde{\bB}$, i.e., 
$$\sum_{{\cK}\in\bKK'}\sum_{s=1}^S\sum_{J\in\bKK({\cK})}a_{s\cK J}A_{i{\cK}}^+q_{sJ}(\bX_i,\bZ_i)=0, \textrm{ for all } i=1,\dots,N $$ if and only if $$a_{s\cK J}=0, \textrm{ for all } \cK\in\bKK', s=1,\dots,S, J\in\bKK({\cK}).$$

\paragraph{Proof of ``if''}
When $a_{s\cK J}=0, \textrm{ for all }  \cK\in\bKK', s=1,\dots,S, J\in\bKK(\cK),$ it is clear that $$\sum_{{\cK}\in\bKK'}\sum_{s=1}^S\sum_{J\in\bKK(\cK)}a_{s\cK J}A_{i{\cK}}^+q_{sJ}(\bX_i,\bZ_i)=0, \textrm{ for all } i=1,\dots,N.$$

\paragraph{Proof of ``only if''}
Suppose 	for all $i=1,\dots,N$,  $$\sum_{{\cK}\in\bKK'}\sum_{s=1}^S\sum_{J\in\bKK({\cK})}a_{s{\cK}J}A_{{\cK}}^+q_{sJ}(\bX_i,\bZ_i)=\sum_{{\cK}\in\bKK'}\sum_{s=1}^S\sum_{J\in\bKK({\cK})}a_{st\cK}A_{i{\cK}}^+h_s(\bX_i)r_J(\bZ_i)=0.$$

Since $A_{i{\cK}}^+=2r_{\cK}(Z_i)-1$ and $r_{\cK}(\bz), \cK\subseteq[K]$ are linearly independent, we have $A_{i{\cK}}^+$ are also linearly independent for $\cK\in\bKK'$. Therefore, we have
$$\sum_{s=1}^S\sum_{J\in\bKK(\cK)}a_{s\cK J}h_s(\bX_i)r_J(\bZ_i)=0,\textrm{ for all } i=1,\dots,N; \cK\in\bKK'.$$
Again, since $r_J(\bz), J\subseteq[K]$ are linearly independent, we have 
$$\sum_{s=1}^S a_{s\cK J}h_s(\bX_i)=0,\textrm{ for all } i=1,\dots,N;\cK\in\bKK';J\in\bKK(\cK).$$
Since $h_s(\bx), s=1,\dots,S$ are orthonormal, by Lemma~\ref{thm:orthonormal}, we have $$a_{s\cK J}=0,\textrm{ for all } i=1,\dots,N;\cK\in\bKK';J\in\bKK(\cK);s=1,\dots,S$$
with probability greater than $1-1/N$ if $(S^2\log S) /N\to0$.
\end{proof}

\subsection{Technical Assumptions for Main Theorems}\label{app:assumption}
\begin{assumption}\label{asreg}
The following regularity conditions hold:
\begin{enumerate}[label=\textbf{S1.\arabic*}:, noitemsep, leftmargin=40pt]		  
	\item $\EE\left[h_{s_1}(\bX_i)h_{s_2}(\bX_i)h_{s_3}(\bX_i)h_{s_4}(\bX_i)\right]<\infty$ for all $s_1, s_2, s_3, s_4=1,\dots,S$.
	
	\item  The minimum eigenvalue of $\EE\left[\bB_i\bB_i^\textup{T}I({\blambda}^{*\textup{T}}\bB_i<0)\right]$ is positive, where \\ $\blambda^*={\rm argmax}_{\blambda}\EE\left[-\frac{1}{4}(\blambda^\textup{T}\mathbf{B}_{i})^2I(\blambda^\textup{T}\mathbf{B}_{i}<0)-\boldsymbol{\lambda}^\textup{T}\mathbf{b}_i\right]$.
	\item  $Y_i(\mathbf{z})=\mathbb{E}[Y_i(\mathbf{z})\mid\mathbf{X}_i]+\epsilon_i$, where $\epsilon_i$'s are independent errors with $\mathbb{E}[\epsilon_i]=0$ and $\mathbb{E}[\epsilon_i^2]=\sigma^2_i<\infty$. Further, the errors are independent of $\mathbf{Z}$ and $\mathbf{X}$.
	\item   The variances satisfy $N^{-1}\sum_{i=1}^N\sigma_i^2 \to \bar\sigma^2$ as $N\to\infty$ for some constant $\bar\sigma^2$.
	\item  
	The errors $\epsilon_i$'s satisfy that as $N \to \infty$,
	$\sum_{i=1}^N\EE\left[|\epsilon_i|^3\right]/\left(\sum_{i=1}^N \sigma^2_i\right)^{3/2}\to0$ and \\ $\frac{\sum_{i=1}^N\EE\left[\left|(\gamma_{i}^*-{\blambda}^{*\textup{T}}\mathbf{B}_{i})\epsilon_i/2 +\frac{1}{2^{K-1}}\sum_{\mathbf{z}\in\mathcal{Z}}g_{{\cK}\mathbf{z}}\mathbb{E}[Y_i(\mathbf{z})|\mathbf{X}_i-\frac{1}{2^{K-1}}\sum_{\mathbf{z}\in\mathcal{Z}}g_{{\cK}\mathbf{z}}\mathbb{E}[Y_i(\mathbf{z})] ]\right|^3\right]}{\left(N\sigma^2_{\cK}+\sum_{i=1}^N  \sigma^2_{\epsilon i}\right)^{3/2}}\to0, $ \\where
	$\gamma_{i}^*={\blambda}^{*\textup{T}}\mathbf{B}_{i}I\left({\blambda}^{*\textup{T}}\mathbf{B}_{i}\geq0\right),$
	$
	\sigma^2_{\cK}={\rm Var}\left(\frac{1}{2^{K-1}}\sum_{\mathbf{z}\in\mathcal{Z}}g_{{\cK}\mathbf{z}}\mathbb{E}[Y_i(\mathbf{z})\mid\mathbf{X}_i]\right),$ and \\ $\sigma^2_{\epsilon i}=\frac{\sigma^2_i}{4}\blambda^{*\textup{T}}\EE\left[\bB_i\bB_i^\textup{T} I(\blambda^{*\textup{T}}\bB_i<0)\right]\blambda^*.
	$
\end{enumerate}
\end{assumption}

\subsection{Proof of Theorem~1}\label{app:thm1}
To prove Theorem~1, we first present some useful lemmas.

\begin{lemma}[Theorem 19.4 in \cite{van2000asymptotic}]\label{lem:gc}
Every class $\mathcal{F}$ of measurable functions such that the bracketing number $N_{[]}(\epsilon,\mathcal{F}, L_1(\PP)) < \infty$ for every $\epsilon > 0$ is $\PP$-Glivenko-Cantelli.
\end{lemma}
\begin{lemma}[Theorem 4.1.2 in \cite{amemiya1985advanced}]\label{lem:amemiya}
Suppose the following assumptions hold:
\begin{itemize}
	\item[(A)] Let $\Theta$ be an open subset of the Euclidean space. (Thus, the true value $\theta_0$ is an interior point of $\Theta$.)
	\item[(B)] $\zeta_N(\bx; \theta)$ is a measurable function of $\bx=(x_1,\dots,x_N)$ for all $\theta\in\Theta$, and its derivative with respect to $\theta$, denoted as $\zeta_N'(\bx; \theta)=\partial \zeta_N(\bx; \theta)/\partial\theta$, exists and is continuous in an open neighborhood $\mathcal{I}_1(\theta_0)$ of $\theta_0$. (Note that this implies $\zeta_N(\bx; \theta)$ is continuous for $\theta\in \mathcal{I}_1$.)
	\item[(C)] There exists an open neighborhood $\mathcal{I}_2(\theta_0)$ of $\theta_0$ such that $N^{-1}\zeta_N(\bx; \theta)$ converges to a nonstochastic function $\zeta(\theta)$ in probability uniformly in $\theta$ in $\mathcal{I}_2(\theta_0)$, and $\zeta(\theta)$ attains a strict local maximum at $\theta_0$ .
\end{itemize}
Let $\Theta_N$ be the set of roots of the equation  $\zeta_N'(\bx; \theta)=0$  corresponding to the local maxima. If $\Theta_N$ is empty, set $\Theta_N=\{0\}$.

Then, for any $\epsilon>0$, $$\lim_{N\to\infty}\PP\left[\inf_{\theta\in\Theta_N}(\theta-\theta_0)^\textup{T}(\theta-\theta_0)>\epsilon\right]=0.$$
\end{lemma}

\begin{lemma}[Theorem 5.23 in \cite{van2000asymptotic}]\label{lem:vandervaart}
For each $\theta$ in an open subset of Euclidean space, let $x \mapsto \kappa_\theta(x)$ be a measurable function such that $\theta \mapsto \kappa_\theta(x)$  is differentiable at $\theta_0$ for $\PP-$almost every $x$ with derivative $\kappa'_{\theta_0}(x)$ and such that, for every $\theta_1$ and $\theta_2$ in a neighborhood of $\theta_0$ and a measurable function $\kappa'$ with $\EE[ (\kappa')^2] < \infty$
$$|\kappa_{\theta_1}(x)-\kappa_{\theta_2}(x)|\leq \kappa'(x)||\theta_1-\theta_2||.$$
Furthermore, assume that the map $\theta\mapsto\EE m_\theta$ admits a second-order Taylor expansion at a point of maximum $\theta_0$ with nonsingular symmetric second derivative matrix $V_{\theta_0}$. If $N^{-1}\sum_{i=1}^N \kappa_{\widehat{\theta}_N}(X_i) \geq \sup_\theta N^{-1}\sum_{i=1}^N \kappa_\theta(X_i) -o_\PP(N^{-1})$ and $\widehat{\theta}_N \overset{\PP}{\to}\theta_0$, then
$$\sqrt{N}(\widehat{\theta}_N-\theta_0)=-V_{\theta_0}^{-1}\frac{1}{\sqrt{N}}\sum_{i=1}^N \kappa'_{\theta_0}(X_i)+o_\PP(1).$$
In particular, the sequence $\sqrt{N}(\widehat{\theta}_N-\theta_0)$ is asymptotically normal with mean zero and covariance matrix $V_{\theta_0}^{-1}\EE\left[\kappa'_{\theta_0}(\kappa'_{\theta_0})^\textup{T}\right]V_{\theta_0}^{-1}$.
\end{lemma}

\begin{lemma}\label{lem:lambda}
Let $\widehat{\blambda}_N=\text{argmax}_{\blambda} \sum_{i=1}^N \psi_i(\blambda)$ and ${\blambda}^*=\text{argmax}_{\blambda}\EE[\psi_i(\blambda)],$ where $\psi_i(\blambda)=-\frac{1}{4}(\boldsymbol{\lambda}^\textup{T}\mathbf{B}_{i})^2I(\boldsymbol{\lambda}^\textup{T}\mathbf{B}_{i}<0)-\boldsymbol{\lambda}^\textup{T}\mathbf{b}_i.$ Under Assumption~\ref{asreg}, we have
$$||\widehat{\blambda}_N-{\blambda}^*||_2=O_\PP\left({N}^{-1/2}\right),$$where $||\cdot||_2$ denotes the $l_2$ norm for a vector and spectral norm for a matrix. 
\end{lemma}
\begin{proof}[Proof of Lemma~\ref{lem:lambda}]
First, we consider the consistency of $\widehat{\blambda}_N$. 
We can verify that $\psi(\cdot)$ is continuous and differentiable, and its derivative $\psi_i'(\blambda)=-\frac{1}{2}\mathbf{B}_{i}\mathbf{B}_{i}^\textup{T}\blambda I(\boldsymbol{\lambda}^\textup{T}\mathbf{B}_{i}<0)-\mathbf{b}_i$ is also continuous. 
In addition, for any $\blambda$ in some open neighborhood $\Lambda$ of  ${\blambda}^*$, we have the derivative $\psi_i'(\blambda)$ is bounded and $\psi(\cdot)$ is local Lipschitz (we can simply choose the Lipschitz constant as the largest norm of the derivative in the neighborhood by the Mean Value Theorem, which is obtained at $\widetilde{\blambda}\in\Lambda$). By Example 19.7 in \cite{van2000asymptotic}, the parametric class of function $\{\psi(\blambda): \blambda\in\Lambda\}$ has a bounded bracketing number. By Lemma~\ref{lem:gc}, we know $$\sup_{\blambda\in\Lambda}\left|\frac{1}{N}\sum_{i=1}^N\psi_i(\blambda)-\EE[\psi_i(\blambda)]\right|\overset{\PP}{\rightarrow}0.$$
Then by applying Lemma~\ref{lem:amemiya}, we obtain $$\widehat{\blambda}_N\overset{\PP}{\rightarrow}{\blambda}^*.$$

Next, we study the asymptotic distribution of $\widehat{\blambda}_N$. Under Assumption~\ref{asreg}.3, we have $\left|\left|\EE\left[\psi_i'(\widetilde{\blambda})^2\right]\right|\right|_2<\infty$. Under Assumption~\ref{asreg}.2, we have the second derivative of $\EE[\psi_i(\blambda)]$ is a non-singular matrix when $\blambda={\blambda}^*$. Then we apply Lemma~\ref{lem:vandervaart} and obtain $$\sqrt{N}(\widehat{\blambda}_N-{\blambda}^*)\rightarrow \cN\left(0, \Sigma_{\blambda}\right),$$
where $$\Sigma_{\blambda}=\Sigma^{*-1}\EE\left[\psi_i'({\blambda}^*)\psi_i'({\blambda}^*)^\textup{T}\right]\Sigma^{*-1} \quad \text{and} \quad \Sigma^*=\EE\left[\bB_i\bB_i^\textup{T}I({\blambda}^{*\textup{T}}\bB_i<0)\right].$$
Therefore, we have $$||\widehat{\blambda}_N-{\blambda}^*||_2=O_\PP\left(N^{-1/2}\right).$$
\end{proof}

\begin{lemma}\label{lem:weightdiff}
Let $\blambda^*=\text{argmax}_{\blambda}\EE[\psi_i(\blambda)],$ where $\psi_i(\blambda)=-\frac{1}{4}(\boldsymbol{\lambda}^\textup{T}\mathbf{B}_{i})^2I(\boldsymbol{\lambda}^\textup{T}\mathbf{B}_{i}<0)-\boldsymbol{\lambda}^\textup{T}\mathbf{b}_i.$ 
Let $\gamma_{i}^*=\blambda^{*\textup{T}}\mathbf{B}_{i}I\left(\blambda^{*\textup{T}}\mathbf{B}_{i}\geq0\right)$ denote entries of $\bgamma^*$. Let $ w_i^*=(\gamma_{i}^*-\blambda^{*\textup{T}}\mathbf{B}_{i})/2$ denote the corresponding factorial weights. Under Assumption~\ref{asreg}, we have
$$\frac{1}{N}\sum_{i=1}^N (\widehat w_i - w_i^*)\epsilon_i=o_\PP\left(N^{-1/2}\right).$$
\end{lemma}
\begin{proof}[Proof of Lemma~\ref{lem:weightdiff}]
The function $\phi(x)=-xI(x<0)/2$ is continuous, differentiable, and Lipschitz. Specifically, for any $x,y$, we have $|\phi(x)-\phi(y)|\leq|x-y|/2$. In addition, the weights satisfy
$w_i^*=\left[\blambda^{*\textup{T}}\mathbf{B}_{i}I\left(\blambda^{*\textup{T}}\mathbf{B}_{i}\geq0\right)-\blambda^{*\textup{T}}\mathbf{B}_{i}\right]/2=\phi\left(\blambda^{*\textup{T}}\mathbf{B}_{i}\right)$ and similarly $\widehat w_i=\phi\left(\widehat{\blambda}^\textup{T}_N\mathbf{B}_{i}\right)$. Therefore, by the Lipschitz property of $\phi(\cdot)$ and the Cauchy-Schwarz inequality, we have 
$$|\widehat w_i - {w}_i^*|\leq |(\widehat{\blambda}_N-{\blambda}^*)^\textup{T}\mathbf{B}_{i}|/2\leq||\widehat{\blambda}_N-{\blambda}^*||_2||\mathbf{B}_{i}||_2/2.$$
Hence, we have
$$\frac{1}{N}\sum_{i=1}^N (\widehat w_i - {w}_i^*)\epsilon_i\leq\frac{||\widehat{\blambda}_N-{\blambda}^*||_2}{2N}\sum_{i=1}^N ||\mathbf{B}_{i}||_2\epsilon_i.$$

By Assumption~\ref{asreg}.3, we have
\begin{align*}
	&\EE\left[||\mathbf{B}_{i}||_2\epsilon_i\right]
	=\EE\left[\EE\left(||\mathbf{B}_{i}||_2\epsilon_i \ \middle|\  \bX,\bZ\right)\right]
	=\EE\left[||\mathbf{B}_{i}||_2\EE(\epsilon_i)\right]=0,\\
	&\text{Var}\left(||\mathbf{B}_{i}||_2\epsilon_i \right)
	=\EE\left[\EE\left(||\mathbf{B}_{i}||_2^2\epsilon_i^2 \  \middle|\   \bX,\bZ\right)\right]
	=\EE\left[||\mathbf{B}_{i}||_2^2\EE(\epsilon_i^2)\right] 
	=\sigma^2_i\EE\left[||\mathbf{B}_{i}||_2^2\right]<\infty,\\
	&\EE\left[\left|||\mathbf{B}_{i}||_2\epsilon_i\right|^3\right]
	=\EE\left[\EE\left(||\mathbf{B}_{i}||_2^3 |\epsilon_i|^3 \  \middle|\   \bX,\bZ\right)\right]
	=\EE\left[||\mathbf{B}_{i}||_2^3\EE(|\epsilon_i|^3)\right].
\end{align*}
Since $(\bX_i, \bZ_i)$, $i=1,\dots,N$ are identically distributed, the Lyapunov's condition $$\lim_{N\to\infty}\sum_{i=1}^N\EE\left[\left|||\mathbf{B}_{i}||_2\epsilon_i\right|^3\right]/\left(\sum_{i=1}^N \text{Var}\left(||\mathbf{B}_{i}||_2\epsilon_i \right) \right)^{3/2}=0$$ is satisfied by Assumption~\ref{asreg}.5.

Since $(\bX_i, \bZ_i, \epsilon_i)$, $i=1,\dots,N$, are independent, by Lyapunov central limit theorem and Slutsky's Theorem, we have $$\frac{1}{\sqrt{N}}\sum_{i=1}^N ||\mathbf{B}_{i}||_2\epsilon_i\rightarrow \cN\left(0,\EE\left[||\mathbf{B}_{i}||_2^2\right]\bar \sigma^2\right).$$
Therefore, we know $$\frac{1}{N}\sum_{i=1}^N ||\mathbf{B}_{i}||_2\epsilon_i=O_\PP\left(N^{-1/2}\right).$$

In addition, by Lemma~\ref{lem:lambda}, we have $||\widehat{\blambda}_N-{\blambda}^*||_2=O_\PP\left(N^{-1/2}\right)$ under Assumption~\ref{asreg}.

Putting everything together, we have $$\frac{1}{N}\sum_{i=1}^N (\widehat w_i - {w}_i^*)\epsilon_i\leq\frac{||\widehat{\blambda}_N-{\blambda}^*||_2}{2N}\sum_{i=1}^N ||\mathbf{B}_{i}||_2\epsilon_i=O_\PP\left(N^{-1}\right)=o_\PP\left(N^{-1/2}\right).$$
\end{proof}

\begin{lemma}\label{lem:error}
Let ${\blambda}^*=\text{argmax}_{\blambda}\EE[\psi_i(\blambda)],$ where $\psi_i(\blambda)=-\frac{1}{4}(\boldsymbol{\lambda}^\textup{T}\mathbf{B}_{i})^2I(\boldsymbol{\lambda}^\textup{T}\mathbf{B}_{i}<0)-\boldsymbol{\lambda}^\textup{T}\mathbf{b}_i.$ 
Let $\gamma_{i}^*=\blambda^{*\textup{T}}\mathbf{B}_{i}I\left(\blambda^{*\textup{T}}\mathbf{B}_{i}\geq0\right)$ denote entries of $\bgamma^*$. Let $w_i^*=(\gamma_{i}^*-{\blambda}^{*\textup{T}}\mathbf{B}_{i})/2$ denote the corresponding factorial weights. Under Assumptions~\ref{asreg}.3, \ref{asreg}.4 and \ref{asreg}.5, we have
\small
$$\sqrt{N}\left\{\frac{1}{N}\sum_{i=1}^N  w_i^*\epsilon_i+\frac{1}{2^{K-1}}\sum_{\mathbf{z}\in\mathcal{Z}}g_{{\cK}\mathbf{z}}\left(\frac{1}{N}\sum_{i=1}^N\mathbb{E}[Y_i(\mathbf{z})\mid\mathbf{X}_i]\right)-\frac{1}{2^{K-1}}\sum_{\mathbf{z}\in\mathcal{Z}}g_{{\cK}\mathbf{z}}\mathbb{E}[Y_i(\mathbf{z})]\right\} \rightarrow \cN\left(0,\sigma^2_{\epsilon }+\sigma^2_{\cK}\right),$$
\normalsize
where $$\sigma^2_{\epsilon }=\frac{\bar\sigma^2}{4}{\blambda}^{*\textup{T}}\EE\left[\bB_i\bB_i^\textup{T} I({\blambda}^{*\textup{T}}\bB_i<0)\right]{\blambda}^* \quad \text{ and } \quad \sigma^2_{\cK}={\rm Var}\left(\frac{1}{2^{K-1}}\sum_{\mathbf{z}\in\mathcal{Z}}g_{{\cK}\mathbf{z}}\mathbb{E}[Y_i(\mathbf{z})\mid\mathbf{X}_i]\right).$$
\end{lemma}
\begin{proof}[Proof of Lemma~\ref{lem:error}]
First,
\begin{align*}
	&\frac{1}{N}\sum_{i=1}^N  w_i^*\epsilon_i+\frac{1}{2^{K-1}}\sum_{\mathbf{z}\in\mathcal{Z}}g_{{\cK}\mathbf{z}}\left(\frac{1}{N}\sum_{i=1}^N\mathbb{E}[Y_i(\mathbf{z})\mid\mathbf{X}_i]\right)-\frac{1}{2^{K-1}}\sum_{\mathbf{z}\in\mathcal{Z}}g_{{\cK}\mathbf{z}}\mathbb{E}[Y_i(\mathbf{z})]\\
	=&\frac{1}{N}\sum_{i=1}^N (U_i+V_i) -\frac{1}{2^{K-1}}\sum_{\mathbf{z}\in\mathcal{Z}}g_{{\cK}\mathbf{z}}\mathbb{E}[Y_i(\mathbf{z})],
\end{align*}
where $U_i=w_i^*\epsilon_i$ and $V_i=\frac{1}{2^{K-1}}\sum_{\mathbf{z}\in\mathcal{Z}}g_{{\cK}\mathbf{z}}\mathbb{E}[Y_i(\mathbf{z})\mid\mathbf{X}_i]$.

By Assumption~\ref{asreg}.3,
\begin{align*}
	\EE\left(U_i \right)
	=\EE\left[\EE\left(w_i^*\epsilon_i\   \middle|\  \bX,\bZ\right)\right]
	=\EE\left[w_i^*\EE(\epsilon_i)\right]=0,
\end{align*}
and  
\begin{align*}
	\text{Var}\left(U_i \right)
	=\EE\left[\EE\left(w_i^{*2}\epsilon_i^2 \  \middle|\   \bX,\bZ\right)\right]
	=\EE\left[w_i^{*2}\EE(\epsilon_i^2)\right] 
	=\sigma^2_i\EE\left[w_i^{*2}\right] 
	=\frac{\sigma^2_i}{4}{\blambda}^{*\textup{T}}\EE\left[\bB_i\bB_i^\textup{T} I({\blambda}^{*\textup{T}}\bB_i<0)\right]{\blambda}^*
	=\sigma^2_{\epsilon i}.
\end{align*}

In addition,  by Assumption~\ref{asreg}.3, we can verify that
$$\mathbb{E}[V_i]=\frac{1}{2^{K-1}}\sum_{\mathbf{z}\in\mathcal{Z}}g_{{\cK}\mathbf{z}}\mathbb{E}[Y_i(\mathbf{z})] \ \text{and} \  {\rm Var}(V_i)={\rm Var}\left(\frac{1}{2^{K-1}}\sum_{\mathbf{z}\in\mathcal{Z}}g_{{\cK}\mathbf{z}}\mathbb{E}[Y_i(\mathbf{z})\mid\mathbf{X}_i]\right)={\sigma}^2_{\cK}<\infty.$$

Furthermore, we have
\begin{align*}
	\text{Cov}\left(U_i, V_i \right)
	&=\frac{1}{2^{K-1}}\EE\left[\EE\left(w_i^*\epsilon_i \sum_{\mathbf{z}\in\mathcal{Z}}g_{{\cK}\mathbf{z}}\mathbb{E}[Y_i(\mathbf{z})\mid\mathbf{X}_i]\  \middle|\   \bX,\bZ\right)\right]\\
	&=\frac{1}{2^{K-1}}\EE\left[w_i^* \sum_{\mathbf{z}\in\mathcal{Z}}g_{{\cK}\mathbf{z}}\mathbb{E}[Y_i(\mathbf{z})\mid\mathbf{X}_i]\EE(\epsilon_i)\right] \\
	&=0.
\end{align*}

Since $(\bX_i, \bZ_i, \epsilon_i)$, $i=1,\dots,N$, are independent, under Assumption~\ref{asreg}.5, Lyapunov central limit theorem and Slutsky's Theorem suggest 
$$\sqrt{N}\left\{\frac{1}{N}\sum_{i=1}^N (U_i+V_i)-\frac{1}{2^{K-1}}\sum_{\mathbf{z}\in\mathcal{Z}}g_{{\cK}\mathbf{z}}\mathbb{E}[Y_i(\mathbf{z})] \right\}\rightarrow \cN\left(0,\bar\sigma^2_{\epsilon }+\sigma^2_{\cK}\right).$$

\end{proof}

Next, we prove Theorem~1.
\begin{proof}[Proof of Theorem~1]
Without loss of generality, we consider the estimates for the factorial effect $\tau_{\cK}$ with contrast vector $\mathbf{g}_{\cK}$. That is, we focus on $$\widehat\tau_{\cK}-\tau_{\cK}=\left(\widehat\tau^+_{\cK}-\tau^+_{\cK}\right)-\left(\widehat\tau^-_{\cK}-\tau^-_{\cK}\right).$$
We first consider the estimation error for the positive part $\tau^+_{\cK}$. Specifically, we have
\begin{align*}
	\widehat\tau^+_{\cK}-\tau^+_{\cK}&=
	\underbrace{\left\{\frac{1}{N}\sum_{i=1}^N \widehat w_iA_{i{\cK}}^+Y_i^\textup{obs}-\frac{1}{2^{K-1}}\sum_{\mathbf{z}\in\mathcal{Z}}g_{{\cK}\mathbf{z}}^+\left(\frac{1}{N}\sum_{i=1}^N\mathbb{E}[Y_i(\mathbf{z})\mid\mathbf{X}_i]\right)\right\}}_{\text{part (\star)}}\\
	&\quad\quad+\frac{1}{2^{K-1}}\sum_{\mathbf{z}\in\mathcal{Z}}g_{{\cK}\mathbf{z}}^+\left(\frac{1}{N}\sum_{i=1}^N\mathbb{E}[Y_i(\mathbf{z})\mid\mathbf{X}_i]\right)-\frac{1}{2^{K-1}}\sum_{\mathbf{z}\in\mathcal{Z}}g_{{\cK}\mathbf{z}}^+\mathbb{E}[Y_i(\mathbf{z})].
\end{align*}
Here, part (\star) can be simplified due to our balance constraints, which gives
\begin{align*}
	&\frac{1}{N}\sum_{i=1}^N \widehat w_i A_{i{\cK}}^+Y_i^\textup{obs}-\frac{1}{2^{K-1}}\sum_{\mathbf{z}\in\mathcal{Z}}g_{{\cK}\mathbf{z}}^+\left(\frac{1}{N}\sum_{i=1}^N\mathbb{E}[Y_i(\mathbf{z})\mid\mathbf{X}_i]\right)\\
	=&\frac{1}{N}\sum_{i=1}^N \widehat w_i A_{i{\cK}}^+(\mathbb{E}[Y_i(\mathbf{Z}_i)\mid\mathbf{X}_i]+\epsilon_i)-\frac{1}{2^{K-1}}\frac{1}{N}\sum_{i=1}^N\sum_{\mathbf{z}\in\mathcal{Z}}g_{{\cK}\mathbf{z}}^+\mathbb{E}[Y_i(\mathbf{z})\mid\mathbf{X}_i]\\
	=&\frac{1}{N}\underbrace{\left\{\sum_{i=1}^N \widehat w_i A_{i{\cK}}^+\mathbb{E}[Y_i(\mathbf{Z}_i)\mid\mathbf{X}_i]-\frac{1}{2^{K-1}}\sum_{i=1}^N\sum_{\mathbf{z}\in\mathcal{Z}}g_{{\cK}\mathbf{z}}^+\mathbb{E}[Y_i(\mathbf{z})\mid\mathbf{X}_i]\right\}}_{=0\text{ by construction}}+\frac{1}{N}\sum_{i=1}^N \widehat w_iA_{i{\cK}}^+\epsilon_i\\
	=&\frac{1}{N}\sum_{i=1}^N \widehat w_iA_{i{\cK}}^+\epsilon_i.
\end{align*}
Therefore, we have 
$$\widehat\tau^+_{\cK}-\tau^+_{\cK}=\frac{1}{N}\sum_{i=1}^N \widehat w_iA_{i{\cK}}^+\epsilon_i+\frac{1}{2^{K-1}}\sum_{\mathbf{z}\in\mathcal{Z}}g_{{\cK}\mathbf{z}}^+\left(\frac{1}{N}\sum_{i=1}^N\mathbb{E}[Y_i(\mathbf{z})\mid\mathbf{X}_i]\right)-\frac{1}{2^{K-1}}\sum_{\mathbf{z}\in\mathcal{Z}}g_{{\cK}\mathbf{z}}^+\mathbb{E}[Y_i(\mathbf{z})].$$
Similarly, for the estimation error for the negative part $\tau^-_{\cK}$. have 
$$\widehat\tau^-_{\cK}-\tau^-_{\cK}=\frac{1}{N}\sum_{i=1}^N \widehat w_iA_{i{\cK}}^-\epsilon_i+\frac{1}{2^{K-1}}\sum_{\mathbf{z}\in\mathcal{Z}}g_{{\cK}\mathbf{z}}^-\left(\frac{1}{N}\sum_{i=1}^N\mathbb{E}[Y_i(\mathbf{z})\mid\mathbf{X}_i]\right)-\frac{1}{2^{K-1}}\sum_{\mathbf{z}\in\mathcal{Z}}g_{{\cK}\mathbf{z}}^-\mathbb{E}[Y_i(\mathbf{z})].$$
Putting these two parts back together, we obtain
$$\widehat\tau_{\cK}-\tau_{\cK}=\underbrace{\frac{1}{N}\sum_{i=1}^N \widehat w_i\epsilon_i}_{\text{part I}}+\underbrace{\frac{1}{2^{K-1}}\sum_{\mathbf{z}\in\mathcal{Z}}g_{{\cK}\mathbf{z}}\left(\frac{1}{N}\sum_{i=1}^N\mathbb{E}[Y_i(\mathbf{z})\mid\mathbf{X}_i]\right)-\frac{1}{2^{K-1}}\sum_{\mathbf{z}\in\mathcal{Z}}g_{{\cK}\mathbf{z}}\mathbb{E}[Y_i(\mathbf{z})]}_{\text{part II}}.$$
With $m(w_i)=w_i^2$, the optimal factorial weights $\widehat{\bw}$ can be written as
$$\widehat w_i=(\widehat\gamma_{iN}-\widehat{\blambda}^\textup{T}_N\mathbf{B}_{i})/2,$$  
where $$\widehat\gamma_{iN}=\widehat{\blambda}^\textup{T}_N\mathbf{B}_{i}I\left(\widehat{\blambda}^\textup{T}_N\mathbf{B}_{i}\geq0\right)$$
and $\widehat{\blambda}_N$ is the $M$-estimator such that 
$$\widehat{\blambda}_N=\text{argmax}_{\blambda} \sum_{i=1}^N \psi_i(\blambda), \text{ with }\psi_i(\blambda)=-\frac{1}{4}(\boldsymbol{\lambda}^\textup{T}\mathbf{B}_{i})^2I(\boldsymbol{\lambda}^\textup{T}\mathbf{B}_{i}<0)-\boldsymbol{\lambda}^\textup{T}\mathbf{b}_i.$$
In addition, with $\blambda^*=\text{argmax}_{\blambda}\EE[\psi_i(\blambda)],$ we define the corresponding $\bgamma^*$ as $$\gamma_{i}^*=\blambda^{*\textup{T}}\mathbf{B}_{i}I\left(\blambda^{*\textup{T}}\mathbf{B}_{i}\geq0\right),$$ and the corresponding factorial weights $\bw^*$ as $$w_i^*=(\gamma_{i}^*-\blambda^{*\textup{T}}\mathbf{B}_{i})/2.$$  

Then we can further decompose part I into two parts
\begin{align*}
	\frac{1}{N}\sum_{i=1}^N \widehat w_i\epsilon_i=\underbrace{\frac{1}{N}\sum_{i=1}^N (\widehat w_i - w_i^*)\epsilon_i}_{\text{part IA}}+\underbrace{\frac{1}{N}\sum_{i=1}^N  w_i^*\epsilon_i}_{\text{part IB}}.
\end{align*}
By Lemma~\ref{lem:weightdiff}, part IA is $o_\PP\left(N^{-1/2}\right)$. 

Next, we consider part IB and part II by looking at their sum directly. By Lemma~\ref{lem:error}, we have 
\small
$$\sqrt{N}\left\{\frac{1}{N}\sum_{i=1}^N  w_i^*\epsilon_i+\frac{1}{2^{K-1}}\sum_{\mathbf{z}\in\mathcal{Z}}g_{{\cK}\mathbf{z}}\left(\frac{1}{N}\sum_{i=1}^N\mathbb{E}[Y_i(\mathbf{z})\mid\mathbf{X}_i]\right)-\frac{1}{2^{K-1}}\sum_{\mathbf{z}\in\mathcal{Z}}g_{{\cK}\mathbf{z}}\mathbb{E}[Y_i(\mathbf{z})]\right\} \rightarrow \cN\left(0,\sigma^2_{\epsilon}+\sigma^2_{\cK}\right),$$
\normalsize
where $$\sigma^2_{\epsilon }=\frac{\bar\sigma^2_i}{4}\blambda^{*\textup{T}}\EE\left[\bB_i\bB_i^\textup{T} I(\blambda^{*\textup{T}}\bB_i<0)\right]\blambda^* \quad \text{ and } \quad \sigma^2_{{\cK}}=\frac{1}{2^K}{\rm Var}\left(\sum_{\mathbf{z}\in\mathcal{Z}}g_{{\cK}\mathbf{z}}\mathbb{E}[Y_i(\mathbf{z})\mid\mathbf{X}_i]\right).$$

Putting everything together, we can conclude that $$\sqrt{N}(\widehat\tau_{\cK}-\tau_{\cK})\rightarrow \cN\left(0,\sigma^2_{\epsilon}+\sigma^2_{\cK}\right),$$

\end{proof}
In the case of homoskedasticity, where the conditional variance of $Y_i$ given $\bX_i$, $\sigma_i^2$, equals $\sigma^2$, we can omit Assumptions~\ref{asreg}.4 (automatically satisfied) and \ref{asreg}.5 and still achieve asymptotic normality.

\subsection{Proof of Theorem~2}\label{app:thm2}
\begin{proof}[Proof of Theorem~2]
To estimate the asymptotic variance of $\widehat{\tau}_{\cK}$, we consider a combined parameter $\btheta=(\blambda^\textup{T},t)^\textup{T}$. Suppose problem~(5) has $P$ balance constraints, then $\btheta$ is a vector of length $P+1$. In addition, we define $\eta_i(\btheta):\mathbb{R}^{P+1} \rightarrow \mathbb{R}^{P+1}$ as $$\eta_i(\btheta)=\begin{pmatrix}
	&\eta_{i,1}(\btheta)&\\
	&\vdots&\\
	&\eta_{i,P+1}(\btheta)&
\end{pmatrix}=
\begin{pmatrix}
	&\psi_i'(\blambda)&\\
	&w_i(A_{i{\cK}}^+-A_{i{\cK}}^-)Y_i^\textup{obs}-t&
\end{pmatrix},
$$
where $\psi_i'(\blambda)=-\frac{1}{2}\mathbf{B}_{i}\mathbf{B}_{i}^\textup{T}\blambda I(\boldsymbol{\lambda}^\textup{T}\mathbf{B}_{i}<0)-\mathbf{b}_i$ is the first derivative of $\psi_i(\blambda)=-\frac{1}{4}(\boldsymbol{\lambda}^\textup{T}\mathbf{B}_{i})^2I(\boldsymbol{\lambda}^\textup{T}\mathbf{B}_{i}<0)-\boldsymbol{\lambda}^\textup{T}\mathbf{b}_i,$
and the weights $w_i=(\gamma_i-\lambda^\textup{T}\bB_i)/2=-\blambda^\textup{T}\bB_i I(\blambda^\textup{T}\bB_i<0)/2$ is also a function of $\blambda$. 

Recall that $$\widehat{\blambda}={\rm argmax}_{\blambda} \frac{1}{N}\sum_{i=1}^N\psi_i(\blambda) \quad{\rm and } \quad \widehat{\tau}_{\cK}=\frac{1}{N}\sum_{i=1}^Nw_i A_{i{\cK}}^+Y_i^\textup{obs}-\frac{1}{N}\sum_{i=1}^Nw_i A_{i{\cK}}^-Y_i^\textup{obs}.$$
Then we have $\widehat{\btheta}=(\widehat{\blambda}^\textup{T},\widehat\tau_{\cK})^\textup{T}$ is the $Z$-estimator that satisfies the following estimating equation $$\frac{1}{N}\sum_{i=1}^N \eta_i(\widehat{\btheta})=0.$$
In addition, we have ${\btheta}^*=({\blambda}^{*\textup{T}},\tau_{\cK})^\textup{T}$ satisfies $$ \EE\left[\eta_i(\btheta^*)\right]=0,$$
where $${\blambda}^*={\rm argmax}_{\blambda} \EE\left[\psi_i(\blambda)\right], \quad w_i^*=-\blambda^{*\textup{T}}\bB_i I(\blambda^{*\textup{T}}\bB_i<0)/2,$$
and  $${\tau}_{\cK}=\EE\left[w_i^* (A_{i{\cK}}^+-A_{i{\cK}}^-)Y_i^\textup{obs}\right] $$ holds by the construction of balance constraints.
Next, for $p=1,\dots,P+1$, applying the Mean Value Theorem to the $Z$-estimating equation for $\eta_{i,p}(\widehat{\btheta})$, we have 
$$0=\frac{1}{N}\sum_{i=1}^N \eta_{i,p}(\widehat{\btheta})=\frac{1}{N}\sum_{i=1}^N \eta_{i,p}({\btheta}^*)+\frac{1}{N}\sum_{i=1}^N \eta_{i,p}'(\widetilde{\btheta}_p)\left(\widehat{\btheta}-\btheta^*\right),$$where $\widetilde{\btheta}_p$ is between $\widehat{\btheta}$ and $\btheta^*$, and $ \eta_{i,p}'(\widetilde{\btheta}_p)$ is the first derivative of $\eta_{i,p}(\btheta)$ evaluated at $\btheta=\widetilde{\btheta}_p$.

Let $\bH_p=\frac{1}{N}\sum_{i=1}^N \eta_{i,p}'(\widetilde{\btheta}_p)$ and $\bH=
(\bH_1^\textup{T}, \dots, \bH_{P+1}^\textup{T})
^\textup{T}$. Let $\bH_p^*=\EE\left[\eta_{i,p}'({\btheta}^*)\right]$ and $\bH^*=
(\bH_1^{*\textup{T}}, \dots, \bH_{P+1}^{*\textup{T}})
^\textup{T}$. We have
$$\bH\left(\widehat{\btheta}-\btheta^*\right)=-\frac{1}{N}\sum_{i=1}^N \eta_i({\btheta}^*).$$
Equivalently, we have $$\widehat{\btheta}-\btheta^*=-\bH^{-1}\frac{1}{N}\sum_{i=1}^N \eta_i({\btheta}^*).$$
By Lemma~\ref{lem:lambda} and Theorem~1, we have $$||\widetilde{\btheta}_p-\btheta^*||_2\leq ||\widehat{\btheta}-\btheta^*||_2=O_\PP\left(N^{{-1/2}}\right), \textrm{ for all } p=1,\dots,P+1.$$
For $q=1,\dots,P+1$, let $\bH_{p,q}$ denote the $q$th entry of $\bH_{p}$ and $\eta_{i,p,q}'$ denote $q$th entry of $\eta_{i,p}'$. Then, for $p=1,\dots,P+1$, by the Mean Value Theorem again for each entry of $\bH_{p}$, we have 
$$\bH_{p,q }=\frac{1}{N}\sum_{i=1}^N \eta_{i,p,q}'(\widetilde{\btheta}_p)=\frac{1}{N}\sum_{i=1}^N \eta_{i,p,q}'({\btheta}^*)+\frac{1}{N}\sum_{i=1}^N \eta_{i,p,q}''(\mathring{\btheta}_{pq})\left(\widetilde{\btheta}_p-\btheta^*\right)=\frac{1}{N}\sum_{i=1}^N \eta_{i,p,q}'({\btheta}^*)+O_\PP\left(N^{-1/2}\right),$$where $\mathring{\btheta}_{pq}$ is between $\widetilde{\btheta}$ and $\btheta^*$, and $ \eta_{i,p,q}''(\mathring{\btheta}_{pq})$ is the derivative of $\eta_{i,p,q}'(\btheta)$ evaluated at $\btheta=\mathring{\btheta}_{pq}$.

By the Law of Large Numbers and Continuous Mapping Theorem, we have $$\bH^{-1}\overset{\PP}{\rightarrow}(\bH^*)^{-1}.$$

By the Central Limit Theorem, we have $$\frac{1}{\sqrt{N}}\sum_{i=1}^N \eta_i({\btheta}^*)\rightarrow \cN\left(0,{\rm Var}(\eta_i({\btheta}^*))\right)=\cN\left(0,\EE\left[\eta_i({\btheta}^*)\eta_i({\btheta}^{*})^\textup{T}\right]\right).$$

By Slutsky's Theorem, we have 
$$\sqrt{N}\left(\widehat{\btheta}-\btheta^*\right)\rightarrow \cN\left(0,\bH^{*-1}\EE\left[\eta_i({\btheta}^*)\eta_i({\btheta}^{*})^\textup{T}\right]\bH^{*-1}\right).$$

Finally, we are going to show that $\left[\frac{1}{N}\sum_{i=1}^N\eta_i'(\widehat{\btheta})\right]^{-1}\left[\frac{1}{N}\sum_{i=1}^N\eta_i(\widehat{\btheta})\eta_i(\widehat{\btheta})^\textup{T}\right]\left[\frac{1}{N}\sum_{i=1}^N\eta_i'(\widehat{\btheta})\right]^{-1}$ is a consistent estimator for the asymptotic variance.

By the Mean Value Theorem again, we have for $p,q=1,\dots,P+1$,
$$\frac{1}{N}\sum_{i=1}^N \eta_{i,p,q}'(\widehat{\btheta})=\frac{1}{N}\sum_{i=1}^N \eta_{i,p,q}'({\btheta}^*)+\frac{1}{N}\sum_{i=1}^N \eta_{i,p,q}''(\breve{\btheta}_{pq})\left(\widehat{\btheta}-\btheta^*\right)=\frac{1}{N}\sum_{i=1}^N \eta_{i,p,q}'({\btheta}^*)+O_\PP\left(N^{-1/2}\right),$$where $\breve{\btheta}_{pq}$ is between $\widehat{\btheta}$ and $\btheta^*$. By the Law of Large Numbers and Continuous Mapping Theorem, we have $$\left[\frac{1}{N}\sum_{i=1}^N \eta_i'(\widehat{\btheta})\right]^{-1}\overset{\PP}{\rightarrow}(\bH^*)^{-1}.$$
In addition, we can verify that each entry of the matrix $\eta_i(\btheta)\eta_i(\btheta)^\textup{T}$ is locally Lipschitz around $\btheta^*$. Let $\mathbf{L}$ denote the matrix of corresponding Lipschitz constants. We have
$$\frac{1}{N}\sum_{i=1}^N\eta_i(\widehat{\btheta})\eta_i(\widehat{\btheta})^\textup{T}-\frac{1}{N}\sum_{i=1}^N\eta_i({\btheta}^*)\eta_i({\btheta}^*)^\textup{T}\leq \mathbf{L}||\widehat{\btheta}-\btheta^*||_2=O_\PP\left(N^{-1/2}\right).$$
By the Law of Large Numbers, we have $$\frac{1}{N}\sum_{i=1}^N\eta_i(\widehat{\btheta})\eta_i(\widehat{\btheta})^\textup{T}\overset{\PP}{\rightarrow}\EE\left[\eta_i({\btheta}^*)\eta_i({\btheta}^{*})^\textup{T}\right].$$
Putting everything together and applying Continuous Mapping Theorem, we have 
$$\left[\frac{1}{N}\sum_{i=1}^N\eta_i'(\widehat{\btheta})\right]^{-1}\left[\frac{1}{N}\sum_{i=1}^N\eta_i(\widehat{\btheta})\eta_i(\widehat{\btheta})^\textup{T}\right]\left[\frac{1}{N}\sum_{i=1}^N\eta_i'(\widehat{\btheta})\right]^{-1}\overset{\PP}{\rightarrow}\bH^{*-1}\EE\left[\eta_i({\btheta}^*)\eta_i({\btheta}^{*})^\textup{T}\right]\bH^{*-1}.$$
Since we are only interested in the asymptotic variance of $\sqrt{N}(\widehat{\tau}_{\cK}-\tau_{\cK})$, we only need to estimate the last estimate of the $\bH^{*-1}\EE\left[\eta_i({\btheta}^*)\eta_i({\btheta}^{*})^\textup{T}\right]\bH^{*-1}$, which can be consistently estimated by the last element of $\left[\frac{1}{N}\sum_{i=1}^N\eta_i'(\widehat{\btheta})\right]^{-1}\left[\frac{1}{N}\sum_{i=1}^N\eta_i(\widehat{\btheta})\eta_i(\widehat{\btheta})^\textup{T}\right]\left[\frac{1}{N}\sum_{i=1}^N\eta_i'(\widehat{\btheta})\right]^{-1}$. Specifically, we only need to consider the last row of $\left[\frac{1}{N}\sum_{i=1}^N\eta_i'(\widehat{\btheta})\right]^{-1}$, which is $$L^\textup{T}=\left[\left(\frac{1}{N}\sum_{i=1}^N-\frac{1}{2}\bB_i^\textup{T}(A_{i{\cK}}^+-A_{i{\cK}}^-)Y_i^\textup{obs}I\left(\widehat{\blambda}^\textup{T}\bB_i<0\right)\right)\left(\frac{1}{N}\sum_{i=1}^N-\frac{1}{2}\bB_i\bB_i^\textup{T} I\left(\widehat{\blambda}^\textup{T}\bB_i<0\right)\right)^{-1}, -1\right].$$ Therefore, we have $$L^\textup{T}\left[\frac{1}{N}\sum_{i=1}^N\eta_i(\widehat{\btheta})\eta_i(\widehat{\btheta})^\textup{T}\right]L=\frac{1}{N}\sum_{i=1}^N\left(\eta_i(\widehat{\theta})^\textup{T}L\right)^2$$ as the consistent estimator of the asymptotic variance of $\sqrt{N}(\widehat{\tau}_{\cK}-\tau_{\cK})$.
\end{proof}

\subsection{Proof of Proposition~1}
\begin{proof}[Proof of  Proposition~1]
Let  $(\mathbf{G}_{\textup{u}-}^\textup{T})^{\dagger}=\mathbf{G}_{\textup{u}-}(\mathbf{G}_{\textup{u}-}^\textup{T}\mathbf{G}_{\textup{u}-})^{-1}$. 
We have
\begin{align*}
	\boldsymbol{\tau}_+&=\frac{1}{2^{K-1}}\begin{bmatrix}
		\mathbf{G}_{\textup{o}+} \\
		\mathbf{G}_{\textup{u}+}
	\end{bmatrix}^\textup{T} 
	\begin{bmatrix}
		(\mathbb{E}[Y(\mathbf{z})])_{\mathbf{z}\in\mathcal{Z}_\textup{o}}\\
		(\mathbb{E}[Y(\mathbf{z})])_{\mathbf{z}\in\mathcal{Z}_\textup{u}}
	\end{bmatrix}\\
	&=\frac{1}{2^{K-1}}	\mathbf{G}_{\textup{o}+}^\textup{T}	(\mathbb{E}[Y(\mathbf{z})])_{\mathbf{z}\in\mathcal{Z}_\textup{o}}+	\mathbf{G}_{\textup{u}+}^\textup{T}	\left\{-(\mathbf{G}_{\textup{u}-}^\textup{T})^{\dagger}\mathbf{G}_{\textup{o}-}^\textup{T}(\mathbb{E}[Y(\mathbf{z})])_{\mathbf{z}\in\mathcal{Z}_\textup{o}}\right\}\\
	&=\frac{1}{2^{K-1}}	\left\{\mathbf{G}_{\textup{o}+}^\textup{T}	-\mathbf{G}_{\textup{u}+}^\textup{T}	(\mathbf{G}_{\textup{u}-}^\textup{T})^{\dagger}\mathbf{G}_{\textup{o}-}^\textup{T}\right\}(\mathbb{E}[Y(\mathbf{z})])_{\mathbf{z}\in\mathcal{Z}_\textup{o}}\\
\end{align*}
\end{proof}
	\bibliographystyle{apalike}
	
	\bibliography{bib}

\begin{thebibliography}{}

\bibitem[Amemiya, 1985]{amemiya1985advanced}
Amemiya, T. (1985).
\newblock {\em Advanced econometrics}.
\newblock Harvard university press.

\bibitem[Arif and Shah, 2007]{arif2007association}
Arif, A.~A. and Shah, S.~M. (2007).
\newblock Association between personal exposure to volatile organic compounds
  and asthma among us adult population.
\newblock {\em International Archives of Occupational and Environmental
  Health}, 80:711--719.

\bibitem[Athey et~al., 2018]{athey2018approximate}
Athey, S., Imbens, G.~W., and Wager, S. (2018).
\newblock Approximate residual balancing: debiased inference of average
  treatment effects in high dimensions.
\newblock {\em Journal of the Royal Statistical Society: Series B},
  80(4):597--623.

\bibitem[Batterman et~al., 2014]{batterman2014personal}
Batterman, S., Su, F.-C., Li, S., Mukherjee, B., and Jia, C. (2014).
\newblock Personal exposure to mixtures of volatile organic compounds: modeling
  and further analysis of the riopa data.
\newblock {\em Research Report (Health Effects Institute)}, Jun(181):3--63.

\bibitem[Ben-Michael et~al., 2024]{ben2023multilevel}
Ben-Michael, E., Feller, A., and Hartman, E. (2024).
\newblock Multilevel calibration weighting for survey data.
\newblock {\em Political Analysis}, 32:65--83.

\bibitem[Berenson et~al., 1998]{berenson1998association}
Berenson, G.~S., Srinivasan, S.~R., Bao, W., Newman, W.~P., Tracy, R.~E., and
  Wattigney, W.~A. (1998).
\newblock Association between multiple cardiovascular risk factors and
  atherosclerosis in children and young adults.
\newblock {\em New England Journal of Medicine}, 338(23):1650--1656.

\bibitem[Berry et~al., 2015]{berry2015platform}
Berry, S.~M., Connor, J.~T., and Lewis, R.~J. (2015).
\newblock The platform trial: an efficient strategy for evaluating multiple
  treatments.
\newblock {\em Journal of the American Medical Association},
  313(16):1619--1620.

\bibitem[Box et~al., 1978]{box1978statistics}
Box, G.~E., Hunter, W.~H., and Hunter, S. (1978).
\newblock {\em Statistics for experimenters}.
\newblock New York: John Wiley and Sons.

\bibitem[Box and Meyer, 1986]{box1986analysis}
Box, G.~E. and Meyer, R.~D. (1986).
\newblock An analysis for unreplicated fractional factorials.
\newblock {\em Technometrics}, 28(1):11--18.

\bibitem[Branson et~al., 2016]{branson2016improving}
Branson, Z., Dasgupta, T., and Rubin, D.~B. (2016).
\newblock Improving covariate balance in {$2^K$} factorial designs via
  rerandomization with an application to a {New York City Department of
  Education High School Study}.
\newblock {\em The Annals of Applied Statistics}, 10(4):1958–1976.

\bibitem[Bruns-Smith et~al., 2025]{bruns2025augmented}
Bruns-Smith, D., Dukes, O., Feller, A., and Ogburn, E.~L. (2025).
\newblock Augmented balancing weights as linear regression.
\newblock {\em Journal of the Royal Statistical Society Series B: Statistical
  Methodology}, page qkaf019.

\bibitem[Bruns-Smith and Feller, 2022]{bruns2022outcome}
Bruns-Smith, D. and Feller, A. (2022).
\newblock Outcome assumptions and duality theory for balancing weights.
\newblock In {\em International Conference on Artificial Intelligence and
  Statistics}, pages 11037--11055. PMLR.

\bibitem[Byar et~al., 1993]{byar1993incomplete}
Byar, D.~P., Herzberg, A.~M., and Tan, W.-Y. (1993).
\newblock Incomplete factorial designs for randomized clinical trials.
\newblock {\em Statistics in Medicine}, 12(17):1629--1641.

\bibitem[Chan et~al., 2016]{chan2016globally}
Chan, K. C.~G., Yam, S. C.~P., and Zhang, Z. (2016).
\newblock Globally efficient non-parametric inference of average treatment
  effects by empirical balancing calibration weighting.
\newblock {\em Journal of the Royal Statistical Society: Series B},
  78(3):673--700.

\bibitem[Chattopadhyay and Zubizarreta, 2022]{chattopadhyay2022implied}
Chattopadhyay, A. and Zubizarreta, J.~R. (2022).
\newblock On the implied weights of linear regression for causal inference.
\newblock {\em arXiv preprint arXiv:2104.06581v3}.

\bibitem[Cohn et~al., 2023]{cohn2023balancing}
Cohn, E.~R., Ben-Michael, E., Feller, A., and Zubizarreta, J.~R. (2023).
\newblock Balancing weights for causal inference.
\newblock In {\em Handbook of Matching and Weighting Adjustments for Causal
  Inference}, pages 293--312. Chapman and Hall/CRC.

\bibitem[Dasgupta et~al., 2015]{dasgupta2015causal}
Dasgupta, T., Pillai, N.~S., and Rubin, D.~B. (2015).
\newblock Causal inference from {$2^K$} factorial designs by using potential
  outcomes.
\newblock {\em Journal of the Royal Statistical Society: Series B},
  77(4):727--753.

\bibitem[Dong, 2015]{dong2015using}
Dong, N. (2015).
\newblock Using propensity score methods to approximate factorial experimental
  designs to analyze the relationship between two variables and an outcome.
\newblock {\em American Journal of Evaluation}, 36(1):42--66.

\bibitem[Egami and Imai, 2019]{egami2018causal}
Egami, N. and Imai, K. (2019).
\newblock Causal interaction in factorial experiments: Application to conjoint
  analysis.
\newblock {\em Journal of the American Statistical Association},
  114(536):529--540:.

\bibitem[Fan et~al., 2023]{fan2023optimal}
Fan, J., Imai, K., Liu, H., Ning, Y., and Yang, X. (2023).
\newblock Optimal covariate balancing conditions in propensity score
  estimation.
\newblock {\em Journal of Business \& Economic Statistics}, 41(1):97--110.

\bibitem[Fong and Imai, 2014]{fong2014covariate}
Fong, C. and Imai, K. (2014).
\newblock Covariate balancing propensity score for general treatment regimes.
\newblock {\em Princeton Manuscript}, pages 1--31.

\bibitem[Hahn, 1998]{hahn1998role}
Hahn, J. (1998).
\newblock On the role of the propensity score in efficient semiparametric
  estimation of average treatment effects.
\newblock {\em Econometrica}, 66(2):315--331.

\bibitem[Hainmueller, 2012]{hainmueller2012entropy}
Hainmueller, J. (2012).
\newblock Entropy balancing for causal effects: A multivariate reweighting
  method to produce balanced samples in observational studies.
\newblock {\em Political Analysis}, 20(1):25--46.

\bibitem[Han and Rubin, 2021]{han2021contrast}
Han, S. and Rubin, D.~B. (2021).
\newblock Contrast-specific propensity scores.
\newblock {\em Biostatistics \& Epidemiology}, 5(1):1--8.

\bibitem[Hirano and Imbens, 2001]{hirano2001estimation}
Hirano, K. and Imbens, G.~W. (2001).
\newblock Estimation of causal effects using propensity score weighting: An
  application to data on right heart catheterization.
\newblock {\em Health Services and Outcomes Research Methodology},
  2(3-4):259--278.

\bibitem[Hirano et~al., 2003]{hirano2003efficient}
Hirano, K., Imbens, G.~W., and Ridder, G. (2003).
\newblock Efficient estimation of average treatment effects using the estimated
  propensity score.
\newblock {\em Econometrica}, 71(4):1161--1189.

\bibitem[Imai and Van~Dyk, 2004]{imai2004causal}
Imai, K. and Van~Dyk, D.~A. (2004).
\newblock Causal inference with general treatment regimes: Generalizing the
  propensity score.
\newblock {\em Journal of the American Statistical Association},
  99(467):854--866.

\bibitem[Imbens, 2000]{imbens2000role}
Imbens, G.~W. (2000).
\newblock The role of the propensity score in estimating dose-response
  functions.
\newblock {\em Biometrika}, 87(3):706--710.

\bibitem[Johnson, 1999]{johnson1999review}
Johnson, B.~L. (1999).
\newblock A review of the effects of hazardous waste on reproductive health.
\newblock {\em American Journal of Obstetrics and Gynecology}, 181(1):S12--S16.

\bibitem[Kang and Schafer, 2007]{kang2007demystifying}
Kang, J.~D. and Schafer, J.~L. (2007).
\newblock Demystifying double robustness: A comparison of alternative
  strategies for estimating a population mean from incomplete data.
\newblock {\em Statistical Science}, 22(4):523--539.

\bibitem[Kuhn et~al., 2019]{kuhn2019simulation}
Kuhn, J., Sheldrick, R.~C., Broder-Fingert, S., Chu, A., Fortuna, L., Jordan,
  M., Rubin, D., and Feinberg, E. (2019).
\newblock Simulation and minimization: technical advances for factorial
  experiments designed to optimize clinical interventions.
\newblock {\em BMC Medical Research Methodology}, 19:1--9.

\bibitem[Li and Li, 2019]{li2019propensity}
Li, F. and Li, F. (2019).
\newblock Propensity score weighting for causal inference with multiple
  treatments.
\newblock {\em The Annals of Applied Statistics}, 13(4):2389--2415.

\bibitem[Li et~al., 2018]{li2018balancing}
Li, F., Morgan, K.~L., and Zaslavsky, A.~M. (2018).
\newblock Balancing covariates via propensity score weighting.
\newblock {\em Journal of the American Statistical Association},
  113(521):390--400.

\bibitem[Li et~al., 2020]{li2020rerandomization}
Li, X., Ding, P., and Rubin, D.~B. (2020).
\newblock Rerandomization in {$2^K$} factorial experiments.
\newblock {\em The Annals of Statistics}, 48(1):43--63.

\bibitem[Lopez and Gutman, 2017]{lopez2017estimation}
Lopez, M.~J. and Gutman, R. (2017).
\newblock Estimation of causal effects with multiple treatments: a review and
  new ideas.
\newblock {\em Statistical Science}, 32(3):432--454.

\bibitem[Lu, 2016]{lu2016covariate}
Lu, J. (2016).
\newblock Covariate adjustment in randomization-based causal inference for
  {$2^K$} factorial designs.
\newblock {\em Statistics \& Probability Letters}, 119:11--20.

\bibitem[McCaffrey et~al., 2013]{mccaffrey2013tutorial}
McCaffrey, D.~F., Griffin, B.~A., Almirall, D., Slaughter, M.~E., Ramchand, R.,
  and Burgette, L.~F. (2013).
\newblock A tutorial on propensity score estimation for multiple treatments
  using generalized boosted models.
\newblock {\em Statistics in Medicine}, 32(19):3388--3414.

\bibitem[Montero-Montoya et~al., 2018]{montero2018volatile}
Montero-Montoya, R., L{\'o}pez-Vargas, R., and Arellano-Aguilar, O. (2018).
\newblock Volatile organic compounds in air: sources, distribution, exposure
  and associated illnesses in children.
\newblock {\em Annals of Global Health}, 84(2):225.

\bibitem[Mukerjee et~al., 2018]{mukerjee2018using}
Mukerjee, R., Dasgupta, T., and Rubin, D.~B. (2018).
\newblock Using standard tools from finite population sampling to improve
  causal inference for complex experiments.
\newblock {\em Journal of the American Statistical Association},
  113(522):868--881.

\bibitem[Park et~al., 2014]{park2014environmental}
Park, S.~K., Tao, Y., Meeker, J.~D., Harlow, S.~D., and Mukherjee, B. (2014).
\newblock Environmental risk score as a new tool to examine multi-pollutants in
  epidemiologic research: an example from the nhanes study using serum lipid
  levels.
\newblock {\em PloS one}, 9(6):e98632.

\bibitem[Pashley and Bind, 2023]{pashley2022causal}
Pashley, N.~E. and Bind, M.-A.~C. (2023).
\newblock Causal inference for multiple treatments using fractional factorial
  designs.
\newblock {\em Canadian Journal of Statistics}, 51(2):444--468.

\bibitem[Patel et~al., 2013]{patel2013systematic}
Patel, C.~J., Rehkopf, D.~H., Leppert, J.~T., Bortz, W.~M., Cullen, M.~R.,
  Chertow, G.~M., and Ioannidis, J.~P. (2013).
\newblock Systematic evaluation of environmental and behavioural factors
  associated with all-cause mortality in the united states national health and
  nutrition examination survey.
\newblock {\em International Journal of Epidemiology}, 42(6):1795--1810.

\bibitem[Resa and Zubizarreta, 2020]{de2020direct}
Resa, M. d. l.~A. and Zubizarreta, J.~R. (2020).
\newblock Direct and stable weight adjustment in non-experimental studies with
  multivalued treatments: analysis of the effect of an earthquake on
  post-traumatic stress.
\newblock {\em Journal of the Royal Statistical Society: Series A},
  183(4):1387--1410.

\bibitem[Rillig et~al., 2019]{rillig2019role}
Rillig, M.~C., Ryo, M., Lehmann, A., Aguilar-Trigueros, C.~A., Buchert, S.,
  Wulf, A., Iwasaki, A., Roy, J., and Yang, G. (2019).
\newblock The role of multiple global change factors in driving soil functions
  and microbial biodiversity.
\newblock {\em Science}, 366(6467):886--890.

\bibitem[Robins et~al., 2000]{robins2000marginal}
Robins, J.~M., Hernan, M.~A., and Brumback, B. (2000).
\newblock Marginal structural models and causal inference in epidemiology.
\newblock {\em Epidemiology}, 11(5):550--560.

\bibitem[Robins et~al., 1994]{robins1994estimation}
Robins, J.~M., Rotnitzky, A., and Zhao, L. (1994).
\newblock Estimation of regression coefficients when some regressors are not
  always observed.
\newblock {\em Journal of the American Statistical Association},
  89(427):846--866.

\bibitem[Rosenbaum, 2002]{rosenbaum2002covariance}
Rosenbaum, P.~R. (2002).
\newblock Covariance adjustment in randomized experiments and observational
  studies.
\newblock {\em Statistical Science}, 17(3):286--327.

\bibitem[Rosenbaum and Rubin, 1983]{rosenbaum1983central}
Rosenbaum, P.~R. and Rubin, D.~B. (1983).
\newblock The central role of the propensity score in observational studies for
  causal effects.
\newblock {\em Biometrika}, 70(1):41--55.

\bibitem[Rubin, 2008]{rubin2008objective}
Rubin, D.~B. (2008).
\newblock For objective causal inference, design trumps analysis.
\newblock {\em The Annals of Applied Statistics}, 2(3):808--840.

\bibitem[Shu et~al., 2023]{shu2023robust}
Shu, D., Han, P., Hennessy, S., and Miano, T.~A. (2023).
\newblock Robust causal inference of drug-drug interactions.
\newblock {\em Statistics in Medicine}, 42(7):970--992.

\bibitem[Soriano et~al., 2023]{soriano2023interpretable}
Soriano, D., Ben-Michael, E., Bickel, P.~J., Feller, A., and Pimentel, S.~D.
  (2023).
\newblock Interpretable sensitivity analysis for balancing weights.
\newblock {\em Journal of the Royal Statistical Society Series A: Statistics in
  Society}, 186(4):707--721.

\bibitem[Van~der Vaart, 2000]{van2000asymptotic}
Van~der Vaart, A.~W. (2000).
\newblock {\em Asymptotic statistics}.
\newblock Cambridge University Press.

\bibitem[Wang and Zubizarreta, 2020]{wang2020minimal}
Wang, Y. and Zubizarreta, J.~R. (2020).
\newblock Minimal dispersion approximately balancing weights: asymptotic
  properties and practical considerations.
\newblock {\em Biometrika}, 107(1):93--105.

\bibitem[Weichenthal et~al., 2011]{weichenthal2011traffic}
Weichenthal, S., Kulka, R., Dubeau, A., Martin, C., Wang, D., and Dales, R.
  (2011).
\newblock Traffic-related air pollution and acute changes in heart rate
  variability and respiratory function in urban cyclists.
\newblock {\em Environmental Health Perspectives}, 119(10):1373--1378.

\bibitem[Wong and Chan, 2018]{wong2018kernel}
Wong, R.~K. and Chan, K. C.~G. (2018).
\newblock Kernel-based covariate functional balancing for observational
  studies.
\newblock {\em Biometrika}, 105(1):199--213.

\bibitem[Woodruff et~al., 2011]{woodruff2011environmental}
Woodruff, T.~J., Zota, A.~R., and Schwartz, J.~M. (2011).
\newblock Environmental chemicals in pregnant women in the united states:
  Nhanes 2003--2004.
\newblock {\em Environmental Health Perspectives}, 119(6):878--885.

\bibitem[Wu and Hamada, 2021]{wu2011experiments}
Wu, C.~J. and Hamada, M.~S. (2021).
\newblock {\em Experiments: planning, analysis, and optimization}.
\newblock New York: John Wiley \& Sons.

\bibitem[Yu and Wang, 2024]{yu2020treatment}
Yu, R. and Wang, S. (2024).
\newblock Treatment effects estimation by uniform transformer.
\newblock {\em International Conference on Learning Representations}.

\bibitem[Zhao and Ding, 2022]{zhao2022regression}
Zhao, A. and Ding, P. (2022).
\newblock Regression-based causal inference with factorial experiments:
  estimands, model specifications and design-based properties.
\newblock {\em Biometrika}, 109(3):799--815.

\bibitem[Zhao and Ding, 2023]{zhao2023covariate}
Zhao, A. and Ding, P. (2023).
\newblock Covariate adjustment in multiarmed, possibly factorial experiments.
\newblock {\em Journal of the Royal Statistical Society: Series B},
  85(1):1--23.

\bibitem[Zhao et~al., 2018]{zhao2018randomization}
Zhao, A., Ding, P., Mukerjee, R., and Dasgupta, T. (2018).
\newblock Randomization-based causal inference from split-plot designs.
\newblock {\em The Annals of Statistics}, 46(5):1876--1903.

\bibitem[Zhao, 2019]{zhao2019covariate}
Zhao, Q. (2019).
\newblock Covariate balancing propensity score by tailored loss functions.
\newblock {\em The Annals of Statistics}, 47(2):965--993.

\bibitem[Zhao et~al., 2019]{zhao2019sensitivity}
Zhao, Q., Small, D.~S., and Bhattacharya, B.~B. (2019).
\newblock Sensitivity analysis for inverse probability weighting estimators via
  the percentile bootstrap.
\newblock {\em Journal of the Royal Statistical Society Series B: Statistical
  Methodology}, 81(4):735--761.

\bibitem[Zubizarreta, 2015]{zubizarreta2015stable}
Zubizarreta, J.~R. (2015).
\newblock Stable weights that balance covariates for estimation with incomplete
  outcome data.
\newblock {\em Journal of the American Statistical Association},
  110(511):910--922.

\end{thebibliography}
\end{document}